\documentclass[reqno,11pt]{amsart}
\usepackage[utf8]{inputenc}
\usepackage{graphicx}
\usepackage{slashed}
\usepackage{amscd}
\usepackage{amssymb}
\usepackage{esint}
\usepackage{enumitem}
\usepackage[arrow, matrix, curve]{xy}
\usepackage[mathscr]{eucal}
\numberwithin{equation}{section}
\allowdisplaybreaks[1]

\title[Existence of Minimizers in the Infinite-Dimensional Setting]{Causal Variational Principles in the Infinite-Dimensional Setting: Existence of Minimizers}

\author[C.\ Langer]{Christoph Langer \\ \\ January 2021}
\address{Fakult\"at f\"ur Mathematik \\ Universit\"at Regensburg \\ D-93040 Regensburg \\ Germany}
\email{
	christoph.langer@ur.de}

\newtheorem{Def}{Definition}[section]
\newtheorem{Thm}[Def]{Theorem}
\newtheorem{Prp}[Def]{Proposition}
\newtheorem{Lemma}[Def]{Lemma}
\newtheorem{Remark}[Def]{Remark}
\newtheorem{Corollary}[Def]{Corollary}

\newcommand{\Thanks}{\vspace*{.5em} \noindent \thanks}
\newcommand{\beq}{\begin{equation}}
\newcommand{\eeq}{\end{equation}}
\newcommand{\Proof}{\begin{proof}}
	\newcommand{\QED}{\end{proof} \noindent}

\newcommand{\C}{\mathbb{C}}
\newcommand{\R}{\mathbb{R}}
\newcommand{\Id}{\mbox{\rm 1 \hspace{-1.05 em} 1}}

\newcommand{\N}{\mathbb{N}}

\DeclareMathOperator{\tr}{tr}

\renewcommand{\L}{{\mathcal{L}}}
\newcommand{\LL}{{\text{\rm{L}}}}
\newcommand{\Sact}{{\mathcal{S}}}

\newcommand\B{{\mathscr{B}}}

\renewcommand{\H}{\mathscr{H}}

\DeclareMathOperator{\dist}{dist}

\newcommand{\F}{{\mathscr{F}}}

\newcommand{\K}{{\mathscr{K}}}

\newcommand{\D}{\mathscr{D}}

\DeclareMathOperator{\supp}{supp}

\newcommand{\bitem}{\begin{itemize}[leftmargin=2.5em]}
\newcommand{\eitem}{\end{itemize}}

\DeclareFontFamily{OT1}{rsfso}{}
\DeclareFontShape{OT1}{rsfso}{m}{n}{ <-7> rsfso5 <7-10> rsfso7 <10-> rsfso10}{}
\DeclareMathAlphabet{\mycal}{OT1}{rsfso}{m}{n}

\begin{document}
	
\thispagestyle{empty} 

\begin{abstract}
	We provide a method for constructing (possibly non-trivial) measures on non-locally compact Polish subspaces of infinite-dimensional separable Banach spaces which, under suitable assumptions, are minimizers of causal variational principles in the non-locally compact setting. Moreover, for non-trivial minimizers the corresponding Euler-Lagrange equations are derived. The method is to exhaust the underlying Banach space by finite-dimensional subspaces and to prove existence of minimizers of the causal variational principle restricted to these finite-dimensional subsets of the Polish space under suitable assumptions on the Lagrangian. This gives rise to a corresponding sequence of minimizers. Restricting the resulting sequence to countably many compact subsets of the Polish space, by considering the resulting diagonal sequence we are able to construct a regular measure on the Borel algebra over the whole topological space. 
	For continuous Lagrangians of bounded range it can be shown that, under suitable assumptions, the obtained measure is a (possibly non-trivial) minimizer under variations of compact support. Under additional assumptions, we prove that the constructed measure is a minimizer under variations of finite volume and solves the corresponding Euler-Lagrange equations. Afterwards, we extend our results to continuous Lagrangians vanishing in entropy. Finally, assuming that the obtained measure is locally finite, topological properties of spacetime are worked out and a connection to dimension theory is established.  
\end{abstract}

\maketitle

\thispagestyle{empty} 

\tableofcontents

\thispagestyle{empty} 

\section{Introduction}\label{Subsection Overview}
In the physical theory of causal fermion systems, 
spacetime and the structures therein are described by a minimizer
(for an introduction to the physical background and the mathematical context, we refer the interested reader to~\S\ref{Subsection Motivation}, the textbook~\cite{cfs} and the survey articles~\cite{review, dice2014}). 
{\em{Causal variational principles}} evolved as a mathematical generalization
of the causal action principle~\cite{continuum, jet}, and were studied in more detail in \cite{noncompact}.
The starting point in~\cite{noncompact} is a second-countable, locally compact Hausdorff space~$\F$ together with a non-negative function~$\L : \F \times \F \rightarrow \R^+_0 := [0, \infty)$ (the {\em{Lagrangian}}) which is assumed
to be lower semi-continuous, symmetric and positive on the diagonal. 
The causal variational principle is to minimize the {\em{action}}~$\Sact$ defined as
the double integral over the Lagrangian
\[ \Sact (\rho) = \int_\F d\rho(x) \int_\F d\rho(y)\: \L(x,y) \]
under variations of the measure~$\rho$ within the class of regular Borel measures on~$\F$,
keeping the (possibly infinite) total volume~$\rho(\F)$ fixed ({\em{volume constraint}}).
The aim of the present paper is to extend the existence theory for minimizers of such variational principles to the case that~$\F$ is non-locally compact
and the total volume is infinite. 
We also work out the corresponding Euler-Lagrange (EL) equations.

In order to put the paper into the mathematical context,
in~\cite{pfp} it was proposed to
formulate physics by minimizing a new type of variational principle in spacetime.
The suggestion in~\cite[Section~3.5]{pfp} led to the causal action principle
in discrete spacetime, which was first analyzed mathematically in~\cite{discrete}.
A more general and systematic enquiry of causal variational principles on measure spaces was carried out
in~\cite{continuum}. In this article, the existence of minimizers is proven in the case that the total volume is finite. 
In~\cite{jet}, the setting is generalized to non-compact manifolds of possibly infinite volume
and the corresponding EL equations are analyzed. However, the existence of minimizers is not proved. This is done in \cite{noncompact} in the slightly more general setting of second-countable, locally compact Hausdorff spaces. In this paper, we extend the results of \cite{noncompact} by developing the existence theory in the non-locally compact setting. 

The main difficulty in dealing with non-locally compact spaces is that it is no longer possible to restrict attention to compact neighborhoods. Moreover, it turns out that we can no longer assume that the underlying topological space $\F$ is $\sigma$-compact. As a consequence, at first sight it is not clear how to construct global measures on the whole topological space at all. The way out is to introduce a countable collection of suitable compact subsets, which indeed allows us to construct a global measure ${\rho}$ on $\F$. For simplicity, we first assume that the Lagrangian is of \emph{bounded range} (see Definition~\ref{Definition bounded range}). In this case, the minimizing property of the measure $\rho$ is proved in two steps: We first show that $\rho$ is a minimizer under \emph{variations of compact support}. In a second step, we extend this result to variations of finite volume under the assumption that property~(iv) in~\S\ref{seccvpsigma} holds, i.e.\
\[ \sup_{x \in \F} \int_{\F}\L(x,y)\: d\rho(y) < \infty \:. \] 
Afterwards, we generalize our results to Lagrangians which do not have bounded range, but instead have suitable decay properties. To this end, we 
consider Lagrangians \emph{vanishing in entropy} (see Definition \ref{Definition vanish in entropy}). Introducing \emph{spacetime} as the support of the measure $\rho$, we finally analyze topological properties of spacetime; moreover, a connection to dimension theory is established.

The paper is organized as follows. In Section~\ref{secbackground} we give a short physical motivation (\S \ref{Subsection Motivation}) and recall the main definitions and existence results as obtained in~\cite{noncompact} (\S \ref{seccvpsigma}). In Section \ref{Section non-locally compact} we begin by working out important topological properties of infinite-dimensional causal fermion systems (\S \ref{S CFS}); afterwards, we introduce causal variational principles in the non-locally compact (or infinite-dimensional) setting  
by considering  
non-locally compact Polish subspaces of (infinite-dimensional)
separable Banach spaces (\S \ref{S Basic Definitions}). Exhausting the underlying Banach space by finite-dimensional subspaces and making use of the results in~\cite{noncompact}, the existence of minimizers is proven for the causal variational principle restricted to finite-dimensional subspaces (\S \ref{S Approximating}). In Section~\ref{Section Construction Global} we provide a method for constructing a regular global measure on the Borel algebra over the whole non-locally compact Polish space. More precisely, we first introduce a countable collection of compact subsets of the underlying 
topological space (\S \ref{S Construction of Countable Set}). Next, making use of Prohorov's theorem and applying Cantor's diagonal argument, we are able to construct a (possibly non-trivial) regular measure on the whole topological space (\S \ref{S Global Measure}). Finally, we derive useful properties of the constructed measure (\S \ref{S Vague Convergence}). In Section \ref{Section minimizers bounded range} we prove that, under suitable assumptions, this measure is a minimizer for continuous Lagrangians of bounded range (see Definition~\ref{Definition bounded range}). 
More precisely, we 
first introduce an appropriate assumption on the obtained measure (see condition~(B) in~\S\ref{S Preliminaries Condition B}). Next, we prove that the obtained measure is a minimizer under variations of finite-dimensional compact support (\S \ref{S minimizer compact}) as well as a minimizer under variations of compact support (\S \ref{S compact support}). Under additional assumptions we show that the constructed measure is a minimizer under variations of finite volume (\S \ref{S finite volume}). Assuming that the measure under consideration is non-zero, we prove that the corresponding Euler-Lagrange (EL) equations are satisfied (\S \ref{S EL}). The goal in Section~\ref{Section Decay} is to weaken the assumption that the Lagrangian is of bounded range. To this end, we introduce Lagrangians \emph{vanishing in entropy} (see Definition~\ref{Definition entropy new}) which generalize the notion of Lagrangians decaying in entropy (see Definition~\ref{Definition vanishing}). The concept of Lagrangians vanishing in entropy (\S \ref{S Entropy}) can be extended to non-locally compact topological spaces (see Definition~\ref{Definition vanish in entropy}). 
For such Lagrangians, we repeat the above construction steps, thus giving rise to 
a regular measure on the underlying topological space (\S \ref{S Preparatory results}). It is shown that, under suitable assumptions, the considered measure is minimizer of the causal action under variations of compact support as well as under variations of finite volume (\S \ref{S compact support decay}). 
We finally derive the corresponding EL equations (\S \ref{S EL decay}). Introducing spacetime as the support of the minimizing measure under consideration, in Section~\ref{Section topological properties of support} we conclude the paper by analyzing topological properties of spacetime and establishing a connection to dimension theory. To this end, we first recall some concepts from dimension theory (\S \ref{S dimension support}), and afterwards apply them to causal fermion systems (\S \ref{S Application}). In the appendix we summarize useful results which will be referred to frequently: Appendix~\ref{Appendix Polish} is dedicated to the proof that causal fermion systems are Polish (see Theorem~\ref{Theorem Polish}); the main result in Appendix~\ref{Appendix Support} states that the support of locally finite measures on Polish spaces is $\sigma$-compact (see Lemma~\ref{Lemma support LCH}). 

\section{Physical Background and Mathematical Preliminaries}\label{secbackground}
\subsection{Physical Context and Motivation}\label{Subsection Motivation}
The purpose of this subsection is to outline a few concepts of causal fermion systems and
to explain how the present paper fits into the general physical context and the ongoing research program.
The reader not interested in the physical background may skip this section.

The theory of causal fermion systems is a recent approach to fundamental physics motivated originally in order 
to resolve shortcomings of relativistic quantum field theory (QFT).
Namely, due to ultraviolet divergences, perturbative quantum field theory is well-defined only after regularization, which is usually understood as a set of prescriptions for how to make divergent integrals finite (e.g.~by introducing a suitable ``cutoff'' in momentum space). The regularization is then removed using the renormalization procedure. However, this concept is not convincing from neither the physical nor the mathematical point of view. More precisely, in view of Heisenberg's uncertainty principle, physicists infer a correspondence between large momenta and small distances. Because of that, the regularization length is often associated to the Planck length~$\ell_P \approx 1.6 \cdot 10^{-35}\;\text{m}$. Accordingly, by introducing an ultraviolet cutoff in momentum space, one disregards distances which are smaller than the Planck length. As a consequence, the microscopic structure of spacetime is completely unknown. Unfortunately, at present
there is no consensus on what the correct mathematical model for ``Planck scale physics'' should be.

The simplest and maybe most natural approach is to assume that on the Planck scale, spacetime is no longer a continuum but becomes in some way ``discrete.'' This is the starting point in the monograph~\cite{pfp},
where the physical system is described by an ensemble of wave functions in a discrete spacetime.
Motivated by the Lagrangian formulation of classical field theory,
physical equations are formulated by a variational principle in discrete spacetime.
In the meantime, this setting was generalized and developed to the theory of causal fermion systems.
It is an essential feature of the approach that spacetime does not enter the variational principle a-priori,
but instead it emerges when minimizing the action. Thus causal fermion systems allow for the 
description of both discrete and continuous spacetime structures.

In order to get the connection to the present paper, let us briefly outline the main structures of causal fermion systems.
As initially introduced in~\cite{rrev}, a \emph{causal fermion system} consists of a triple~$(\H, \F, \rho)$ together with an integer $n \in \N$, where~$\H$ denotes a complex Hilbert space, $\F \subset \LL(\H)$ being the set of all self-adjoint operators on $\H$ of finite rank with at most $n$ positive and at most $n$ negative eigenvalues, and~$\rho$ being a positive measure on the Borel $\sigma$-algebra~$\B(\F)$ (referred to as \emph{universal measure}). Then for any~$x,y \in \F$, the product~$xy$ is an operator of rank at most~$2n$. Denoting its non-trivial eigenvalues (counting algebraic
multiplicities) by~$\lambda_1^{xy}, \ldots, \lambda_{2n}^{xy} \in \C$, and introducing the spectral weight~$|.|$ of an operator as the sum of the absolute values of its eigenvalues, the \emph{Lagrangian} can be introduced as a mapping
\begin{align*}
\L : \F \times \F \to \R_0^+ \:, \qquad \L(x,y) = \left|(xy)^2\right| - \frac{1}{2n} \left|xy\right|^2 \:.
\end{align*}
As being of relevance for this article, we point out that the Lagrangian is a continuous function which is symmetric in the sense that 
\[\L(x,y) = \L(y,x) \qquad \text{for all $x,y \in \F$} \:. \]
In analogy to classical field theory, one defines the \emph{causal action} by
\begin{align*}
\Sact(\rho) = \iint_{\F \times \F} \L(x,y) \: d\rho(x) \: d\rho(y) \:.
\end{align*}
Finally, the corresponding {\em{causal action principle}} is introduced by varying the measure~$\rho$ in 
the class of regular measures on~$\B(\F)$ under additional constraints
(which assert the existence of non-trivial minimizers). Given a minimizing measure~$\rho$,
{\em{spacetime}}~$M$ is defined as its support,
\[ M := \supp \rho \:. \]
As being outlined in detail in~\cite{cfs}, critical points of the causal action give rise to Euler-Lagrange (EL) equations, which describe the dynamics of the causal fermion system.
In a certain limiting case, the so-called {\em{continuum limit}}, 
one can establish a connection to the conventional formulation of physics in a spacetime continuum.
In this limiting case, the EL equations give rise to classical field equations like the Maxwell and Einstein equations. 
Moreover, quantum mechanics is obtained in a limiting case, and close connections to 
relativistic quantum field theory have been established (for details see~\cite{qft} and~\cite{fockbosonic}).

In order for the causal action principle to be mathematically sensible,
the existence theory is of crucial importance.
If the dimension of the Hilbert space~$\H$ is finite, the existence of minimizers was proven in~\cite[Section~2]{continuum}
(based on existence results in discrete spacetime~\cite{discrete}),
giving rise to minimizing measures $\rho$ on $\F$ of finite total volume~$\rho(\F)<\infty$. 
For this reason, it remains to extend these existence results by developing the existence theory in the case that~$\H$ is infinite-dimensional.
Then the total volume~$\rho(\F)$ is necessarily infinite (for a counter example see~\cite[Exercise~1.3]{cfs}). 
In the resulting {\em{infinite-dimensional setting}} (i.e.~$\dim \H = \infty$ and~$\rho(\F)=\infty$), the task
is to deal with minimizers of infinite total volume on non-locally compact spaces. In preparation, the existence theory of minimizers of possibly infinite total volume $\rho(\F)$ on locally compact spaces is developed  in~\cite{noncompact} in sufficient generality.
The remaining second step, which involves the difficulty of dealing with non-locally compact spaces, is precisely the objective of the present paper. 

\subsection{Causal Variational Principles in the $\sigma$-Locally Compact Setting}\label{seccvpsigma}
Before introducing causal variational principles on non-locally compact spaces in Section \ref{Section non-locally compact} below, we now recall the main results in the less general situation of causal variational principles in the $\sigma$-locally compact setting \cite{noncompact} which are based on results concerning causal variational principles in the non-compact setting as studied in \cite[Section 2]{jet}. 

The starting point in \cite{noncompact} is a second-countable, locally compact topological Hausdorff space~$\F$. 
We let $\rho$ be a (positive) measure on the Borel algebra over~$\F$ (referred to as \emph{universal measure}). Moreover, let~$\L : \F \times \F \rightarrow \R^+_0$ be a non-negative function (the {\em{Lagrangian}}) with the
following properties:
\begin{itemize}[leftmargin=2em]
	\item[(i)] $\L$ is symmetric, i.e.\ $\L(x,y) = \L(y,x)$ for all~$x,y \in \F$. \label{Cond1}%
	\item[(ii)] $\L$ is lower semi-continuous, i.e.\ for all sequences~$x_n \rightarrow x$ and~$y_{n'} \rightarrow y$,
	\label{Cond2}%
	\[ \L(x,y) \leq \liminf_{n,n' \rightarrow \infty} \L(x_n, y_{n'})\:. \]
\end{itemize}
The {\em{causal variational principle}} is to minimize the action
\beq \label{Sact} 
\Sact (\rho) = \int_\F d\rho(x) \int_\F d\rho(y)\: \L(x,y) 
\eeq
under variations of the measure~$\rho$, keeping the total volume~$\rho(\F)$ fixed
({\em{volume constraint}}). The papers \cite{jet, noncompact} mainly focus on the case that the total volume~$\rho(\F)$ is
infinite. In order to implement the volume constraint and to derive the corresponding Euler-Lagrange equations, in~\cite{jet} one makes the following additional assumptions:
\begin{itemize}[leftmargin=2em]
	\item[(iii)] The measure~$\rho$ is {\em{locally finite}}
	(meaning that any~$x \in \F$ has an open neighborhood~$U \subset \F$ with~$\rho(U)< \infty$). \label{Cond3}%
	\item[(iv)] The function~$\L(x,.)$ is $\rho$-integrable for all~$x \in \F$ and \label{Cond4}%
	\beq \label{Lint}
	\sup_{x \in \F} \int_{\F}\L(x,y)\: d\rho(y) < \infty \:.
	\eeq
\end{itemize}
By Fatou's lemma, the integral in~\eqref{Lint} is lower semi-continuous in the variable~$x$. 

A measure on the Borel algebra which satisfies~(iii) will be referred to as a {\em{Borel measure}} (in the sense of \cite{gardner+pfeffer}), and the set of Borel measures on~$\F$ shall be denoted by~$\mathfrak{B}_{\F}$.  
Moreover, the Borel $\sigma$-algebra over $\F$ is denoted by $\B(\F)$. 
A Borel measure is said to be {\em{regular}} if it is inner and outer regular (cf.~\cite[Definition VIII.1.1]{elstrodt}). An inner regular Borel measure is also called a \emph{Radon measure} \cite{schwartz-radon}. 

In \cite{jet, noncompact} one varies in the following class of measures:
\begin{Def} \label{deffv}
	Given a regular Borel measure~$\rho$ on~$\F$, a regular Borel measure~$\tilde{\rho}$ on~$\F$
	is said to be a {\bf{variation of finite volume}} if 
	\beq \label{totvol}
	\big| \tilde{\rho} - \rho \big|(\F) < \infty \qquad \text{and} \qquad
	\big( \tilde{\rho} - \rho \big) (\F) = 0 \:,
	\eeq
	where the total variation $|\tilde{\rho} - \rho|$ of two possibly infinite measures $\rho$ and $\tilde{\rho}$ on $\B(\F)$ is defined in~\cite[\S 2.2]{noncompact} as follows: We say that $|\tilde{\rho} - \rho| < \infty$ if there exists~$B \in \B(\F)$ with~$\rho(B), \tilde{\rho}(B) < \infty$ such that $\rho|_{\F \setminus B} = \tilde{\rho}|_{\F \setminus B}$. In this case, 
	\begin{align*}
	(\tilde{\rho} - \rho)(\Omega) := \tilde{\rho}(B \cap \Omega) - {\rho}(B \cap \Omega)
	\end{align*}
	for any Borel set~$\Omega \subset \F$. 
\end{Def} 

Given a regular Borel measure $\rho \in \mathfrak{B}_{\F}$ and assuming that (i), (ii) and~(iv) hold, for every variation of finite volume~$\tilde{\rho} \in \mathfrak{B}_{\F}$ the difference of the actions as given by
\beq \label{integrals}
\begin{split}
	\big( &\Sact(\tilde{\rho}) - \Sact(\rho) \big) = \int_\F d(\tilde{\rho} - \rho)(x) \int_\F d\rho(y)\: \L(x,y) \\
	&\quad + \int_\F d\rho(x) \int_\F d(\tilde{\rho} - \rho)(y)\: \L(x,y) 
	+ \int_\F d(\tilde{\rho} - \rho)(x) \int_\F d(\tilde{\rho} - \rho)(y)\: \L(x,y)
\end{split}
\eeq
is well-defined in view of~\cite[Lemma~2.1]{jet}. For clarity, we point out that condition~(iii) is not required in order for~\eqref{integrals} to hold. 

Note that the assumptions~(iii) and~(iv) are dropped in \cite{noncompact}. 
The causal variational principle in the $\sigma$-locally compact setting~\cite{noncompact} is then defined as follows. 
\begin{Def} Let~$\F$ be a second-countable, locally compact Hausdorff space,
	and let the Lagrangian~$\L : \F \times \F \rightarrow \R^+_0$ be a symmetric and lower semi-continuous function
	(see conditions~{\rm{(i)}} and~{\rm{(ii)}} above). Moreover, we assume that~$\L$
	is strictly positive on the diagonal, i.e.\
	\begin{align*}
		\L(x,x) > 0 \qquad \text{for all~$x \in \F$} \:.
	\end{align*}
	The {\bf{causal variational principle on $\sigma$-locally compact spaces}} is to minimize
	the causal action~\eqref{Sact}
	under variations of finite volume (see Definition~\ref{deffv}).
\end{Def} \noindent

We point out that (iv) is a sufficient condition for \eqref{integrals} to hold. 
However, since the conditions~(iii) and~(iv) are not imposed in \cite{noncompact},
it is a-priori not clear whether the integrals in~\eqref{integrals} exist.
For this reason, condition~\eqref{integrals} is included into the definition of a minimizer:

\begin{Def} \label{Definition minimizer} A regular Borel measure~$\rho$ on~$\F$ is said to be a {\bf{minimizer}} of the causal action \textbf{under variations of finite volume}~\cite{noncompact} if the difference~\eqref{integrals} is well-defined and
	non-negative for all regular Borel measures~$\tilde{\rho}$ on~$\F$ satisfying~\eqref{totvol},
	\[ \big( \Sact(\tilde{\rho}) - \Sact(\rho) \big) \geq 0 \:. \]
\end{Def}

We denote the support of the measure~$\rho$ by~$M$,
\beq \label{suppdef}
M := \supp \rho = \F \setminus \bigcup \big\{ \text{$\Omega \subset \F$ \,\big|\,
	$\Omega$ is open and $\rho(\Omega)=0$} \big\}
\eeq
(thus the support is the set of all points for which every open neighborhood
has a strictly positive measure; for details and generalizations
see~\cite[Subsection~2.2.1]{federer}). According to Definition \ref{deffv}, the condition~$|\tilde{\rho} - {\rho}| < \infty$ implies that there exists some Borel set~$B \subset \F$ with $\rho(B), \tilde{\rho}(B) < \infty$ and $\rho|_{\F \setminus B} = \tilde{\rho}|_{\F \setminus B}$. In particular, $\rho|_B$ and~$\tilde{\rho}|_B$ are finite Borel measures on $\B(B)$, and thus have support (see \cite[Proposition~7.2.9]{bogachev}). Furthermore, the signed measure $\tilde{\rho} - \rho$ has support. 

We now recall some results from \cite{noncompact} which will be referred to frequently. 
The first existence results in \cite{noncompact} are based on the assumption that the Lagrangian is of compact range. For convenience, let us state the definition.
\begin{Def} \label{defcompactrange}
	The Lagrangian has {\bf{compact range}} if for every compact set~$K \subset \F$ there is
	a compact set~$K' \subset \F$ such that
	\[ \L(x,y) = 0 \qquad \text{for all~$x \in K$ and~$y \not \in K'$}\:. \]
\end{Def}

Moreover, the definition of minimizers under variations of compact support plays an important role in~\cite{noncompact}. 
\begin{Def} \label{Def global}
	A regular Borel measure $\rho$ on $\F$ is said to be a {\bf{minimizer under variations of compact support}}~\cite{noncompact} 
	of the causal action if for any regular Borel measure~$\tilde{\rho}$ on~$\F$ which satisfies~\eqref{totvol} such that the signed measure~$\tilde{\rho} - \rho$ is compactly supported, the inequality 
	\[ \big( \Sact(\tilde{\rho}) - \Sact(\rho) \big) \geq 0 \]
	holds.
\end{Def}

Assuming that the Lagrangian is of compact range, the main results in \cite{noncompact} can be summarized as follows. 

\begin{Thm}[\textbf{Euler-Lagrange equations}] \label{Theorem 4.2}
	Let $\F$ be a second-countable, locally compact Hausdorff space, and 
	assume that~$\L : \F \times \F \to \R_0^+$ is continuous and of compact range.
	Then there exists a regular Borel measure~$\rho$ on~$\F$ which satisfies the Euler-Lagrange equations
	\begin{align}\label{(EL equations)}
	\ell|_{\supp \rho} \equiv \inf_{\text{$x \in \F$}} \ell(x) = 0 \:,
	\end{align}
	where~$\ell \in C(\F)$ is defined by
	\begin{align}\label{(Lagrange functional)}
	\ell (x) := \int_{\F} \L(x,y) \:d\rho(y) - 1 \:.
	\end{align}
\end{Thm}

Combining \cite[Theorem 4.9 and Theorem 4.10]{noncompact}, we obtain the following result.

\begin{Thm}\label{Theorem 4.9 + 4.10}
	Assume that~$\L : \F \times \F \to \R_0^+$ is continuous and of compact range. Then there is a regular Borel measure~$\rho$ on~$\F$ which is a minimizer under variations of compact support~\cite{noncompact} (see Definition \ref{Def global}). Under the additional assumptions that the Lagrangian~$\L$ is bounded and condition~{\rm{(iv)}} is satisfied (see \eqref{Lint}), the measure~$\rho$ is a minimizer under variations of finite volume~\cite{noncompact} (see Definition~\ref{Definition minimizer}).
\end{Thm}

In \cite[Section 5]{noncompact} it was shown that the assumption that the Lagrangian~$\L$ is of compact range can be weakened. To this end, we recall that every second-countable, locally compact Hausdorff space can be endowed with a Heine-Borel metric (for details we refer to the explanations in \cite[\S 3.1 and \S 5.1]{noncompact}). Given an unbounded Heine-Borel metric on the second-countable, locally compact space $\F$, for any $r > 0$ and $x \in \F$ the closed ball~$\overline{B_r(x)}$ is compact, and hence can be covered by finitely many balls of radius~$\delta > 0$. The smallest such number is denoted by $E_x(r,\delta)$ and is called \emph{entropy}. This gives rise to Lagrangians decaying in entropy, being defined as follows (cf.~\cite[Definition 5.1]{noncompact}). 
\begin{Def}\label{Definition vanishing}
	Assume that $\F$ is endowed with an unbounded Heine-Borel metric~$d$. 
	The
	Lagrangian~$\L \colon \F \times \F \to \R_0^+$ is said to {\bf{decay in entropy}}
	if the following conditions are satisfied: 
	\begin{enumerate}[leftmargin=2em]
		\item[\rm{(a)}] $c:=\inf_{{x} \in \F} \L({x},{x}) > 0$.
		\item[\rm{(b)}] There is a compact set $K \subset \F$ such that 
		\[ \delta := \inf_{{x} \in \F \setminus K} \sup \left\{ s \in \R \::\: \L({x},y) \ge \frac{c}{2} \quad \text{for all~$y \in B_s(x)$} \right\} > 0 \:. \]
		\item[\rm{(c)}] The Lagrangian has the following decay property:
		There is a monotonically decreasing, integrable function~$f \in L^1(\R^+, \R^+_0)$ such that
		\begin{align*}
		\L(x,y) \leq \frac{f \big(d(x,y) \big)}{C_x\big(d(x,y), \delta \big)} \qquad \text{for all~$x,y \in \F$ with~$x \neq y$} \:,
		\end{align*}
		where
		\[C_x(r, \delta) := C \: E_x(r+2, \delta) \qquad \text{for all $r > 0$}\:, \]
		and the constant $C$ is given by
		\[ C := 1+\frac{2}{c} < \infty \:. \]
	\end{enumerate}
\end{Def}\noindent
We point out that
the above definition of Lagrangians decaying in entropy as introduced in \cite[Section 5]{noncompact} requires an \emph{unbounded} Heine-Borel metric. For a more general definition we refer to~\S \ref{S Entropy} (see Definition~\ref{Definition entropy new}). 

For clarity we note that, if $(\H, \F, \rho)$ is a causal fermion system with $\dim(\H) < \infty$, the space $\LL(\H)$ of bounded linear operators on~$\H$ is finite-dimensional. Combining the fact that all norms on finite-dimensional vector spaces are equivalent with the Heine-Borel theorem yields that the Fréchet metric induced by the operator norm on the vector space~$\LL(\H)$ is an unbounded Heine-Borel metric on~$\F$.   

Let us now recall the main results in \cite{noncompact} under the assumption that the Lagrangian decays in entropy. 
\begin{Thm}\label{Theorem 5.5}
	Let $\F$ be a second-countable, locally compact Hausdorff space, and assume that $\L : \F \times \F \to \R_0^+$ is continuous and decays in entropy. Then there exists a regular Borel measure~$\rho$ on $\F$ which satisfies the Euler-Lagrange equations
	\begin{align*}
	\ell|_{\supp \rho} \equiv \inf_{x \in \F} \ell(x) = 0 \:,
	\end{align*}
	where $\ell \in C(\F)$ is defined by \eqref{(Lagrange functional)}. 
\end{Thm}

The following theorem ensures the existence of minimizing Borel measures. 
\begin{Thm}\label{Theorem 5.8 + 5.9}
	Assume that $\L : \F \times \F \to \R_0^+$ is continuous and decays in entropy. Then there is a regular Borel measure~$\rho$ on~$\F$ which is a minimizer under variations of compact support~\cite{noncompact}. Under the additional assumptions that the Lagrangian $\L$ is bounded and condition~{\rm{(iv)}} is satisfied (see \eqref{Lint}), the measure~$\rho$ is a minimizer under variations of finite volume~\cite{noncompact}.
\end{Thm}

The goal of this paper is to extend the above results to the infinite-dimensional setting.

\section{Causal Variational Principles in the Non-Locally Compact Setting}\label{Section non-locally compact}
\subsection{Motivation: Infinite-Dimensional Causal Fermion Systems}\label{S CFS}
As explained in Section \ref{Subsection Overview}, causal variational principles evolved as a mathematical generalization of the causal action principle in order to develop the existence theory for causal fermion systems. In order to point out the connection to causal variational principles in the non-locally compact setting, let us briefly recall the basic structures of causal fermion systems (for details cf.~\cite[\S 1.1.1]{cfs}). By definition, causal fermion systems are characterized by a separable complex Hilbert space~$\H$, an integer~$s \in \N$ (the so-called \emph{spin dimension}) and a measure~$\rho$ defined on the Borel $\sigma$-algebra $\B(\F)$, where~$\F \subset \LL(\H)$ consists of all self-adjoint operators on $\H$ which have at most~$s$ positive and at most~$s$ negative eigenvalues. This gives rise to a triple~$(\H, \F, \rho)$. 
The set~$\F$ can be endowed with the topology induced by the operator norm on $\LL(\H)$, thus becoming a topological space. More precisely, denoting the Fréchet metric induced by the operator norm on~$\LL(\H)$ by~${d}$, the space~$(\F, d)$ is a separable complete metric space (Theorem~\ref{Theorem Polish}). 

Whenever $\dim(\H) < \infty$, the topological space $\F \subset \LL(\H)$ is locally compact. On the contrary, whenever~$\H$ is an infinite-dimensional Hilbert space, the corresponding set~$\F \subset \LL(\H)$ is non-locally compact (see Lemma \ref{Lemma Freg} below). In preparation, let us first state the following results.

\begin{Prp}\label{Proposition Banach}
	Any locally compact Banach space $X$ is finite-dimensional. 
\end{Prp}
\begin{proof}
	Let $x_1 \in X$ with $\|x_1\| = 1$. Given $x_1, \ldots, x_r \in X$ linearly independent unit vectors (i.e.~$\|x_i\| = 1$ for all $i =1, \ldots, r$), the space~$G_r = \operatorname{span}\{x_1, \ldots, x_r \}$ is an $r$-dimensional subspace of~$X$. Since $G_r$ is finite-dimensional, it is closed. If $G_r \subsetneq E$, there exists a unit vector~$x_{r+1} \in X$ with $\|x_{r+1} - x_i\| \ge 1/2$ for all $i = 1, \ldots, r$. 
	
	If we assume that $X$ is infinite-dimensional, this holds for every $r \in \N$, thus ending up with an infinite sequence $(x_r)_{r \in \N}$ of unit vectors satisfying~$\|x_p - x_q\| \ge 1/2$ for each $p \not= q$. In particular, the sequence $(x_r)_{r \in \N}$ admits no convergent subsequence in contradiction to the assumption that $E$ is locally compact. 
\end{proof}

\begin{Corollary}\label{Corollary non-locally compact}
	Any infinite-dimensional Banach space~$X$ is non-locally compact. The same holds true for open subsets of~$X$. 
\end{Corollary}
\begin{proof}
	This is an immediate consequence of Proposition \ref{Proposition Banach}. 
\end{proof}

\begin{Lemma}\label{Lemma Freg}
	Let $\H$ be an infinite-dimensional, separable complex Hilbert space, and let $\F^{\textup{reg}} \subset \LL(\H)$ be the set of self-adjoint operators which have exactly $s$ positive and exactly $s$ negative eigenvalues for some $s \in \N$. Then $\F^{\textup{reg}}$ is non-locally compact. 
\end{Lemma}
\begin{proof}
	Since $\F^{\textup{reg}}$ is a Banach manifold (for details see~\cite{finster+lottner}), 
	it can be covered by an atlas~$(U_{\alpha}, \phi_{\alpha})_{\alpha \in A}$ for some index set~$A$ (cf.~\cite[Chapter~73]{zeidlerIV}). In particular, every point~$x \in U^{\textup{reg}}$ is contained in some open set $U_{\alpha}$, whose image $V_{\alpha} := \phi_{\alpha}(U_{\alpha})$ is open in some infinite-dimensional Banach space $X_{\alpha}$. Due to Corollary \ref{Corollary non-locally compact}, the set~$V_{\alpha}$ is non-locally compact. As the mapping $\phi_{\alpha}$ is a homeomorphism, we deduce that $U_{\alpha} \subset \F^{\textup{reg}}$ is non-locally compact for any $\alpha \in A$, which proves the claim. 
\end{proof}

Considering an \emph{infinite-dimensional}, separable complex Hilbert space $\H$, then the set~$\F \subset \LL(\H)$ as introduced in \cite{cfs} is non-locally compact and Polish (see Lemma~\ref{Lemma Freg} and Theorem~\ref{Theorem Polish}). Our goal in the following is to prove the existence of a regular (possibly non-locally finite) measure $\rho$ on the Borel algebra~$\B(\F)$ such that~$\rho$ is a minimizer of the corresponding causal action principle, giving rise to a causal fermion system~$(\H, \F, \rho)$. 
Instead of immediately delving into the corresponding causal action principle (see \cite[\S 1.1.1]{cfs}), we deal with causal variational principles on~$\H$, which can be viewed as generalizations of the causal action principle (as introduced in \cite{continuum}, \cite{jet} and considered in more detail in \cite{noncompact}). Corresponding results concerning the causal action principle are then obtained as a special case. With this in mind, it suffices to prove the existence of minimizers of the causal variational principle \eqref{(cvp)} under the constraints~\eqref{totvol} in the non-locally compact setting as introduced in Definition~\ref{Definition non-locally compact}.

In order to motivate the basic definitions in \S \ref{S Basic Definitions} below, we note that~$\F \subset \mathscr{K}(\H)$, where by~$\mathscr{K}(\H) \subset \LL(\H)$ we denote the set of compact operators on~$\H$. 
Since $\H$ is a separable, infinite-dimensional complex Hilbert space, let us point out that~$\mathscr{K}(\H)$ is a Banach space (see e.g.~\cite[Satz~II.3.2]{werner}) and separable in view of \cite[\S 12.E]{kechris}. 
Moreover, making use of Proposition \ref{Proposition Banach}, Corollary~\ref{Corollary non-locally compact} and Lemma~\ref{Lemma Freg}, we conclude that~$\mathscr{K}(\H) \supset \F$ is infinite-dimensional. This allows us to approximate~$\mathscr{K}(\H)$ by finite-dimensional subspaces. 
More precisely, we may apply \cite[Lemma 7.1]{alt} to deduce that there is a sequence of finite-dimensional subspaces~$(L_n)_{n \in \N}$ in $\mathscr{K}(\H)$ with $L_n \subset L_{n+1}$ for all~$n \in \N$ such that~$\bigcup_{n \in \N} L_n$ is dense in $\mathscr{K}(\H)$. From the fact that subspaces of locally compact spaces are again locally compact we conclude that~$\F^{(n)} := \F \cap L_n$ is locally compact for every~$n \in \N$. 
Denoting by~$d$ the Fréchet metric induced by the operator norm on $\LL(\H)$, the space $(\F, d)$ is Polish (cf.\ Theorem \ref{Theorem Polish}), i.e.\ a separable metric space. As a consequence, the subsets~$\F^{(n)}$ are separable for every~$n \in \N$ due to~\cite[Lemma~2.16]{alt} or~\cite[Corollary 3.5]{aliprantis}. Together with the fact that separable metric spaces are second-countable this yields that the set~$\F^{(n)}$ is a second-countable, locally compact Hausdorff space for every~$n \in \N$. 
Moreover, from Lemma \ref{Lemma Freg} we conclude that~$\F \subset \LL(\H)$ is non-locally compact.

In order to treat the corresponding causal variational principle in sufficient generality, it seems reasonable to vary in the class of regular, not necessarily locally finite measures on the Borel $\sigma$-algebra~$\B(\F)$ (as intended in the textbook~\cite[\S 1.1.1]{cfs}). 
As mentioned in~\cite{cfs}, the causal action principle is ill-posed if the total volume~$\rho(\F)$ is finite and the Hilbert space~$\H$ is infinite-dimensional. However, the causal action principle does make mathematical sense in the so-called \emph{infinite-dimensional setting} where~$\H$ is infinite-dimensional and the total volume is infinite, i.e.~$\rho(\F) = \infty$. These considerations motivate causal variational principles in the infinite-dimensional (or non-locally compact) setting as defined in the next subsection.

\subsection{Basic Definitions}\label{S Basic Definitions}
Let us first state the causal variational principle in the non-locally compact setting and discuss its difficulties afterwards.
\begin{Def}\label{Definition non-locally compact}
	Assume that $X$ is a separable, infinite-dimensional Banach space, and let~$\F \subset X$ be a 
	\emph{non-locally compact} Polish space.
	Moreover, assume that the Lagrangian~$\L : \F \times \F \rightarrow \R^+_0$ is a symmetric and lower semi-continuous function
	(see conditions~{\rm{(i)}} and~{\rm{(ii)}} in~\S\ref{seccvpsigma}) which
	is strictly positive on the diagonal, i.e.\
	\beq \label{strictpositive}
	\L(x,x) > 0 \qquad \text{for all~$x \in \F$} \:.
	\eeq
	The {\bf{causal variational principle in the non-locally compact setting}}\footnote{For clarity we point out that ``{causal variational principles in the non-locally compact setting}'' and ``{causal variational principles in the infinite-dimensional setting}'' are used synonymously.} is to 
	\begin{align}\label{(cvp)}
	\text{minimize} \qquad \Sact(\rho) := \int_{\F} \int_{\F} d\rho(x) \: d\rho(y) \: \L(x,y)
	\end{align}
	under variations of finite volume (see Definition~\ref{Definition finite volume} below) in the class of all regular measures on $\B(\F)$ (in the sense of~\cite{gardner+pfeffer}, cf.~\cite{noncompact}) with $\rho(\F) = \infty$.
\end{Def}\noindent
The condition~\eqref{strictpositive} is needed in order to avoid trivial minimizers supported
at~$x \in \F$ with~$\L(x,x)=0$ (see~\cite[Section~1.2]{support}). Furthermore, condition~\eqref{strictpositive} is
a plausible assumption in view of~\cite[Exercise~1.2]{cfs}. Namely, given a minimizing measure~$\rho$ of the causal action principle~\eqref{(cvp)}, there exists a real constant~$c$ such that~$\tr(x) = c$ for all~$x \in \supp \rho$ according to \cite[Proposition 1.4.1]{cfs}. Under the reasonable assumption that $c \not= 0$ (cf.~\cite[\S 1.4.1]{cfs}), we may conclude that $\L(x,x) > 0$ for all $x \in \F$ in view of~\cite[Exercise~1.2]{cfs}. This motivates as well as justifies the assumption that the Lagrangian is strictly positive on the diagonal.

Dropping the assumption that the measures under consideration are locally finite, we slightly adapt the definition of a minimizer of the causal action as follows.

\begin{Def} \label{Definition finite volume} 
	A regular measure~$\rho$ on~$\B(\F)$ is said to be a {\bf{minimizer}} of the causal action \textbf{under variations of finite volume} if the difference~\eqref{integrals} is well-defined and
	non-negative for all regular measures~$\tilde{\rho}$ on~$\B(\F)$ satisfying~\eqref{totvol},
	\[ \big( \Sact(\tilde{\rho}) - \Sact(\rho) \big) \geq 0 \:. \]
\end{Def}\noindent
Given a measure~$\rho$ on the Borel algebra~$\B(\F)$ and assuming that $\tilde{\rho}$ is a variation of finite volume, in view of Definition~\ref{Definition finite volume} there exists~$B \in \B(\F)$ such that $\rho|_{\F \setminus B} = \tilde{\rho}|_{\F \setminus B}$. In particular, the measures~$\rho|_B$ and~$\tilde{\rho}|_B$ are finite. Henceforth, whenever $\rho$ is locally finite, then the same holds true for the measure~$\tilde{\rho}|_{\F \setminus B}$. From the fact that $\tilde{\rho}|_B$ is a finite measure we conclude that $\tilde{\rho}|_B$ is locally finite. Consequently, the measure~$\tilde{\rho}$ is locally finite if $\rho$ is so. 
For this reason, Definition \ref{Definition finite volume} can be viewed as a generalization of Definition~\ref{Definition minimizer} (cf.~\cite[Definition~2.1]{noncompact}). 
The same holds for Definition~\ref{Definition compact support} below.

\begin{Def} \label{Definition compact support} 
	A regular measure~$\rho$ on~$\B(\F)$ is said to be a {\bf{minimizer under variations of compact support}}
	of the causal action if for any regular measure~$\tilde{\rho}$ on~$\B(\F)$ which satisfies~\eqref{totvol} such that the signed measure~$\tilde{\rho} - \rho$ is compactly supported, the inequality 
	\[ \big( \Sact(\tilde{\rho}) - \Sact(\rho) \big) \geq 0 \]
	holds.
\end{Def}

Let us now point out some difficulties regarding causal variational principles on non-locally compact spaces. First of all, let us recall that a topological space is called hemicompact if there is a sequence $(K_n)_{n \in \N}$ of compact subsets of $X$ such that any compact set~$K \subset X$ is contained in $K_n$ for some $n \in \N$ (see \cite[17I]{willard}). Since $\F$ is first-countable and non-locally compact, by virtue of \cite[Exercise~3.4.E]{engelking} we conclude that~$\F$ cannot be hemicompact. 

Next, by contrast to the $\sigma$-locally compact setting as worked out in~\cite{noncompact}, it is in general not even possible to assume that~$\F$ is $\sigma$-compact, as the following argument shows: Every Polish space (as well as every locally compact Hausdorff space) is Baire according to~\cite[Theorem~(8.4)]{kechris}.\footnote{For clarity, we recall that a topological space~$X$ is said to be \emph{Baire} if the intersection of each countable family of dense open sets in~$X$ is dense (see e.g.~\cite[Definition 25.1]{willard}).} In view of \cite[25B]{willard}, a $\sigma$-compact topological space~$X$ is Baire if and only if the set of points at which~$X$ is locally compact is dense in~$X$. 
Given an infinite-dimensional Hilbert space~$\H$, and defining~$\F \subset \mathscr{K}(\H)$ in analogy to~\cite{cfs} (see~\S \ref{S CFS}), then~$\F$ is a Polish space (see Appendix~\ref{Appendix Polish}). Consequently, 
the assumption that $\F$ is $\sigma$-compact implies that
there exists~$x \in \F$ being contained in a compact neighborhood~$K \subset \F$ with~$K^{\circ} \not= \varnothing$. From this we conclude that the intersection~$K^{\textup{reg}} := K \cap \F^{\textup{reg}}$ is a compact set with non-empty interior, where the Banach manifold~$\F^{\textup{reg}} \subset \F$ is defined in Lemma~\ref{Lemma Freg} (for details see~\cite{finster+lottner}). Given an atlas~$(U_{\alpha}, \phi_{\alpha})_{\alpha \in A}$ of~$\F^{\textup{reg}}$ (cf.~\cite{zeidlerIV}) 
and making use of the fact that each $\phi_{\alpha}$ is a homeomophism mapping to some infinite-dimensional Banach space~$X_{\alpha}$, we deduce that the image of $K^{\textup{reg}}$ is a compact subset with non-empty interior in contradiction to \cite[Exercise 14.3]{koenig}. For this reason, it is not possible to assume that the space~$\F$ is~$\sigma$-compact (by contrast to the setting in~\cite{noncompact}).

Next, it is no longer possible to assume that the Lagrangian $\L : \F \times \F \to \R_0^+$ is simultaneously lower semi-continuous and of compact range (see Definition~\ref{defcompactrange}) as introduced in~\cite[Definition~3.3]{noncompact}. Namely, due to lower semi-continuity of the Lagrangian, the latter assumption already implies that $\F$ is locally compact. 

Finally, it is not possible to assume that the Lagrangian decays in entropy in the sense of \cite[Definition 5.1]{noncompact} (see Definition \ref{Definition vanishing}); 
indeed, this assumption requires a Heine-Borel metric on $\F$, which clearly does not exist in non-locally compact spaces (otherwise each $x \in \F$ is contained in a corresponding ball with compact closure). 

In view of these difficulties in the non-locally compact setting, let us begin by generalizing the assumption that $\L$ is of compact range in the following way.
\begin{Def}\label{Definition bounded range}
	Let $(\F, d)$ be a metric space.
	The Lagrangian $\L : \F \times \F \to \R_0^+$ is said to be \textbf{of bounded range} 
	if every bounded set $B \subset \F$ is contained in a bounded neighborhood $B' \subset \F$ such that 
	\begin{align*}
	\L(x,y) = 0 \qquad \text{for all $x \in B$ and $y \notin B'$} \:.
	\end{align*}
\end{Def}

On proper metric spaces (that is, on spaces with the Heine-Borel property), this definition clearly implies that $\L$ is of compact range (see Definition~\ref{defcompactrange}) as defined in~\cite{noncompact}. 
For this reason, Definition~\ref{Definition bounded range} provides a good starting point for dealing with causal variational principles on non-locally compact spaces. As we shall see below, the assumption that the Lagrangian is of bounded range can be weakened (see \S \ref{S Entropy}).

\subsection{Finite-Dimensional Approximation}\label{S Approximating}
In the infinite-dimensional setting (see Definition \ref{Definition non-locally compact}), the space~$X$ is assumed to be a separable, infinite-dimensional Banach space. Hence we may apply~\cite[Lemma 7.1]{alt} to deduce that there exists a sequence of finite-dimensional subspaces~$(X_n)_{n \in \N}$ in $X$ with $X_n \subset X_{n+1}$ for all~$n \in \N$ such that~$\bigcup_{n \in \N} X_n$ is dense in~$X$. This allows us to  introduce the topological spaces
\begin{align*}
\F^{(n)} := \F \cap X_n \qquad \text{for every $n \in \N$} \:.
\end{align*}
Since finite-dimensional topological vector spaces are locally compact (see e.g.~\cite[\S 15.7]{koethe}), we conclude that each~$X_n \subset X$ is locally compact for all $n \in \N$. 
For ease in notation, we shall denote the restriction of the Lagrangian to~$\F^{(n)} \times \F^{(n)}$ by~$\L^{(n)}$. Thus for every~$n \in \N$, we are given a second-countable, locally compact Hausdorff space~$\F^{(n)} \subset \F$ together with a symmetric, lower semi-continuous Lagrangian 
\begin{align*}
\L^{(n)} : \F^{(n)} \times \F^{(n)} \to \R_0^+ \:,
\end{align*}
which is strictly positive on the diagonal. Henceforth for every~$n \in \N$ we are exactly in the $\sigma$-locally compact setting as worked out in \cite{noncompact}.

In the following, we \emph{additionally assume that} $\L : \F \times \F \to \R_0^+$ is continuous and of bounded range (see Definition \ref{Definition bounded range}). We again consider the above exhaustion~$(X_n)_{n \in\N}$ of $X$ by finite-dimensional subsets with $X_n \subset X_{n+1}$ for all~$n \in \N$. Let us point out that each~$(X_n, \|.\|)$ is a finite-dimensional normed vector space, and all norms on~$X_n$ are equivalent. Due to the Heine-Borel theorem~\cite[Bemerkungen~2.6 and Satz~2.9]{alt}, each closed ball $\overline{B_r(x)} \subset X_n$ is compact for all~$r > 0$ and~$x \in X_n$. As a consequence, each bounded set~$A \subset \F^{(n)}$ is contained in some compact ball~$B:= \overline{B_r(x)} \subset \F^{(n)}$. Definition~\ref{Definition bounded range} yields the existence of a compact set~$B' \subset \F^{(n)}$ such that~$\L^{(n)}(x,y) = 0$ for all~$x \in B$ and $y \notin B'$. These considerations show that, whenever~$\L$ is continuous and of bounded range, for every $n \in \N$ the restricted Lagrangian~$\L^{(n)}$ is continuous and of compact range (see \cite[Definition~3.3]{noncompact} or Definition~\ref{defcompactrange}). As a consequence, by virtue of Theorem~\ref{Theorem 4.2} (also see \cite[Theorem 4.2]{noncompact}), for each~$n \in \N$ there exists a regular Borel measure~${\rho}_n$ on $\F^{(n)}$ such that the following EL equations hold, 
\begin{align}\label{(ELn)}
{\ell}_n|_{\supp {\rho}_n} \equiv \inf_{x \in \F^{(n)}} {\ell}_n(x) = 0 \:,
\end{align}
where ${\ell}_n \in C(\F) = C(\F, \R)$ is defined by
\begin{align}\label{(elln)}
{\ell}_n(x) := \int_{\F^{(n)}} \L^{(n)}(x,y) \: d{\rho}_n(y) -1 \:. 
\end{align}
According to Theorem \ref{Theorem 4.9 + 4.10} (cf.~\cite[Theorem 4.10]{noncompact}), each Borel measure~${\rho}_n \in \mathfrak{B}_{\F^{(n)}}$ is a minimizer of the corresponding causal variational principle 
\begin{align*}
\textup{minimize} \qquad \Sact^{(n)} := \int_{\F^{(n)}} \int_{\F^{(n)}} \L^{(n)} \:d\rho(x) \: d\rho(y)
\end{align*}
under variations of compact support~\cite{noncompact} in the class of regular Borel measures on $\F^{(n)}$ with respect to the constraints \eqref{totvol}.

We extend the measures ${\rho}_n$ by zero on the whole topological space $\F$, 
\begin{align}\label{(rho[n])}
{\rho}^{[n]} (A) := {\rho}_n(A \cap \F^{(n)}) \qquad \text{for all $A \in \B(\F)$} \:.
\end{align}
Thus 
\begin{align}\label{(ELn')}
{\ell}^{[n]}|_{\supp {\rho}^{[n]}} \equiv \inf_{x \in \F^{(n)}} {\ell}^{[n]}(x) = 0 \:,
\end{align}
where the function ${\ell}^{[n]} \in C(\F)$ is defined by
\begin{align}\label{(elln')}
{\ell}^{[n]} (x) := \int_{\F} \L(x,y) \: d{\rho}^{[n]}(y) - 1 \:. 
\end{align}
This gives rise to a sequence of regular Borel measures $({\rho}^{[n]})_{n \in \N}$ on $\F$. 
In particular, whenever condition~{\rm{(iv)}} is satisfied for~$\rho^{[n]}$ (see \eqref{Lint}), that is
\begin{align}\label{(ivn)}
	\sup_{x \in \F} \int_{\F} \L(x,y) \: d\rho^{[n]} < \infty \qquad \text{for all~$n \in \N$} \:, 
\end{align}
each measure~${\rho}^{[n]}$ is a minimizer on $\F^{(n)}$ under variations of finite volume~\cite{noncompact} (see Definition~\ref{deffv} and Definition \ref{Definition minimizer}). In virtue of Theorem~\ref{Theorem 5.8 + 5.9}, the same holds true if the Lagrangian~$\L^{(n)}$ decays in entropy for any~$n \in \N$, provided that condition~\eqref{(ivn)} is satisfied.

\section{Construction of a Global Measure}\label{Section Construction Global}
In the following, let $X$ be an infinite-dimensional, separable complex Banach space, and let $\F \subset X$ be a non-locally compact Polish space endowed with a corresponding metric~$d$ such that $(\F, d)$ is a separable, complete metric space. 
By~$\mathscr{O}(\F)$ and~$\mathscr{P}(\F)$ we denote the collection of open subsets of $\F$ and the power set of~$\F$, respectively. Moreover, the collection of all compact subsets of $\F$ is represented by~$\mathfrak{K}(\F)$. 

The goal of this section is to construct a global measure~$\rho$ based on the sequence of regular Borel  measures~$({\rho}^{[n]})_{n \in \N}$ as obtained in \S \ref{S Approximating}. To this end, in a first step we construct a countable set $\mathscr{D} \subset \mathfrak{K}(\F)$ consisting of compact subsets of $\F$ (\S \ref{S Construction of Countable Set}). In a second step, we make use of the set~$\mathscr{D}$ in order to obtain a measure $\rho$ on $\F$ (\S \ref{S Global Measure}) by a suitable construction process. In \S \ref{S Vague Convergence} we finally prove that, restricted to suitable relatively compact subsets of $\F$, the measure~$\rho$ is the weak limit of a subsequence of~$({\rho}^{[n]})_{n \in \N}$.  

\subsection{Construction of a Countable Collection of Compact Sets}\label{S Construction of Countable Set}
To begin with, separability of $\F$ yields the existence of a countable dense subset~$E := \{x_j : j \in \N \}$
such that, for every $n \in \N$ the set~$E^{(n)} := E \cap \F^{(n)}$ is dense in~$\F^{(n)}$.\footnote{Since~$\F$ is separable, there exists a countable set~$E^{(0)} \subset \F$ being dense in~$\F$. Similarly, for each~$i \in \N$ there are countable sets~$E^{(i)}$ which are dense in~$\F^{(i)}$. As a consequence, the set~$E := \bigcup_{i=0}^{\infty}$ has the desired properties.} 
We denote its elements by~$x_j^{(n)} \in E^{(n)}$ with~$j,n \in \N$. Moreover, since~$\F^{(n)}$ is locally compact, for all~$j, k, n \in \N$ there is a compact neighborhood~$V_{j,k}^{(n)} \subset \F^{(n)}$ of~$x_j^{(n)} \in E^{(n)}$ such that~$V_{j,k}^{(n)} \subset B_{1/k}(x_j^{(n)})$ and each~$V_{j,k}^{(n)}$ being the closure of its interior (in the topology of $\F^{(n)}$, where the interior of a set $V$ shall be denoted by~$V^{\circ}$).\footnote{For simplicity, one may consider~$V_{j,k}^{(n)} = \overline{B_{1/(2k)}(x_j^{(n)}) \cap \F^{(n)}}$ for all~$j,k,n \in \N$.} 
This gives rise to the set 
\[\mathscr{V}^{(1)} := \left\{V_{j,k}^{(n)} : j,k,n \in \N \right\} \:. \]
Denoting the union of 
$\mathscr{V}^{(1)}$ and the empty set $\varnothing$ 
by~$\tilde{\mathscr{D}}^{(1)}$, 
and making use of the fact that a countable union of countable sets is countable (see e.g.~\cite[Section 2]{gaal}), we conclude that $\tilde{\mathscr{D}}^{(1)}$ is countable. Therefore, applying Cantor's diagonal argument and proceeding iteratively, we conclude that
\begin{align*}
\tilde{\mathscr{D}}^{(i)} := \left\{D \cup \tilde{D} : \text{$D, \tilde{D} \in \tilde{\mathscr{D}}^{(i-1)}$} \right\}
\end{align*}
is countable for every $i \in \N$ with~$i \ge 2$. As a consequence, the set
\begin{align}\label{(D)}
{\mathscr{D}} := \bigcup_{i=1}^{\infty} \tilde{\mathscr{D}}^{(i)}
\end{align}
is countable; we denote its members by $(D_m)_{m \in \N}$. In particular, each~$D \in \mathscr{D}$ is a compact subset of $\F$. Moreover, for every $n \in \N$ we introduce~$\mathscr{D}^{(n)} \subset \mathscr{P}(\F^{(n)})$ by
\begin{align}\label{(Dn)}
\mathscr{D}^{(n)} := \left\{D \in \mathscr{D} : D \subset \F^{(n)} \right\} \:.
\end{align}

\subsection{Construction of a Regular Global Measure}\label{S Global Measure}
In order to construct a global measure on~$\F$, 
we proceed similarly to \cite{noncompact} by selecting suitable subsequences of the sequence~$({\rho}^{[n]})_{n \in \N}$ restricted to compact subsets $D \in \mathscr{D}$. This allows us to construct a regular measure~$\rho$ on the whole space $\F$ (see Theorem~\ref{Theorem measure} below). In Section~\ref{Section minimizers bounded range} we will show that, under suitable assumptions, the measure~$\rho$ is indeed a minimizer of the causal variational principle. In analogy to \cite[Lemma~4.1]{noncompact}, let us first state the following result. 

\begin{Lemma}\label{Lemma upper}
	Assume that the Lagrangian $\L : \F \times \F \to \R_0^+$ is lower semi-continuous and strictly positive on the diagonal \eqref{strictpositive}. Furthermore, let $({\rho}^{[n]})_{n \in \N}$ be a sequence of measures~${\rho}^{[n]} : \B(\F) \to [0, \infty]$ such that, for every $x \in \supp {\rho}^{[n]}$, 
	\begin{align*}
	\int_{\F} \L(x,y) \: d{\rho}^{[n]}(y) = 1 \qquad \text{for all $n \in \N$} \:.
	\end{align*}
	Then for every compact subset $K \subset \F$ there is a constant $C_K> 0$ such that 
	\begin{align*}
	{\rho}^{[n]}(K) \le C_K \qquad \text{for all $n \in \N$} \:.
	\end{align*}
\end{Lemma}
\begin{proof}
	This statement is proven exactly as \cite[Lemma 4.1]{noncompact}. 
\end{proof}

Next, we apply Lemma \ref{Lemma upper} to the compact sets $D \in \mathscr{D}$. 
More precisely, restricting the sequence $({\rho}^{[n]})_{n \in \N}$ as obtained in \eqref{(rho[n])} (cf.~\S\ref{S Approximating}) to the compact set~$D_1 \in \mathscr{D}$, the resulting sequence~$({\rho}^{[n]}|_{D_1})_{n \in \N}$ is bounded (due to Lemma \ref{Lemma upper}) as well as uniformly tight (for the definition see~\cite[Definition~8.6.1]{bogachev}). Since compact subsets of Polish spaces are again Polish, Prohorov's theorem (see for instance~\cite[Theorem 8.6.2]{bogachev} or \cite[Satz~VIII.4.23]{elstrodt}) implies that a subsequence of $({\rho}^{[n]}|_{D_1})_{n \in \N}$ converges weakly on~$D_1$. Denoting the corresponding subsequence by $({\rho}^{[1, n_k]})_{k \in \N}$ and considering its restriction to~$D_2 \in \mathscr{D}$, the same arguments as before yield the existence of a weakly convergent subsequence~$({\rho}^{[2, n_k]})_{k \in \N}$ on~$D_2$. Proceeding iteratively, we denote the resulting diagonal sequence by 
\begin{align}\label{(rhorund)}
{\rho}^{(k)} := {\rho}^{[k, n_k]} \qquad \text{for all $k \in \N$} \:.
\end{align}
Thus by construction, for every~$m \in \N$ the sequence~$({\rho}^{(k)}|_{D_m})_{k \in \N}$ converges weakly to some measure ${\rho}_{D_m} : \B(D_m) \to [0, \infty)$, 
\begin{align}\label{(weak convergence on Dm)}
{\rho}^{(k)}|_{D_m} \rightharpoonup {\rho}_{D_m} \:.
\end{align}
In particular, 
\begin{align*}
\lim_{k \to \infty} \rho^{(k)}(D_m) = \rho_{K_m}(D_m) \qquad \text{for all $m \in \N$} \:. 
\end{align*}

We point out that each measure ${\rho}^{(k)}$ is a minimizer on $\F^{(n_k)}$. For this reason, we restrict attention to the finite-dimensional exhaustion $(\F^{(k)})_{k \in \N}$, where for notational simplicity by $\F^{(k)}$ we denote the sets $\F^{(n_k)}$ for all $k \in \N$. 
Note that the sequence constructed in \eqref{(rhorund)} above in general does \emph{not} converge weakly on \emph{arbitrary} compact subsets, but only restricted to compact sets~$D \in \mathscr{D}$ (cf.~\eqref{(weak convergence on Dm)}). In \cite{noncompact}, this problem was resolved by deriving vague convergence of the sequence $(\rho^{(n)})_{n \in \N}$ to some global measure~$\rho$. 
In order to obtain a similar situation, let us state the following result. 
\begin{Prp}\label{Proposition set function}
	The set function ${\varphi} : \mathscr{D} \to [0, \infty)$ defined by  
	\begin{align}\label{(infty)}
	{\varphi}(D) := \lim_{k \to \infty} {\rho}^{(k)}(D) < \infty \qquad \text{for any $D \in \mathscr{D}$} 
	\end{align}
	has the following properties: 
	\begin{enumerate}[leftmargin=2em]
		\item[\rm{(1)}] ${\varphi}(D_1) \le {\varphi}(D_2)$ for all $D_1, D_2 \in \mathscr{D}$ with $D_1 \subset D_2$, 
		\item[\rm{(2)}] ${\varphi}(D_1 \cup D_2) \le {\varphi}(D_1) + {\varphi}(D_2)$ for all $D_1, D_2 \in \mathscr{D}$, and
		\item[\rm{(3)}] ${\varphi}(D_1 \cup D_2) = {\varphi}(D_1) + {\varphi}(D_2)$ for all $D_1, D_2 \in \mathscr{D}$ with $D_1 \cap D_2 = \varnothing$.
	\end{enumerate} 
\end{Prp}
\begin{proof}
	Given $D_1, D_2 \in \mathscr{D}$ with $D_1 \subset D_2$, property (1) follows from 
	\begin{align*}
	\varphi(D_1) = \lim_{k \to \infty} \int_{\F} d\rho^{(k)}|_{D_1} \le \lim_{k \to \infty} \int_{\F} d\rho^{(k)}|_{D_2} = \varphi(D_2) \:. 
	\end{align*}
	Next, for all $D_1, D_2 \in \mathscr{D}$, construction of $\mathscr{D}$ yields $D_1 \cup D_2 \in \mathscr{D}$. Thus property (2) is a consequence of
	\begin{align*}
	\begin{split}
	\varphi(D_1 \cup D_2) &= \lim_{k \to \infty} \int_{\F} d\rho^{(k)}|_{D_1 \cup D_2} \le \lim_{k \to \infty} \int_{\F} d\rho^{(k)}|_{D_1} + \lim_{k \to \infty} \int_{\F} d\rho^{(k)}|_{D_2} \\
	&= \varphi(D_1) + \varphi(D_2) \:. \phantom{\int_{\F}}
	\end{split}
	\end{align*}
	Similarly, for all $D_1, D_2 \in \mathscr{D}$ with $D_1 \cap D_2 = \varnothing$ we obtain 
	\begin{align*}
	\begin{split}
	\varphi(D_1 \cup D_2) = \lim_{k \to \infty} \rho^{(k)}(D_1 \cup D_2) = \lim_{k \to \infty} \rho^{(k)}(D_1) + \lim_{k \to \infty} \rho^{(k)}(D_2) = \varphi(D_1) + \varphi(D_2) \:,
	\end{split}
	\end{align*}
	which proves property (3). 
\end{proof}

In order to construct a global measure~${\rho}$ on $\F$, we proceed in analogy to the proof of~\cite[Satz VIII.4.22]{elstrodt}. We point out that, since the underlying topological space~$\F$ is non-locally compact, we cannot employ the Riesz representation theorem as in~\cite{noncompact}, and Riesz representation theorems on more general Hausdorff spaces as presented in~\cite[Section 16]{koenig} do not seem applicable at this stage. 
Nevertheless, we obtain the following result.

\begin{Thm}\label{Theorem measure}
	Introducing the set function ${\varphi} : \mathscr{D} \to [0, \infty)$ by \eqref{(infty)}
	and defining the set functions ${\mu} : \mathscr{O}(\F) \to [0, \infty]$ and ${\eta} : \mathscr{P}(\F) \to [0, + \infty]$ by
	\begin{align}\label{(regular)}
	\begin{array}{ll}
	{\mu}(U) := \sup \left\{{\varphi}(D) : \text{$D \subset U$, $D \in \mathscr{D}$ } \right\} \quad \quad &\text{for all $U \subset \F$ open} \:, \\ [0.25em]
	{\eta}(A) := \inf \left\{{\mu}(U) : \text{$A \subset U$, $U \subset \F$ open} \right\} \quad \quad &\text{for any $A \in \mathscr{P}(\F)$} \:,
	\end{array}
	\end{align}
	then the restriction
	\begin{align}\label{(rho)}
	{\rho} := {\eta}|_{\B(\F)}
	\end{align}
	defines a (possibly non-trivial) measure on the Borel $\sigma$-algebra $\B(\F)$. In particular, 
	\begin{align}\label{(weak convergence rho tilde)}
	{\rho}(D) = {\varphi}(D) = \lim_{k \to \infty} {\rho}^{(k)}(D) \qquad \text{for any $D \in \mathscr{D}$} \:.
	\end{align}
\end{Thm}
\begin{proof}
	Let us first point out that, by construction, the set function ${\varphi} : \mathscr{D} \to [0, \infty)$ defined by \eqref{(infty)} has the following properties:
	\begin{enumerate}[leftmargin=2em]
		\item[(1)] ${\varphi}(D_1) \le {\varphi}(D_2)$ for all $D_1, D_2 \in \mathscr{D}$ with $D_1 \subset D_2$, 
		\item[(2)] ${\varphi}(D_1 \cup D_2) \le {\varphi}(D_1) + {\varphi}(D_2)$ for all $D_1, D_2 \in \mathscr{D}$, and
		\item[(3)] ${\varphi}(D_1 \cup D_2) = {\varphi}(D_1) + {\varphi}(D_2)$ for all $D_1, D_2 \in \mathscr{D}$ with $D_1 \cap D_2 = \varnothing$.
	\end{enumerate} 
	Indeed, properties (1)--(3) are a consequence of Proposition \ref{Proposition set function}. 
	Moreover, $\varphi(\varnothing) = 0$ (since $\varnothing \in \mathscr{D}$). 
	
	Next, similarly to the proof of \cite[Satz VIII.4.22]{elstrodt}, our goal is to show that~${\eta}$ is an outer measure,\footnote{Given a set $X$, a set function $\eta : \mathscr{P}(X) \to \overline{\R} := [- \infty, + \infty]$ is said to be an \emph{outer measure} if it has the following properties (see e.g.~\cite[Definition II.4.1]{elstrodt}): 
		\begin{enumerate}[leftmargin=2em]
			\item[(i)] $\eta(\varnothing) = 0$. 
			\item[(ii)] For all $A \subset B \subset X$ holds $\eta(A) \le \eta(B)$ (monotonicity). 
			\item[(iii)] For every sequence $(A_n)_{n \in \N}$ of subsets of $X$ holds $\eta\big(\bigcup_{n=1}^{\infty} A_n \big) \le \sum_{n=1}^{\infty} \eta(A_n)$ ($\sigma$-subadditivity).
		\end{enumerate}
		} 
	and that every Borel set $B \in \B(\F)$ is ${\eta}$-measurable. This shall be done in the following by proving that
	\begin{align}\label{(A)}
	\text{\emph{$\eta$ is an outer measure, and each closed set $A \subset\F$ is $\eta$-measurable}} \:.
	\end{align}
	Denoting the $\sigma$-algebra of $\eta$-measurable sets by $\mathfrak{A}_{\eta}$, the statement \eqref{(A)} implies that the Borel $\sigma$-algebra $\B(\F)$ is contained in $\mathfrak{A}_{\eta}$, i.e.~$\B(\F) \subset \mathfrak{A}_{\eta}$. Therefore, in view of Carathéodory's theorem (see e.g.~\cite[Satz II.4.4]{elstrodt}), the restriction
	\begin{align*}
	{\rho} := {\eta}|_{\B(\F)}
	\end{align*}
	defines a measure on the Borel $\sigma$-algebra $\B(\F)$. 
	Thus it suffices to prove \eqref{(A)}, which shall be done in several steps in the remainder of the proof. 
	\begin{enumerate}[leftmargin=2em]
		\item[(a)] \emph{Let $A \subset U \subset \F$ with $A$ closed and $U$ open. Whenever $A \subset D$ for some $D \in \mathscr{D}$, then there exists $E \in \mathscr{D}$ with $A \subset E \subset U$.} 
	\end{enumerate} 
	Proof: Since each $D \in \mathscr{D}$ is compact, the closed set $A \subset D$ is compact as well. Moreover, $D \subset \F^{(n)}$ for sufficiently large $n \in \N$. 
	Since $\F^{(n)}$ is locally compact, for every $x \in A$ there exists $V_x \in \mathscr{D}$ such that $x \in V_x^{\circ} \subset V_x \subset U$. Since $A$ is compact, the set~$E := \bigcup_{j=1}^{N} V_{x_j} \in \mathscr{D}$ for some integer $N = N(A)$ has the desired property. 
	\begin{enumerate}[leftmargin=2em]
		\item[(b)] \emph{Whenever $U, V \subset \F$ open, $\mu(U \cup V) \le \mu(U) + \mu(V)$.}
	\end{enumerate} 
	Proof: Without loss of generality, let $U \not= \F \not= V$ and $\mu(U), \mu(V) < \infty$ (otherwise the inequality is true). For this reason, let~$U, V \subset \F$ be open sets with~$U^c \not= \varnothing \not= V^c$ and~$D \subset U \cup V$ for $D \in \mathscr{D}$. We then consider the closed sets
	\begin{align*}
	A &:= \left\{ x \in D : d(x, U^c) \ge d(x, V^c) \right\} \subset D \:, \\
	B &:= \left\{ x \in D : d(x, U^c) \le d(x, V^c) \right\} \subset D \:.
	\end{align*}
	Obviously, $A \subset U$ and $B \subset V$. Assuming conversely that $x \in A \setminus U$, we conclude that~$x \in V$, and therefore $d(x, U^c) = 0 < d(x, V^c)$ because $V^c$ is closed, giving rise to the contradiction that $x \notin A$. Similarly, we conclude that $B \subset V$. 
	Since $A \subset D$, by virtue of (a) there exists 
	$E \in \mathscr{D}$ with $A \subset E \subset U$. 
	Similarly, there exists $F \in \mathscr{D}$ such that $B \subset F \subset V$, and $D = A \cup B \subset E \cup F$. Hence (1) and (2) yield
	\begin{align*}
	\varphi(D) \le \varphi(E \cup F) \le \varphi(E) + \varphi(F) \le \mu(U) + \mu(V) \:.
	\end{align*}
	Taking the supremum over all $D \in \mathscr{D}$ with $D \subset U \cup V$ gives (b). 
	\begin{enumerate}[leftmargin=2em]
		\item[(c)] \emph{For all $n \in \N$ and $U_n \subset \F$ open, $\mu \big(\bigcup_{n=1}^{\infty} U_n \big) \le \sum_{n=1}^{\infty} \mu(U_n)$.}  
	\end{enumerate} 
	Proof: Let $D \in \mathscr{D}$ with $D \subset \bigcup_{n=1}^{\infty} U_n$. Then by compactness of $D$ there exists $p \in \N$ such that $D \subset \bigcup_{n=1}^p U_n$. Applying (b) inductively, we conclude that
	\begin{align*}
	\varphi(D) \le \mu \left(\bigcup_{n=1}^p U_n \right) \le \sum_{n=1}^p \mu(U_n) \le \sum_{n=1}^{\infty} \mu(U_n) \:.
	\end{align*}
	Since $D \in \mathscr{D}$ with $D \subset \bigcup_{n=1}^{\infty} U_n$ is arbitrary, we obtain (c). 
	\begin{enumerate}[leftmargin=2em]
		\item[(d)] \emph{$\eta$ is an outer measure.}
	\end{enumerate} 
	Proof: As seen before, $\varphi(\varnothing) = 0$, and monotonicity of $\eta$ is a consequence of (1)--(3) and~\eqref{(regular)}. 
	In order to prove $\sigma$-subadditivity, let $\varepsilon > 0$ and $M_n \subset \F$ with $\eta(M_n) < \infty$ for all~$n \in \N$. In view of \eqref{(regular)}, for every $n \in \N$ there exists an open set $U_n \supset M_n$ with~$\mu(U_n) \le \eta(M_n) + 2^{-n} \: \varepsilon$. Making use of~\eqref{(regular)} and applying (c) yields 
	\begin{align*}
	\eta \left(\bigcup_{n=1}^{\infty} M_n \right) \le \mu \left(\bigcup_{n=1}^{\infty} U_n \right) \le \sum_{n=1}^{\infty} \mu(U_n) \le \sum_{n=1}^{\infty} \eta(M_n) + \varepsilon \:.
	\end{align*}
	Since $\varepsilon > 0$ is arbitrary, we obtain (d). 
	\begin{enumerate}[leftmargin=2em]
		\item[(e)] \emph{Each closed set $A \subset \F$ is $\eta$-measurable.}
	\end{enumerate} 
	Proof: By definition of measurability (cf.~\cite[Definition II.4.2]{elstrodt}), we need to show that, for all $Q \subset \F$,
	\begin{align}\label{(4.18)}
	\eta(Q) \ge \eta(Q \cap A) + \eta(Q \cap A^c) \:.
	\end{align}
	Without loss of generality we may assume that $\eta(Q) < \infty$. 
	We first prove \eqref{(4.18)} in the case that $Q = U \subset \F$ is \emph{open}. To this end, let $\varepsilon > 0$ arbitrary. Given~$A \subset \F$ closed, the set $U \cap A^c$ is open and $\mu(U \cap A^c) = \eta(U \cap A^c) < \infty$. In view of \eqref{(regular)}, there exists~$D \in \mathscr{D}$ with $\varphi(D) \ge \mu(U \cap A^c) - \varepsilon$. Next, since $U \cap D^{c}$ is open, we may choose~$E \in \mathscr{D}$ with $E \subset U \cap D^{c}$ and $\varphi(E) \ge \mu(U \cap D^c) - \varepsilon$. Since $D$, $E$ are disjoint and $D \cup E \subset U$, from (1), (3), \eqref{(regular)} and the fact that $U \cap D^c \supset U \cap A$ we conclude that 
	\begin{align*}
	\mu(U) &\ge \varphi(D \cup E) = \varphi(D) + \varphi(E) \\
	&\ge \mu(U \cap A^c) + \mu(U \cap D^c) - 2 \varepsilon \\
	&\ge \eta(U \cap A) + \mu(U \cap A^c) - 2 \varepsilon \:.
	\end{align*}
	Since $\varepsilon > 0$ is arbitrary, we obtain \eqref{(4.18)} for $Q = U$ open. 
	
	Given arbitrary $Q \subset \F$ with $\eta(Q) < \infty$, for $\varepsilon > 0$ arbitrary we choose $U \supset Q$ open with~$\eta(Q) \ge \eta(U) - \varepsilon$ according to \eqref{(regular)}. Then the latter inequality yields 
	\begin{align*}
	\eta(Q) &\ge \eta(U) - \varepsilon \ge \eta(U \cap A) + \eta(U \cap A^c) - \varepsilon \\ 
	&\ge \eta(Q \cap A) + \eta(Q \cap A^c) - \varepsilon \:,
	\end{align*} 
	proving \eqref{(4.18)}.
	
	As a consequence, the set function $\eta$ is an outer measure according to (d), and each closed set $A \subset \F$ is $\eta$-measurable in view of (e). This yields \eqref{(A)}, which completes the proof. 
\end{proof}

The next result shows that the measure $\rho$ given by \eqref{(rho)} is regular~\cite{elstrodt}. 

\begin{Lemma}\label{Lemma inner regular}
	Let ${\rho} : \B(\F) \to [0, \infty]$ be the measure defined by \eqref{(rho)}. Then every open subset of $\F$ is inner regular. Moreover, the measure~$\rho$ is regular. 
\end{Lemma}
\begin{proof}
	Let us first prove that every open subset of $\F$ is inner regular in view of \eqref{(regular)}. Namely, considering arbitary $U \in \mathscr{O}(\F)$ and $K \in \mathfrak{K}(U)$, 
	according to~\eqref{(regular)} and~\eqref{(weak convergence rho tilde)} we obtain 
	\begin{align*}
	{\rho}(K) \le {\rho}(U) &= \sup \left\{{\varphi}(D) : \text{$D \in \mathscr{D}$, $D \subset U$} \right\} = \sup \left\{{\rho}(D) : \text{$D \in \mathscr{D}$, $D \subset U$} \right\} \\
	&\le \sup \left\{{\rho}(K) : \text{$K \subset U$ compact} \right\} \:.
	\end{align*}
	Taking the supremum on the left hand side yields 
	\begin{align*}
	{\rho}(U) = \sup \left\{ \rho(K) : \text{$K \subset U$ compact} \right\} \qquad \text{for all $U \in \mathscr{O}(\F)$} \:.
	\end{align*}
	From this we conclude that every open set $U \in \mathscr{O}(\F)$ is inner regular (in the sense of~\cite[Definition VIII.1.1]{elstrodt}).
	
	In view of \eqref{(regular)}, we are given
	$\rho(U) = \mu(U)$ for any $U \in \mathscr{O}(\F)$, and the above considerations show that every open set is inner regular. From this we conclude that the measure~$\rho$ is regular in the sense of \cite[Definition~VIII.1.1]{elstrodt}. 
\end{proof}

As a matter of fact, in general it seems possible the regular measure~$\rho$ obtained in Theorem~\ref{Theorem measure} to be zero. Nevertheless, the following remark gives a sufficient condition for the measure $\rho$ defined by~\eqref{(rho)} to be non-zero. 

\begin{Remark}\label{Remark non-trivial}
	Let $(\F^{(n)})_{n \in \N}$ be a finite-dimensional approximation of $\F$ (see \S \ref{S Approximating}). 
	By construction of ${\rho}^{(n)}$ we are given $\supp {\rho}^{(n)} \subset \F^{(n)}$ for every $n \in \N$. Assuming that the Lagrangian is bounded and of bounded range, for every $x \in \F$ and $\delta > 0$ there exists~$B_x \subset \F$ bounded and closed such that~$\L(\tilde{x},y) = 0$ for all $\tilde{x} \in {B_{\delta}(x)}$ and all~$y \notin B_x$.
	Furthermore, in view of boundedness of the Lagrangian we introduce the upper bound~$\mathscr{C} < \infty$ by
	\begin{align*}
		\mathscr{C} := \sup_{x, y \in \F} \L(x,y) > 0 \:. 
	\end{align*}
	Thus for any $n \in \N$ we deduce that~$\L^{(n)}(\tilde{x}, y) = 0$ for all $\tilde{x} \in B_{\delta}^{(n)}(x) := {B_{\delta}(x)} \cap \F^{(n)}$ and all $y \notin B_x^{(n)} := B_x \cap \F^{(n)}$. Hence the EL equations \eqref{(ELn')} and \eqref{(elln')} imply that
	\begin{align*}
	1 \le \int_{\F^{(n)}} \L^{(n)}(\tilde{x},y) \: d{\rho}^{(n)} = \int_{B_x^{(n)}} \L^{(n)}(\tilde{x},y) \: d{\rho}^{(n)} \le \sup_{y \in B_x^{(n)}} \L^{(n)}(\tilde{x},y) \: {\rho}^{(n)}(B_x^{(n)}) 
	\end{align*}
	for every $n \in \N$. Thus positivity \eqref{strictpositive} yields 
	$${\rho}^{(n)}(B_x^{(n)}) \ge \mathscr{C}^{-1} > 0 \qquad \text{for all $n \in \N$} \:. $$
	For each $n \in \N$ and arbitrary $\varepsilon > 0$, by regularity of $\rho^{(n)}$ there exists $D_n \in \mathscr{D}$ such that~$\rho^{(n)}(D_n) > \mathscr{C}^{-1} - \varepsilon$. Moreover, $\hat{D}_N := \bigcup_{n=1}^N D_n \in \mathscr{D}$ for every $N \in \N$. Whenever there exists $N \in \N$ such that $\rho^{(n)}(\hat{D}_N) \ge c$ for almost all $n \in \N$ and some $c > 0$, then the measure~$\rho$ defined by~\eqref{(rho)} is non-zero. If this holds true for an infinite number of disjoints sets $(\hat{D}_{N_i})_{i \in \N}$, the measure~$\rho$ possibly has infinite total volume. 
\end{Remark}

Next, in agreement with \cite[Theorem 16.7]{koenig} and the remark thereafter, it is not clear whether~$\rho$ as given by~\eqref{(rho)} is a locally finite measure. 
Nevertheless, the following results provide sufficient conditions for $\rho$ as obtained in~\eqref{(rho)} to be locally finite.

\begin{Lemma}\label{Lemma locally finite}
	Let ${\rho} : \B(\F) \to [0, \infty]$ be defined by \eqref{(rho)}. Assuming that~$\rho(K) < \infty$ for all $K \in \mathfrak{K}(\F)$, then the measure~$\rho$ is locally finite and thus a Borel measure in the sense of~\cite{gardner+pfeffer}. In this case, $\rho$ is regular and moderate. 
\end{Lemma}
\begin{proof}
	We point out that the space~$\F$ is first-countable. Thus
	under the assumption that $\rho(K) < \infty$ for all~$K \in \mathfrak{K}(\F)$, the statement that $\rho$ is locally finite is a consequence of Lemma \ref{Lemma inner regular} and~\cite[Folgerungen~VIII.1.2~(d)]{elstrodt}. 
	The last statement follows from Ulam's theorem \cite[Theorem VIII.1.16]{elstrodt}. 
\end{proof}

\begin{Remark}
	We point out that, if~$\F$ in Definition \ref{Definition non-locally compact} \emph{were} locally compact, then the measure~$\rho$ constructed in the proof of Theorem \ref{Theorem measure} \emph{would be} locally finite, i.e.\ a Borel measure in the sense of~\cite{gardner+pfeffer}. Namely, whenever $x \in \F$, there exists a compact neighborhood~$V$ of~$x$. Hence Lemma \ref{Lemma upper} implies that~$\rho^{(n)}(V) \le C_V$ for all $n \in \N$ and some~$C_V > 0$. Choosing~$U_x \subset V$ open with $x \in U$, we conclude that
	\begin{align*}
	\rho(U_x) = \sup \left\{\varphi(D) : \text{$D \subset U_x$, $D \in \mathscr{D}$} \right\} \le \sup \left\{\lim_{n \to \infty} \rho^{(n)}(D) : \text{$D \subset V$, $D \in \mathscr{D}$} \right\} \le C_V 
	\end{align*}
	as desired. 
\end{Remark}

In the remainder of this subsection, we discuss the properties~\rm{(iii)} and~\rm{(iv)} in \S \ref{seccvpsigma}. Neither condition~\rm{(iii)} nor condition~\rm{(iv)} do hold in general, but the following result establishes a connection between 
condition~\rm{(iv)} in~\S \ref{seccvpsigma} and locally finite measures.

\begin{Lemma}\label{Lemma iv Borel}
	Assume that the Lagrangian $\L : \F \times \F \to \R_0^+$ is lower semi-continuous, symmetric and strictly positive on the diagonal \eqref{strictpositive}, and let $\rho$ be a measure on $\B(\F)$. Under the assumption that condition~{\rm{(iv)}} in \S \ref{seccvpsigma} holds (see \eqref{Lint}), i.e.
	\begin{align*}
	\sup_{x \in \F} \int_{\F} \L(x,y) \: d\rho(y) < \infty \:,
	\end{align*}
	the measure $\rho$ is locally finite (i.e.\ condition~{\rm{(iii)}} in \S \ref{seccvpsigma} is satisfied).  
\end{Lemma}
\begin{proof}
	Assume conversely that there exists $x \in \F$ such that $\rho(U) = \infty$ for any open neighborhood $U$ of $x$. Then~$\L(x,x) > 0$ due to strict positivity on the diagonal~\eqref{strictpositive}. Moreover, by lower semi-continuity of the Lagrangian there exists an open neighborhood $U_x$ of $x$ such that $\L(x,y) > \L(x,x)/2 > 0$ for all $y \in U_x$. Consequently,
	\begin{align*}
	\int_{\F} \L(x,y) \: d\rho(y) \ge \int_{U_x} \L(x,y) \: d\rho(y) > \L(x,x)/2 \: \rho(U_x) = \infty
	\end{align*}
	in contradiction to condition~{\rm{(iv)}} in \S \ref{seccvpsigma}. 
\end{proof}

\subsection{Convergence on Relatively Compact Subsets}\label{S Vague Convergence}
In Section \ref{Section minimizers bounded range} below our goal is to show that, under suitable assumptions, the measure~$\rho$ as given by \eqref{(rho)} is a minimizer under variations of finite volume (see Definition \ref{Definition finite volume}). To this end, we provide some useful tools which shall be worked out in the remainder of this section. For clarity, we point out that for every~$D \in \mathscr{D}$ there exists some $n' \in \N$ such that $D^{\circ} \not= \varnothing$ in the relative topology of~$\F^{(n)}$ for all $n \le n'$ and $D^{\circ} = \varnothing$ in the relative topology of~$\F^{(n)}$ for all $n > n'$. Considering the restriction $\rho|_{D^{\circ}}$ shall always be understood in the sense of a restriction to the interior of~$D$ 
in the relative topology of $\F^{(n')}$. In order to derive weak convergence on suitable relatively compact sets (see Lemma \ref{Lemma weak convergence} below), we require some properties of the measures~$\rho_D$ with~$D \in \mathscr{D}$ as obtained in \eqref{(weak convergence on Dm)}. 

Given $D, E \in \mathscr{D}$ with $D \subset E$ and 
denoting the interior of $D$ by $D^{\circ}$, we claim that 
\begin{align}\label{(coincidence)}
\int_{D} f \: d\rho_{D} = \int_{E} f \: d\rho_{E} \qquad \text{for all $f \in C_c(D^{\circ})$} \:. 
\end{align}
Namely, in view of $C_c(D^{\circ}) \subset C_b(D) \cap C_b(E)$, weak convergence \eqref{(weak convergence on Dm)} yields 
\begin{align*}
\int_{D} f \: d\rho_{D} = \lim_{k \to \infty} \int_{D} f \: d\rho^{(k)}|_{D} \stackrel{\text{($\star$)}}{=} \lim_{k \to \infty} \int_{E} f \: d\rho^{(k)}|_{E} = \int_{E} f \: d\rho_{E}
\end{align*}
for all $f \in C_c(D^{\circ})$, 
where in ($\star$) we made use of the fact that $\supp f \subset D^{\circ} \subset E$. 

This allows us to prove the following result. 

\begin{Prp}
	Whenever $D \in \mathscr{D}$ and $\Omega \subset D^{\circ}$ open, then 
	\begin{align}\label{(sup coincidence)}
		\sup \left\{\rho_{\tilde{D}}|_{D^{\circ}}(\tilde{D}) : \text{$\tilde{D} \subset \Omega$, $\tilde{D} \in \mathscr{D}$} \right\} =\sup \left\{\rho_{D}|_{D^{\circ}}(\tilde{D}) : \text{$\tilde{D} \subset \Omega$, $\tilde{D} \in \mathscr{D}$} \right\} \:. 
	\end{align}
\end{Prp}
\begin{proof}
	In order to prove \eqref{(sup coincidence)}, we need to show that for all $E, F \in \mathscr{D}$ with $E, F \subset \Omega$ there exist sets $\tilde{E}, \tilde{F} \in \mathscr{D}$ with $\tilde{E}, \tilde{F} \subset \Omega$ such that
	\begin{align*}
	\rho_{E}|_{D^{\circ}}(E) \le \rho_{D}|_{D^{\circ}}(\tilde{E}) \qquad \text{and} \qquad \rho_{D}|_{D^{\circ}}(F) \le \rho_{\tilde{F}}|_{D^{\circ}}(\tilde{F}) \:. 
	\end{align*}
	Whenever $E \subset \Omega$ compact, there exists $\eta \in C_c(\Omega; [0,1])$ with $\eta|_{E} \equiv 1$ and~$\tilde{E} \in \mathscr{D}$ with~$\supp \eta \subset \tilde{E}^{\circ} \subset \Omega$, and by weak convergence \eqref{(weak convergence on Dm)} we are given
	\begin{align*}
	\rho_{E}|_{D^{\circ}}(E) &= \int_{E} d\rho_E = \lim_{n \to \infty} \int_{\F} d\rho^{(n)}|_E \le \lim_{n \to \infty} \int_{\F} \eta \: d\rho^{(n)}|_{D} = \int_{\F} \eta \: d\rho_{D} 
	\le \rho_D|_{D^{\circ}}(\tilde{E}) \:. 
	\end{align*}
	On the other hand, whenever $F \in \mathscr{D}$ with $F \subset \Omega$, there exists $\eta \in C_c(\Omega; [0,1])$ with~$\eta|_{F} \equiv 1$ as well as $\tilde{F} \in \mathscr{D}$ with $\supp \eta \subset \tilde{F}^{\circ}$ and~$\tilde{F}^{\circ} \subset \Omega$. Thus by~\eqref{(coincidence)} we obtain
	\begin{align*}
		\rho_{D}|_{D^{\circ}}(F) = \int_{F} d\rho_D \le \int_{D} \eta \: d\rho_D = \int_{\tilde{F}} \eta \: d\rho_{\tilde{F}} \le \rho_{\tilde{F}}|_{D^{\circ}}(\tilde{F}) \:,
	\end{align*}
	which completes the proof. 
\end{proof}

\begin{Lemma}\label{Lemma vague convergence}
	For every $D \in \mathscr{D}$, the measures $\rho|_{D^{\circ}}$ and $\rho_D|_{D^{\circ}}$ coincide. Moreover,
	\begin{align}\label{(vague convergence D)}
	\rho^{(n)}|_{D^{\circ}}  \stackrel{v}{\to} \rho|_{D^{\circ}} \qquad \text{vaguely} \:.
	\end{align}
\end{Lemma}
\begin{proof}
	According to \eqref{(weak convergence on Dm)} and \eqref{(infty)}, the measure ${\rho}_D|_{D^{\circ}} : \B(D^{\circ}) \to [0, \infty)$ is finite for any $D \in \mathscr{D}$. As a consequence, it is locally finite and thus Borel in the sense of~\cite{gardner+pfeffer}. Moreover, since open subsets of Polish spaces are Polish (see \cite[\S 26]{bauer}), it is regular in view of Ulam's theorem~\cite[Satz~VIII.1.16]{elstrodt} (alternatively, regularity follows by~\cite[Corollary~7.1.9]{bogachev} and the fact that each metrizable space is perfectly normal~\cite[Corollary~3.21]{aliprantis}). Similar arguments yield regularity of ${\rho}|_{D^{\circ}}$ for any~$D \in \mathscr{D}$. Thus for any $D \in \mathscr{D}$, we may approximate arbitrary Borel sets $A \in \B(D^{\circ})$ by compact sets from inside,  
	\begin{align*}
	{\rho}_D|_{D^{\circ}}(A) &= \sup \left\{{\rho}_D|_{D^{\circ}}(K) : \text{$K \subset A$ compact} \right\} \:, \\
	{\rho}|_{D^{\circ}}(A) &= \sup \left\{{\rho}|_{D^{\circ}}(K) : \text{$K \subset A$ compact} \right\} \:. 
	\end{align*}
	Whenever $D \in \mathscr{D}$ and $\Omega \subset D$ open, for each $K \subset \Omega$ compact there exists $\tilde{D} \in \mathscr{D}$ such that $K \subset \tilde{D} \subset \Omega$ by construction of $\mathscr{D}$. From this we conclude that
	\begin{align*}
	{\rho}_D|_{D^{\circ}}(\Omega) &= \sup \left\{{\rho}_D|_{D^{\circ}}(K) : \text{$K \subset \Omega$ compact} \right\} = \sup \left\{{\rho}_D|_{D^{\circ}}(\tilde{D}) : \text{$\tilde{D} \subset \Omega$, $\tilde{D} \in \mathscr{D}$} \right\} \:, \\
	{\rho}|_{D^{\circ}}(\Omega) &= \sup \left\{{\rho}|_{D^{\circ}}(K) : \text{$K \subset \Omega$ compact} \right\} = \sup \left\{{\rho}|_{D^{\circ}}(\tilde{D}) : \text{$\tilde{D} \subset \Omega$, $\tilde{D} \in \mathscr{D}$} \right\} \:. 
	\end{align*} 
	Moreover, for any $\tilde{D} \in \mathscr{D}$, by \eqref{(weak convergence on Dm)} and \eqref{(weak convergence rho tilde)} we obtain 
	\begin{align}\label{(equality)}
	{\varphi}(\tilde{D}) = \lim_{k \to \infty} {\rho}^{(k)}(\tilde{D}) = \lim_{k \to \infty} {\rho}^{(k)}|_{\tilde{D}}(\tilde{D}) = {\rho}_{\tilde{D}}(\tilde{D}) \:.
	\end{align}
	Given $D \in \mathscr{D}$ and $\Omega \subset D^{\circ}$ open, we conclude that $\rho|_{D^{\circ}}$ as well as $\rho_D|_{D^{\circ}}$ are regular finite Borel measures on~$\B(D^{\circ})$, implying that
	\begin{align*}
	\rho|_{D^{\circ}}(\Omega) &= \sup \left\{\varphi(\tilde{D}) : \text{$\tilde{D} \subset \Omega$, $\tilde{D} \in \mathscr{D}$} \right\} \stackrel{\eqref{(equality)}}{=} \sup \left\{\rho_{\tilde{D}}(\tilde{D}) : \text{$\tilde{D} \subset \Omega$, $\tilde{D} \in \mathscr{D}$} \right\} \\
	&= \sup \left\{\rho_{\tilde{D}}|_{D^{\circ}}(\tilde{D}) : \text{$\tilde{D} \subset \Omega$, $\tilde{D} \in \mathscr{D}$} \right\} \stackrel{\eqref{(sup coincidence)}}{=} \sup \left\{\rho_{D}|_{D^{\circ}}(\tilde{D}) : \text{$\tilde{D} \subset \Omega$, $\tilde{D} \in \mathscr{D}$} \right\} \\
	&= \rho_D|_{D^{\circ}}(\Omega) \:. \phantom{\sup \left\{\rho_{D}|_{D^{\circ}}(\tilde{D}) : \text{$\tilde{D} \subset \Omega$, $\tilde{D} \in \mathscr{D}$} \right\}}
	\end{align*}
	As a consequence, the measures $\rho|_{D^{\circ}}$ and $\rho_D|_{D^{\circ}}$ coincide on all open sets $\Omega \subset \D^{\circ}$. Making use of \cite[Lemma 7.1.2]{bogachev}, we conclude that $\rho|_{D^{\circ}}$ and $\rho_D|_{D^{\circ}}$ already coincide on all Borel sets, i.e.
	\begin{align}\label{(Borel coincidence)}
	\rho|_{D^{\circ}} = \rho_D|_{D^{\circ}} \qquad \text{for all $D \in \mathscr{D}$} \:.
	\end{align}
	Given $f \in C_c(D^{\circ})$, we thus obtain
	\begin{align*}
		\lim_{n \to \infty} \int_{\F} f \: d\rho^{(n)}|_{D^{\circ}} = \lim_{n \to \infty} \int_{\F} f \: d\rho^{(n)}|_D \stackrel{\eqref{(weak convergence on Dm)}}{=} \int_{\F} f \: d\rho_D = \int_{\F} f \: d\rho_D|_{D^{\circ}} \stackrel{\eqref{(Borel coincidence)}}{=}  \int_{\F} f \: d\rho|_{D^{\circ}} \:.
	\end{align*} 
	Since $f \in C_c(D^{\circ})$ was arbitrary, we obtain vague convergence
	\begin{align*}
		\rho^{(n)}|_{D^{\circ}}  \stackrel{v}{\to} \rho|_{D^{\circ}} \:.
	\end{align*}
	This completes the proof. 
\end{proof}

Having proved vague convergence on open subsets of $D \in \mathscr{D}$, the following result even yields weak convergence on suitable relatively compact subsets $V \subset \F$ (so-called \emph{continuity sets}, cf.~\cite[Section~8.2]{bogachev}).   

\begin{Lemma}\label{Lemma weak convergence}
	For every $D \in \mathscr{D}$ 
	there exists $E \in \mathscr{D}$ with $D \subset E^{\circ}$ as well as a relatively compact, open subset $V \subset E^{\circ}$ with $D \subset V$ such that 
	\begin{align}\label{(weak convergence V)}
		\rho^{(n)}|_V \rightharpoonup \rho|_V \qquad \text{weakly} \:.
	\end{align} 
	Similarly, whenever $D \in \mathscr{D}^{(n)}$ and $U \supset D$ open, there exists a relatively compact, open subset $V \subset \F^{(n)}$ with $D \subset V \subset U$ such that \eqref{(weak convergence V)} holds. 
\end{Lemma}
\begin{proof}
	Given $D \in \mathscr{D}$, 
	by construction of $\mathscr{D}$ we know that~$D$ is compact and thus contained in the interior of some~$E \in \mathscr{D}$ with $\rho(E) < \infty$ in view of~\eqref{(infty)} and~\eqref{(weak convergence rho tilde)}. Therefore, $\rho|_E : \B(E) \to [0, \infty)$ is a nonnegative finite Borel measure, where~$\B(E)$ denotes the Borel $\sigma$-algebra on the topological space~$E$. Since $E$ is metrizable, it is completely regular, implying that the class $\Gamma_{\rho|_E}$ of all Borel sets in $E$ with boundaries of $\rho|_E$-measure zero contains a base (consisting of open sets) of the topology of $E$ (for details see~\cite[Proposition 8.2.8]{bogachev}). Since $E$ is compact, the set $D \subset E$ can be covered by finitely many relatively compact, open sets $V_1, \ldots, V_N \subset E^{\circ}$ in $\Gamma_{\rho|_E}$. By construction, the closure of the set $V := \bigcup_{i=1}^N V_i \subset E^{\circ}$ is compact. Considering the restriction $\rho^{(n)}|_V$, for any $f \in C_c(V)$ we obtain
	\begin{align*}
		\lim_{n \to \infty} \int_V f \: d\rho^{(n)}|_V = \lim_{n \to \infty} \int_{\F} f \: d\rho^{(n)}|_{E^{\circ}} \stackrel{\eqref{(vague convergence D)}}{=} \int_{\F} f \: d\rho|_{E^{\circ}} = \int_V f \: d\rho|_V \:,
	\end{align*}
	proving vague convergence 
	\begin{align}\label{(vague convergence V)}
		\rho^{(n)}|_V \stackrel{v}{\to} \rho|_V \:. 
	\end{align}
	Since $\Gamma_{\rho|_E}$ is a subalgebra in $\B(E)$ (see \cite[Proposition 8.2.8]{bogachev}), the set $V$ is also contained in $\Gamma_{\rho|_E}$, implying that $\rho(\partial V) = \rho|_E(\partial V) = 0$. Note that the measures $(\rho^{(n)}|_{E^{\circ}})$ as well as $\rho|_{E^{\circ}}$ are regular Borel measures and thus Radon \cite{schwartz-radon}. Therefore, making use of vague convergence \eqref{(vague convergence D)} and applying \cite[Theorem 30.2]{bauer}, for the relatively compact, open set $V \subset E^{\circ}$ and the compact set $\overline{V} \subset E^{\circ}$ we obtain
	\begin{align*}
		\rho(V) = \rho(\overline{V}) \ge \limsup_{n \to \infty} \rho^{(n)}(\overline{V}) \ge \limsup_{n \to \infty} \rho^{(n)}(V) \ge \liminf_{n \to \infty} \rho^{(n)}(V) \ge \rho(V) \:,
	\end{align*}
	proving that
	\begin{align}\label{(4.19)}
		\rho(V) = \lim_{n \to \infty} \rho^{(n)}(V) \:.
	\end{align}
	Let us point out that, for each $n \in \N$, the measure $\rho^{(n)}|_V/\rho^{(n)}(V)$ is normalized in the sense that $\rho^{(n)}|_V(V)/\rho^{(n)}(V) = 1$ (cf.~\cite[\S 3.2]{noncompact}). Furthermore, applying vague convergence \eqref{(vague convergence V)} as well as~\eqref{(4.19)}, for any $f \in C_c(V)$ we are given 
	\begin{align*}
		\lim_{n \to \infty} \int_V f \: d\rho^{(n)}|_V/\rho^{(n)}(V) = \int_V f \: d\rho|_V/\rho(V) \:.
	\end{align*}
	As a consequence, the sequence of normalized measures~$(\rho^{(n)}|_V/\rho^{(n)}(V))_{n \in \N}$ converges vaguely to the normalized measure~$\rho|_V/\rho(V)$, and from \cite[Corollary 30.9]{bauer} we deduce that~$(\rho^{(n)}|_V/\rho^{(n)}(V))_{n \in \N}$ converges \emph{weakly} to the normalized measure~$\rho|_V/\rho(V)$. Thus in view of 
	\begin{align*}
		\lim_{n \to \infty} \int_{V} f \: d\rho^{(n)}|_V &= \lim_{n \to \infty} \rho^{(n)}(V) \int_{V} f \: d\rho^{(n)}|_V/\rho^{(n)}(V) = \rho(V) \int_{V} f \: d\rho^{(n)}|_V/\rho(V) \\
		&= \int_{V} f \: d\rho^{(n)}|_V
	\end{align*}
	for any $f \in C_b(V)$, we finally obtain weak convergence $\rho^{(n)}|_V \rightharpoonup \rho|_V$. 
\end{proof}

By contrast to \cite{noncompact}, it is not reasonable to consider vague convergence $\rho^{(n)} \stackrel{v}{\to} \rho$ in view of~\cite[Exercise 14.4]{koenig}.

\section{Minimizers for Lagrangians of Bounded Range}\label{Section minimizers bounded range}
\subsection{Preliminaries}\label{S Preliminaries Condition B} 
This section is devoted to the proof that, under suitable assumptions, the measure~$\rho$ as defined in \eqref{(rho)} is a minimizer of the causal variational principle~\eqref{(cvp)} under variations of finite volume (see Definition \ref{Definition finite volume}). 
This is accomplished in \S \ref{S finite volume}. 
To this end, we proceed in several steps. In this subsection, we introduce an additional assumption on the measure~$\rho$ obtained in Theorem~\ref{Theorem measure} (see~\eqref{(assumption B)}). Afterwards we prove that $\rho$ is a minimizer on suitable compact subsets (see \S \ref{S minimizer compact} and~\S \ref{S compact support}). 
Assuming that $\rho \not= 0$ is locally finite, we finally show that~$\rho$ satisfies corresponding Euler-Lagrange (EL) equations (see~\S \ref{S EL}), which have the same structure as the EL equations obtained in~\cite{noncompact}. 
Throughout this section, we shall assume that the Lagrangian is of bounded range (see Definition~\ref{Definition bounded range}). 

In order to prove that $\rho$ is a minimizer, we impose the following condition:  
\begin{itemize}
	\item[(B)] For any $\varepsilon > 0$ and $B \subset \F$ bounded, there exists $N \in \N$ such that 
	\begin{align}\label{(assumption B)}
	\rho(B \setminus B^{(n)}) < \varepsilon \qquad \text{for all $n \ge N$} \:,
	\end{align}
	where $B^{(n)} := B \cap \F^{(n)}$. 
\end{itemize}

\begin{Lemma}\label{Lemma B locally finite}
	Assume that the measure $\rho$ defined by \eqref{(rho)} satisfies condition (B). Then the measure~$\rho$ is locally finite, and any bounded subset of $\F$ has finite $\rho$-measure. 
\end{Lemma}
\begin{proof}
	Assuming that $B \subset \F$ is bounded, in view of condition~(B) there exists~$N \in \N$ such that $\rho(B \setminus B^{(n)}) < \varepsilon$ for all $n \ge N$, where~$B^{(n)} := B \cap \F^{(n)}$ 
	is relatively compact in view of~\cite[Bemerkungen 2.9]{alt}. For this reason, $B^{(n)}$ can be covered by a finite number of compact sets $D_1, \ldots, D_L$ with $D_i \in \mathscr{D}^{(n)}$ for all~$i = 1, \ldots, L$, where~$\mathscr{D}^{(n)}$ is given by~\eqref{(Dn)}. From \eqref{(infty)} we obtain~$\rho(B^{(n)}) < \infty$, implying that
	\begin{align*}
	\rho(B) = \rho(B \setminus B^{(n)}) + \rho(B^{(n)}) < \infty \:.
	\end{align*}
	Since each compact set $K \subset \F$ is bounded, the measure $\rho$ is locally finite in view of Lemma \ref{Lemma locally finite}. 
\end{proof}

\subsection{Minimizers under Variations of Finite-Dimensional Compact Support}\label{S minimizer compact}
Before proving our first existence result, we point out that the restricted Lagrangian 
$$\L^{(n)} = \L|_{\F^{(n)} \times \F^{(n)}} : \F^{(n)} \times \F^{(n)} \to \R_0^+$$ 
is of compact range (see Definition \ref{defcompactrange} and \S \ref{S Approximating}) for every $n \in \N$. 
Therefore, for all~$j,k,n \in \N$, there exist compact subsets~$(V_{j,k}^{(n)})' \subset \F^{(n)}$ such that~$\L(x,y) = 0$ for all~$x \in V_{j,k}^{(n)}$ and~$y \in \F^{(n)} \setminus (V_{j,k}^{(n)})'$. 
By construction of $\mathscr{D}$ (see \S \ref{S Construction of Countable Set}), the set~$(V_{j,k}^{(n)})'$ can be covered by a finite number of sets~$(V_{j'_{\ell},k'_{\ell}}^{(n)})_{\ell = 1, \ldots, L}$ in~$\mathscr{D}$, 
whose union is also contained in~$\mathscr{D}$. For this reason, we may assume 
that~$(V_{j,k}^{(n)})' \in \mathscr{D}$ for all~$j,k,n \in \N$. 

After these preparations, we are now in the position to state our first existence result. 

\begin{Prp}\label{Proposition minimizer on V}
	Assume that the Lagrangian $\L \in C_b(\F \times \F; \R_0^+)$ is of bounded range, and that condition~\eqref{(ivn)} holds. Moreover, assume that the measure~${\rho}$ defined by~\eqref{(rho)} satisfies condition~{\rm{(B)}} in Section~\ref{Section minimizers bounded range}. 
	Then $\rho$ is a \textbf{minimizer under variations in~$\mathscr{D}$} in the sense that
	\[ \big( \Sact(\tilde{\rho}) - \Sact(\rho) \big) \geq 0 \] 
	whenever $\tilde{\rho}$ satisfying~\eqref{totvol} is a regular measure on $\B(\F)$ with $\supp (\tilde{\rho} - \rho) \in \mathscr{D}$. 
\end{Prp}
\begin{proof}
	Assume that $\tilde{\rho} : \B(\F) \to [0, \infty]$ is a measure on the Borel $\sigma$-algebra of $\F$ with~$D := \supp (\tilde{\rho} - \rho) \in \mathscr{D}$ such that \eqref{totvol} is satisfied, i.e. 
	\[0 < \tilde{\rho}(D) = \rho(D) < \infty \:. \]
	Since $D \in \mathscr{D}$ is compact, the difference \eqref{integrals} is well-defined in view of \cite[\S 4.3]{noncompact}, 
	\begin{align*}
	\Sact(\tilde{\rho}) - \Sact(\rho) &= \int_{\F} d(\tilde{\rho} - \rho)(x) \int_{\F} d\rho(y) \:\L(x,y) + \int_{\F} d\rho(x) \int_{\F} d(\tilde{\rho} - \rho)(y) \:\L(x,y) \\
	&\qquad \quad + \int_{\F} d(\tilde{\rho}- \rho)(x) \int_{\F} d(\tilde{\rho}- \rho)(y) \:\L(x,y) \:.
	\end{align*}
	Making use of the symmetry of the Lagrangian and applying Fubini's theorem, we can write this expression as
	\begin{align*}
	\Sact(\tilde{\rho}) - \Sact(\rho) &= 2\int_{\F} d(\tilde{\rho} - \rho)(x) \int_{\F} d\rho(y) \:\L(x,y) \\
	&\qquad \quad + \int_{\F} d(\tilde{\rho}- \rho)(x) \int_{\F} d(\tilde{\rho}- \rho)(y) \:\L(x,y) \:.
	\end{align*}

	Since $D \in \mathscr{D}$ is a bounded subset of $\F$, Definition \ref{Definition bounded range} yields the existence of some bounded set $W \subset \F$ such that $\L(x,y) = 0$ for all $x \in D$ and $y \in \F \setminus W$. 
	For arbitrary~$\tilde{\varepsilon} > 0$, by virtue of condition~{\rm{(B)}} in Section~\ref{Section minimizers bounded range} (see \eqref{(assumption B)}) there exists some integer~$n' = n'(D)$ such that~$\rho(W \setminus W^{(n)}) < \tilde{\varepsilon}$ for all~$n \ge n'$. We choose~$n' \in \N$ sufficiently large and let~$\tilde{D} \in \mathscr{D}$ with~$\tilde{D} \subset \F^{(n')}$ and~$\tilde{D} \supset D$. 
	As explained at the beginning of~\S \ref{S minimizer compact}, 
	there exists~$\tilde{D}' \in \mathscr{D}$ with~$\tilde{D}' \subset \F^{(n')}$ and~$\L(x,y) = 0$ for all $x \in \tilde{D}$ and~$y \in \F^{(n')} \setminus \tilde{D}'$. 
	In particular, $\L(x,y) = 0$ for all $x \in D$ and~$y \in \F^{(n')} \setminus \tilde{D}'$. 	
	In view of Lemma~\ref{Lemma weak convergence}, there exist relatively compact open sets~$V, V' \subset \F^{(n')}$ with~$V \supset \tilde{D}$ and~$V' \supset \tilde{D}'$ such that $\L(x,y) = 0$ for all~$x \in V$ and $y \in \F^{(n')} \setminus V'$, and $\rho^{(n)}|_V \rightharpoonup \rho|_V$ as well as~$\rho^{(n)}|_{V'} \rightharpoonup \rho|_{V'}$. 
	These considerations give rise to  
	\begin{align*}
	&\Sact(\tilde{\rho}) - \Sact(\rho) = 2\int_V d(\tilde{\rho}- \rho)(x) \int_{V'} d\rho(y) \:\L(x,y) \\
	&\qquad + 2\int_V d(\tilde{\rho}- \rho)(x) \int_{\F \setminus V'} d\rho(y) \:\L(x,y) + \int_V d(\tilde{\rho} - \rho)(x) \int_V d(\tilde{\rho}- \rho)(y) \:\L(x,y) \:,
	\end{align*}
	where the expression
	\begin{align*}
		\int_V d(\tilde{\rho}- \rho)(x) \int_{\F \setminus V'} d\rho(y) \:\L(x,y) = \int_D d(\tilde{\rho}- \rho)(x) \int_{W \setminus W^{(n')}} d\rho(y) \:\L(x,y) 
	\end{align*}
	can be chosen arbitrarily small for sufficiently large $n' \in \N$. Thus for $\varepsilon > 0$ arbitrary, we may arrange that
	\begin{align*}
	\Sact(\tilde{\rho}) - \Sact(\rho) &\ge 2\int_V d(\tilde{\rho}- \rho)(x) \int_{V'} d\rho(y) \:\L(x,y) \\
	&\qquad \quad + \int_V d(\tilde{\rho} - \rho)(x) \int_V d(\tilde{\rho}- \rho)(y) \:\L(x,y) - \varepsilon \:. 
	\end{align*}
	
	Next, in analogy to the proof of \cite[Theorem 4.9]{noncompact}, for any~$n \in \N$ we introduce the measures $\tilde{\rho}_n : \B(\F) \to [0, \infty]$ by 
	\begin{align*}
	\tilde{\rho}_n := \left\{ \begin{array}{cl}
	c_n \:\tilde{\rho} \qquad &\text{on $V$} \\[0.2em]
	{\rho}^{(n)} \qquad &\text{on $\F \setminus V$}
	\end{array} \right. \qquad \text{with} \qquad c_n := \frac{{\rho}^{(n)}(V)}{\tilde{\rho}(V)} \qquad \text{for all $n \in \N$} \:.
	\end{align*} 
	Since $\rho$ and $\tilde{\rho}$ coincide on $\F \setminus D$, by virtue of \eqref{(4.19)} (see Lemma \ref{Lemma weak convergence}) we obtain 
	\begin{align*}
	\lim_{n \to \infty} c_n = \lim_{n \to \infty} \frac{\rho^{(n)}(V)}{\tilde{\rho}(V)} = \frac{\rho(V)}{\tilde{\rho}(V)} = \frac{\rho(V \setminus D) + \rho(D)}{\tilde{\rho}(V \setminus D) + \tilde{\rho}(D)} = 1 \:. 
	\end{align*}
	Making use of the fact that $V \subset \F$ is separable (see for instance~\cite[Corollary 3.5]{aliprantis}) and applying~\cite[Theorem 2.8]{billingsley} in a similar fashion to the proof of \cite[Theorem 4.9]{noncompact}, we thus 
	arrive at
	\begin{align*}
	\Sact(\tilde{\rho}) - \Sact(\rho) &\ge \lim_{n \to \infty} \left[2 \int_V d\big(c_n \:\tilde{\rho}- \rho^{(n)}\big)(x) \int_{V'} d\rho^{(n)}(y) \:\L(x,y) \right. \\
	&\qquad \qquad  \left. + \int_V d\big(c_n \:\tilde{\rho} - \rho^{(n)}\big)(x) \int_V d\big(c_n \:\tilde{\rho}- \rho^{(n)}\big)(y) \: \L(x,y) \right] - 2 \varepsilon \:.
	\end{align*}
	In view of the fact that $\tilde{\rho}_n$ and $\rho^{(n)}$ coincide on $\F^{(n)} \setminus V$ for sufficiently large~$n \in \N$ and $\L(x,y) = 0$ for all $x \in V$ and~$y \notin V'$, the difference $\Sact(\tilde{\rho}) - \Sact(\rho)$ can finally be estimated by 
	\begin{align*}
	\Sact(\tilde{\rho}) - \Sact(\rho) &\ge \lim_{n \to \infty} \left[2 \int_{\F^{(n)}} d\big(\tilde{\rho}_n- \rho^{(n)}\big)(x) \int_{\F^{(n)}} d\rho^{(n)}(y) \:\L(x,y) \right. \\
	&\qquad \qquad \left. + \int_{\F^{(n)}} d\big(\tilde{\rho}_n - \rho^{(n)}\big)(x) \int_{\F^{(n)}} d\big(\tilde{\rho}_n - \rho^{(n)} \big)(y) \:\L(x,y) \right] - 2 \varepsilon \:. 
	\end{align*}
	Since $\rho^{(n)}$ is a minimizer on $\F^{(n)}$ for every $n \in \N$ (see \S \ref{S Approximating} and \S \ref{S Global Measure}), we are given
	\begin{align}\label{(ge')}
	\big(\Sact_{\F^{(n)}}(\tilde{\rho}_n) - \Sact_{\F^{(n)}}(\rho^{(n)})\big) \ge 0 \qquad \text{for all $n \in \N$} \:.
	\end{align}
	Taking the limit $n \to \infty$ on the left hand side of \eqref{(ge')}, one obtains exactly the above expression in square brackets for~$\Sact(\tilde{\rho}) - \Sact(\rho)$. Since $\varepsilon > 0$ is arbitrary, this implies that
	\begin{align*}
	\left(\Sact(\tilde{\rho}) - \Sact(\rho)\right) \ge 0 \:.
	\end{align*}
	Hence $\rho$ is a minimizer under variations in~$\mathscr{D}$.  
\end{proof}

Our next goal is to extend the previous result to minimizers under variations of finite-dimensional compact support, which is defined as follows. 

\begin{Def}\label{Definition finite-dimensional support}
	A regular measure $\rho$ on $\B(\F)$ is said to be a \textbf{minimizer under variations of finite-dimensional compact support} if the inequality 
	\begin{align*}
	\big(\Sact(\tilde{\rho}) - \Sact(\rho) \big) \ge 0
	\end{align*}
	holds for any regular measure $\tilde{\rho}$ on $\B(\F)$ satisfying \eqref{totvol} with the following property: There exists~$n' \in \N$ such that, for all $n \ge n'$,  
	\[\text{$\supp \big(\tilde{\rho} - \rho \big) \subset \F^{(n)}$ compact} \qquad \text{and} \qquad \supp \big(\tilde{\rho} - \rho \big) \cap \F \setminus \F^{(n)} = \varnothing \:. \] 
\end{Def}

Based on this definition, we may state the following existence result. 

\begin{Prp}\label{Proposition finite-dimensional}
	Assume that $\L \in C_b(\F \times \F; \R_0^+)$ is of bounded range, and assume that condition~\eqref{(ivn)} holds. Furthermore, assume that the measure~${\rho}$ defined by \eqref{(rho)} satisfies condition~{\rm{(B)}} in Section~\ref{Section minimizers bounded range}. 
	Then $\rho$ is a minimizer under variations of finite-dimensional compact support.  
\end{Prp}
\begin{proof}
	Assuming that $\tilde{\rho}$ is a variation of finite-dimensional compact support, there exists $n' \in \N$ such that 
	\[\text{$K_{\sharp} := \supp \big(\tilde{\rho} - \rho \big) \subset \F^{(n)}$} \]
	is compact for all $n \ge n'$, implying that 
	$\tilde{\rho}(K_{\sharp}) = \rho(K_{\sharp}) < \infty$ (in virtue of Lemma~\ref{Lemma B locally finite}). 
	According to Lemma \ref{Lemma inner regular}, the measure $\rho$ given by~\eqref{(rho)} is regular, and by
	regularity of~${\rho}$ and~$\tilde{\rho}$, for $\tilde{\varepsilon} > 0$ arbitrary we may choose~$U \supset K_{\sharp}$ open such that
	\[{\rho}(U \setminus K_{\sharp}) < \tilde{\varepsilon} \qquad \text{and} \qquad \tilde{\rho}(U \setminus K_{\sharp}) < \tilde{\varepsilon} \:. \]
	By construction of $\mathscr{D}$, there is some compact set~$D \in \mathscr{D}$ such that $K_{\sharp} \subset D \subset U^{(n')}$, where $U^{(n')} := U \cap \F^{(n')}$. In particular,  
	\[{\rho}(D \setminus K_{\sharp}) < \tilde{\varepsilon} \qquad \text{and} \qquad \tilde{\rho}(D \setminus K_{\sharp}) < \tilde{\varepsilon} \:. \]
	Applying Lemma \ref{Lemma weak convergence} yields the existence of some relatively compact set $V \subset \F$ with~$D \subset V \subset U^{(n')}$ such that \eqref{(weak convergence V)} holds. Moreover, 
	\[{\rho}(V \setminus K_{\sharp}) < \tilde{\varepsilon} \qquad \text{and} \qquad \tilde{\rho}(V \setminus K_{\sharp}) < \tilde{\varepsilon} \:. \]
	Since $\tilde{\rho} - \rho$ is a signed measure of finite total variation and compact support, the difference \eqref{integrals} is well-defined (cf. \cite[\S 4.3]{noncompact}), 
	\begin{align*}
	\left(\Sact(\tilde{\rho}) - \Sact(\rho)\right) &= 2\int_{K_{\sharp}} d(\tilde{\rho} - \rho)(x) \int_{\F} d\rho(y) \:\L(x,y) \\
	&\qquad \quad + \int_{K_{\sharp}} d(\tilde{\rho}- \rho)(x) \int_{K_{\sharp}} d(\tilde{\rho}- \rho)(y) \:\L(x,y) \:. 
	\end{align*}
	By adding and subtracting the terms
	\begin{align*}
	2 \int_{V \setminus K_{\sharp}} d(\tilde{\rho} - \rho)(x) \int_{\F} d\rho(y) \: \L(x,y) + \int_{V \setminus K_{\sharp}} d(\tilde{\rho} - \rho)(x) \int_{K_{\sharp}} d(\tilde{\rho} - \rho)(y) \: \L(x,y)
	\end{align*}
	as well as
	\[ \int_{V} d(\tilde{\rho} - \rho)(x) \int_{V \setminus K_{\sharp}} d(\tilde{\rho} - \rho)(y) \: \L(x,y) \:, \]
	one easily verifies that
	\begin{align*}
	&\big(\Sact(\tilde{\rho}) - \Sact(\rho) \big) = 2 \int_{V} d(\tilde{\rho} - \rho)(x) \int_{\F} d\rho(y) \: \L(x,y) \\
	&\quad \!\! + \int_{V} d(\tilde{\rho} - \rho)(x) \int_V d(\tilde{\rho} - \rho)(y) \: \L(x,y) - \left[ \int_{V} d(\tilde{\rho} - \rho)(x) \int_{V \setminus K_{\sharp}} d(\tilde{\rho} - \rho)(y) \: \L(x,y) \right. \\
	&\quad \!\! + \left. 2 \int_{V \setminus K_{\sharp}} d(\tilde{\rho} - \rho)(x) \int_{\F} d\rho(y) \: \L(x,y) + \int_{V \setminus K_{\sharp}} d(\tilde{\rho} - \rho)(x) \int_{K_{\sharp}} d(\tilde{\rho} - \rho)(y) \: \L(x,y)\right] \:. 
	\end{align*}
	Since the Lagrangian $\L$ is of bounded range (see Definition~\ref{Definition bounded range}), its restriction $\L^{(n')}$ is of compact range, implying that $\L(x,y) = 0$ for all~$x \in V$ and~$y \notin V'$ for some relatively compact, open set~$V' \subset \F^{(n')}$ such that \eqref{(weak convergence V)} holds. 
	Choosing the open set~$U \supset K_{\sharp}$ suitably, one thus can arrange that 
	\begin{align*}
	\left|\int_{V \setminus K_{\sharp}} d(\tilde{\rho} - \rho)(x) \int_{\F} d\rho(y) \: \L(x,y) \right| \le \underbrace{\sup_{x,y \in V'} \L(x,y) \: \rho(V')}_{\text{$< \infty$}} \underbrace{\big(\left|\tilde{\rho}(V \setminus K_{\sharp})\right| + \left|\rho(V \setminus K_{\sharp}) \right|\big)}_{\text{$< \tilde{\varepsilon}$}}
	\end{align*}
	is arbitrarily small. 
	Applying similar arguments to all summands of the above term in square brackets, one obtains the estimate
	\begin{align*}
	\big(\Sact(\tilde{\rho}) - \Sact(\rho) \big) &\ge \left\{ 2 \int_{V} d(\tilde{\rho} - \rho)(x) \int_{\F} d\rho(y) \: \L(x,y) \right. \\
	& \qquad \qquad \left. + \int_{V} d(\tilde{\rho} - \rho)(x) \int_V d(\tilde{\rho} - \rho)(y) \: \L(x,y) \right\} - \varepsilon 
	\end{align*}
	for any given $\varepsilon > 0$. 
	Proceeding in analogy to the proof of Proposition~\ref{Proposition minimizer on V} by applying weak convergence~\eqref{(weak convergence V)} together with \cite[Theorem 2.8]{billingsley} (for details we refer to the proof of \cite[Theorem 4.9]{noncompact}), one can show that the term in curly brackets is greater than or equal to zero, up to an arbitrarily small error term. Since $\varepsilon > 0$ was chosen arbitrarily, we finally arrive at
	\begin{align*}
	\big(\Sact(\tilde{\rho}) - \Sact(\rho) \big) \ge 0 \:,
	\end{align*}
	which proves the claim. 
\end{proof}

\subsection{Existence of Minimizers under Variations of Compact Support}\label{S compact support}
Having derived the above preparatory results in \S \ref{S minimizer compact}, we are now in the position to deal with minimizers under variations of compact support (see Definition \ref{Definition compact support}). For clarity, we point out that, whenever~$\tilde{\rho}$ is a variation of compact support of the measure $\rho$ satisfying~\eqref{totvol}, the condition~$|\tilde{\rho} - \rho| < \infty$ yields~$\rho(K) = \tilde{\rho}(K) < \infty$, where the compact set $K \subset \F$ is defined by~$K := \supp (\tilde{\rho} - \rho)$ (for details see Definition~\ref{deffv} and the explanations in \S \ref{seccvpsigma}). 

\begin{Lemma}\label{Lemma compact support}
	Assume that $\L \in C_b(\F \times \F; \R_0^+)$ is of bounded range, and assume that condition~\eqref{(ivn)} holds. Moreover, assume that the measure~${\rho}$ defined by \eqref{(rho)} satisfies condition~{\rm{(B)}} in Section~\ref{Section minimizers bounded range}. 
	Then $\rho$ is a minimizer under variations of compact support. 
\end{Lemma}
\begin{proof}
	Let $\tilde{\rho}$ be a variation of compact support. Then the set~$K := \supp (\tilde{\rho} - \rho) \subset \F$ is compact, and~$\rho(K) = \tilde{\rho}(K) < \infty$ according to \eqref{totvol}.  
	Given $\tilde{\varepsilon} > 0$ arbitrary, by regularity of $\rho$, $\tilde{\rho}$ there exists~$U \supset K$ open such that $\rho(U), \tilde{\rho}(U) < \infty$ and 
	\begin{align*}
	\rho(U \setminus K) < \tilde{\varepsilon} \qquad \text{and} \qquad \tilde{\rho}(U \setminus K) < \tilde{\varepsilon} \:.
	\end{align*}
	Moreover, in view of \eqref{(regular)}, there exists $D \in \mathscr{D}$ such that
	\begin{align*}
	\rho(U \setminus D) < \tilde{\varepsilon} \qquad \text{and} \qquad \tilde{\rho}(U \setminus D) < \tilde{\varepsilon} \:.
	\end{align*}
	Since $K \subset \F$ is compact, the difference \eqref{integrals} is well-defined (cf. \cite[\S 4.3]{noncompact}), 
	\begin{align*}
	\left(\Sact(\tilde{\rho}) - \Sact(\rho)\right) &= 2\int_{K} d(\tilde{\rho} - \rho)(x) \int_{\F} d\rho(y) \:\L(x,y) \\
	&\qquad \quad + \int_{K} d(\tilde{\rho}- \rho)(x) \int_{K} d(\tilde{\rho}- \rho)(y) \:\L(x,y) \:. 
	\end{align*}
	Note that, by definition of $\mathscr{D}$, each $D \in \mathscr{D}$ is the finite union of finite-dimensional subsets of~$(\F^{(n)})_{n \in \N}$ (cf.~\S \ref{S Construction of Countable Set}). As a consequence, for each $D \in \mathscr{D}$ there exists $n' \in \N$ such that $D \subset \F^{(n)}$ for all $n \ge n'$. Moreover, by construction of $\mathscr{D}$ there exists $E \in \mathscr{D}$ such that $D \subset E^{\circ} \subset U^{(n)}$, where $U^{(n)} = U \cap \F^{(n)}$, and according to Lemma \ref{Lemma weak convergence} there exists a relatively compact, open set $V \subset E^{\circ}$ such that $D \subset V$. In particular, 
	\[\rho(U \setminus V) < \tilde{\varepsilon} \:, \qquad \tilde{\rho}(U \setminus V) < \tilde{\varepsilon} \:.  \]
	Thus adding and subtracting the terms
	\begin{align*}
	2 \int_{U \setminus K} d(\tilde{\rho} - \rho)(x) \int_{\F} d\rho(y) \: \L(x,y) + \int_{U \setminus K} d(\tilde{\rho} - \rho)(x) \int_{K} d(\tilde{\rho} - \rho)(y) \: \L(x,y)
	\end{align*}
	as well as
	\[ \int_{U} d(\tilde{\rho} - \rho)(x) \int_{U \setminus K} d(\tilde{\rho} - \rho)(y) \: \L(x,y) \:, \] 
	and proceeding in analogy to the proof of Proposition \ref{Proposition finite-dimensional}, we conclude that 
	\begin{align*}
	\big(\Sact(\tilde{\rho}) - \Sact(\rho) \big) \ge 0 \:. 
	\end{align*}
	Hence $\rho$ is indeed a minimizer under variations of compact support. 
\end{proof}

\subsection{Existence of Minimizers under Variations of Finite Volume}\label{S finite volume} We now proceed similarly to~\cite[\S 4.4]{noncompact} in order to prove the existence of minimizers under variations of finite volume (see Definition \ref{Definition finite volume}). For the difference of the actions \eqref{integrals} to be well-defined, we require the additional property~{\rm{(iv)}} in \S \ref{seccvpsigma}, i.e. 
\begin{align*}
	\sup_{x \in \F} \int_{\F}\L(x,y)\: d\rho(y) < \infty \:.
\end{align*}
Then we can state the following result.
\begin{Thm}\label{Theorem finite volume}
	Assume that $\L \in C_b(\F \times \F; \R_0^+)$ is of bounded range, and assume that condition~\eqref{(ivn)} holds. Furthermore, assume that the measure~${\rho}$ defined by~\eqref{(rho)} satisfies condition~{\rm{(B)}} in Section~\ref{Section minimizers bounded range}. 
	Then $\rho$ is a minimizer under variations of finite volume. 
\end{Thm}
\begin{proof}
	Assume that $\tilde{\rho}$ be a variation of finite volume satisfying \eqref{totvol}. Introducing the set~$B := \supp (\tilde{\rho} - \rho)$, we thus obtain $\rho(B) = \tilde{\rho}(B) < \infty$. Given $\tilde{\varepsilon} > 0$ arbitrary, by regularity of $\rho$, $\tilde{\rho}$ there exists~$U \supset B$ open such that 
	\begin{align*}
	\rho(U \setminus B) < \tilde{\varepsilon} \qquad \text{and} \qquad \tilde{\rho}(U \setminus B) < \tilde{\varepsilon} \:.
	\end{align*}
	Making use of the additional assumption that condition (iv) in \S \ref{seccvpsigma} is satisfied, the difference of the actions \eqref{integrals} is well-defined. 
	Therefore, proceeding in analogy to the proof of Lemma \ref{Lemma compact support} finally gives rise to 
	\begin{align*}
	\big(\Sact(\tilde{\rho}) - \Sact(\rho) \big) \ge 0 \:,
	\end{align*}
	which implies that $\rho$ is a minimizer under variations of finite volume.  
\end{proof}

The remainder of this section is devoted to the derivation of the corresponding EL equations for minimizers under variations of finite volume (\S \ref{S EL}).

\subsection{Derivation of the Euler-Lagrange Equations}\label{S EL}
The strategy in \cite{noncompact} was to derive the EL equations in order to prove the existence of minimizers under variations of finite volume. However, proceeding similar to \cite[\S 4.2]{noncompact} does not seem promising in the infinite-dimensional setting. For this reason, we rather proceed in the opposite direction by first proving the existence of minimizers and then deriving the EL equations. More precisely, under the assumptions that condition (iv) in \S \ref{seccvpsigma} as well as condition~{\rm{(B)}} in Section~\ref{Section minimizers bounded range} are satisfied, 
Theorem \ref{Theorem finite volume} shows that the measure $\rho$ given by~\eqref{(rho)} is a minimizer under variations of finite volume (see Definition \ref{Definition finite volume}). Under these assumptions, Lemma \ref{Lemma iv Borel} implies that $\rho$ is locally finite. 
This allows us to proceed similarly to the proof of~\cite[Lemma 2.3]{jet}, thus giving rise to corresponding EL equations. For convenience, let us state the latter result in greater generality.

\begin{Thm}[\textbf{The Euler-Lagrange equations}]\label{Theorem EL}
	Let $\F$ be topological Hausdorff space, let $\rho$ be a Borel measure on $\F$ (in the sense of~\cite{gardner+pfeffer}, i.e.\ a locally finite measure) and assume that $\L : \F \times \F \to \R_0^+$ is symmetric and lower semi-continuous. If $\rho \not= 0$ is a 
	minimizer of the causal variational principle \eqref{(cvp)}, \eqref{totvol} under variations of finite volume, then the \emph{Euler-Lagrange equations}
	\begin{align}\label{(EL)}
	\ell|_{\supp \rho} \equiv \inf_{x \in \F} \ell(x) 
	\end{align}
	hold, where the mapping $\ell : \F \to [0, \infty)$ is defined by 
	\begin{align}\label{(ell)}
	\ell(x) := \int_{\F} \L(x,y) \: d\rho(y) - \mathfrak{s}
	\end{align}
	for some parameter~$\mathfrak{s} \in \R$. 
\end{Thm}
\begin{proof}
	Proceed in analogy to the proof of \cite[Lemma 2.3]{jet}.
\end{proof}

For clarity, we point out that Theorem \ref{Theorem EL} requires that $\rho$ is locally finite. Choosing the parameter~$\mathfrak{s}$ suitably, one can arrange that the infimum in \eqref{(EL)} vanishes: 

\begin{Lemma}\label{Lemma EL0}
	Assume that the measure~$\rho$ given by \eqref{(rho)} is non-zero. 
	Then, under the assumptions of Theorem \ref{Theorem finite volume}, 
	for a suitable choice of~$\mathfrak{s} \ge 0$ in~\eqref{(ell)} the Euler-Lagrange equations~\eqref{(EL)} read  
	\begin{align}\label{(EL0)}
	\ell|_{\supp \rho} \equiv \inf_{x \in \F} \ell(x) = 0 \:.
	\end{align}
\end{Lemma}
\begin{proof}
	Assuming that $\rho \not= 0$, we conclude that $\supp \rho \not= \varnothing$. 
	Moreover, under the assumptions of Theorem \ref{Theorem finite volume}, by Lemma \ref{Lemma iv Borel} we know that~$\rho$ is locally finite. Next, from Theorem \ref{Theorem finite volume} and Theorem \ref{Theorem EL} we deduce that~$\rho$ satisfies the EL equations~\eqref{(EL)}. By assumption, condition~{\rm{(iv)}} in~\S \ref{seccvpsigma} is satisfied, implying that 
	\begin{align*}
	0 \le \int_{\F} \L(x,y) \: d\rho(y) < \infty \qquad \text{for every $x \in \F$} \:. 
	\end{align*}
	This allows us to choose $\mathfrak{s} \ge 0$ such that \eqref{(EL0)} holds.
\end{proof}

Lemma \ref{Lemma EL0} generalizes the results of \cite[Section 2]{jet} and \cite[Section 4]{noncompact} 
to the infinite-dimensional setting. 
It remains an open task to prove the existence of a Lagrangian of bounded range such that the measure~$\rho$ in \eqref{(rho)} is non-zero and satisfies condition~{\rm{(iv)}} in~\S \ref{seccvpsigma} as well as condition~{\rm{(B)}} in Section~\ref{Section minimizers bounded range}.

\section{Minimizers for Lagrangians Vanishing in Entropy}\label{Section Decay}	
In Section~\ref{Section minimizers bounded range} the results from \cite[Section 4]{noncompact} were generalized to the non-locally compact setting.  
This raises the question whether it is possible also to weaken the assumption that the Lagrangian is of bounded range (see Definition \ref{Definition bounded range}) similarly to Lagrangians decaying in entropy as introduced in \cite[Section 5]{noncompact}. It is precisely the objective of this section to analyze this question in detail. To this end, we first generalize the notion of Lagrangians decaying in entropy (\S \ref{S Entropy}). Afterwards we proceed similarly as in \cite[Section 5]{noncompact} and Section \ref{Section minimizers bounded range} to prove that, under suitable assumptions, the measure $\rho$ obtained in Theorem \ref{Theorem measure} is a minimizer under variations of compact support and variations of finite volume (see Definition \ref{Definition compact support} and Definition \ref{Definition finite volume}). 

\subsection{Lagrangians Vanishing in Entropy}\label{S Entropy}
This subsection is devoted to generalize the notion of Lagrangians decaying in entropy as introduced in \cite[Section~5]{noncompact}. More precisely, 
in order for the constructions in \cite[Section 5]{noncompact} to work, the definition of Lagrangians decaying in entropy (see Definition~\ref{Definition vanishing} and~\cite[Definition~5.1]{noncompact}) requires an 
\emph{un\-bounded} Heine-Borel metric. 
On the other hand, any Heine-Borel space (that is, a topological space endowed with a Heine-Borel metric) is $\sigma$-compact and locally compact, see~\cite{williamson}. In particular, every separable Heine-Borel space~$X$ 
is a second-countable, locally compact Hausdorff space,
and hemicompact (see~\cite[Problem~17I]{willard})
in view of~\cite[Exercise~3.8.C]{engelking}. Accordingly, there is a sequence~$(K_n)_{n \in \N}$ of compact subsets of $X$ with $K_n \subset K_{n+1}^{\circ}$ for every~$n \in \N$ such that any compact set~$K \subset X$ is contained in~$K_n$ for some~$n \in \N$ and~$X = \bigcup_{n=1}^{\infty} K_n$ (also see~\cite[Lemma 29.8]{bauer}). Moreover, in view of \cite[Theorem 31.5]{bauer}, the space~$X$ is Polish. These considerations motivate the following procedure.

For any second-countable, locally compact Hausdorff topological space~$X$, assume that the Lagrangian~$\L : X \times X \to \R_0^+$ is continuous, symmetric and positive on the diagonal~\eqref{strictpositive}. Moreover, assume that the measure~$\tilde{\rho}$ on $\B(X)$ is obtained similarly to the constructions in \cite[\S 4.1]{noncompact}. In order for the constructions in \cite[Section 5]{noncompact} to apply, one requires that, for any $x \in X$ and arbitrary~$\varepsilon > 0$, there exists $K_{x, \varepsilon} \subset X$ compact such that 
\begin{align*}
\int_{X \setminus K_{x, \varepsilon}} \L(x,y) \: d\tilde{\rho}(y) < \varepsilon \:.
\end{align*}
The results in \cite{williamson} imply that for any separable, locally compact metric space~$(X,d)$ there is a Heine-Borel metric~$d_{\textup{HB}}$ on~$X$ which generates the same topology.
In order to also allow for \emph{bounded} Heine-Borel metrics in \cite[Section 5]{noncompact}, it seems preferable to not specify the set $K_{x,  \varepsilon}$ in \cite[eq.~(5.2)]{noncompact} and the calculations thereafter in terms of a possibly bounded Heine-Borel metric~$d_{\textup{HB}}$. For this reason, it seems preferable to formulate the calculations in \cite[\S 5.1]{noncompact} purely in terms of compact subsets rather than in terms of (closed) balls (whose diameter depends on the corresponding Heine-Borel metric). To this end, for any topological space $Y$, we denote the set of all functions~$f : Y \to \R$ by~$\mathbb{F}(Y)$, and let $\mathbb{F}^+(Y)$ be the subset of non-negative such functions. Given a second-countable, locally compact Hausdorff space $X$, by~$\mathbb{F}_0^+(X)$ we denote the subset of non-negative functions \emph{vanishing at infinity} in the sense that, for any $\varepsilon > 0$, there exists $K \subset X$ compact with $f|_{X \setminus K} < \varepsilon$ (for continuous functions vanishing at infinity we refer to~\cite[\S VIII.2]{elstrodt}). 

This allows us to generalize the definition of Lagrangians decaying in entropy by restating condition~(c) in Definition \ref{Definition vanishing} in the following way: Given the compact exhaustion~$(K_m)_{m \in \N}$ of $X$ with~$K_m \subset K_{m+1}^{\circ}$ for every $m \in \N$ and~$X = \bigcup_{m=1}^{\infty} K_m$, 
for every~$x \in X$ let $N = N(x)$ be the least integer such that~$x \in K_m$ for all $m \ge N$. We now introduce the sets $(K_m(x))_{m \in \N}$ by
\begin{align}\label{(Km)}
K_m(x) := K_{m + N -1} \qquad \text{for all $m \in \N$} \:. 
\end{align}
Introducing \emph{entropy}~$E_x(K_m(x), \delta)$ according to~\S \ref{seccvpsigma} as the smallest number of balls of radius $\delta > 0$ covering $K_m(x)$, and replacing (c) in Definition \ref{Definition vanishing} by (c'), we define Lagrangians vanishing in entropy as follows. 
\begin{Def}\label{Definition entropy new}
	Let $(X, d)$ be a second-countable, locally compact metric space. 
	Then the Lagrangian $\L : X \times X \to \R_0^+$ is said to {\bf{vanish in entropy}}
	if the following conditions are satisfied: 
	\begin{enumerate}[leftmargin=2em]
		\item[\rm{(a)}] $c:=\inf_{{x} \in X} \L({x},{x}) > 0$.
		\item[\rm{(b)}] There is a compact set $K \subset X$ such that 
		\[ \delta := \inf_{{x} \in X \setminus K} \sup \left\{ s \in \R \::\: \L({x},y) \ge \frac{c}{2} \quad \text{for all~$y \in B_s(x)$} \right\} > 0 \:. \]
		\item[\rm{(c')}] The Lagrangian has the following decay property:
		Given an exhaustion of $X$ by compact subsets $(K_m)_{m \in \N}$, 
		there exists $f : X \times X \to \R_0^+$ with~$f(x, \cdot) \in \mathbb{F}_0^+(X)$ for every $x \in X$ such that, for every $x \in X$ and all~$m \in \N$,
		\begin{align*}
		\L(x,y) \le 2^{-m} \frac{f(x,y)}{C_x(m, \delta)} \qquad \text{for all $y \in K_m(x)$} \:,
		\end{align*}
		where $(K_m(x))_{m \in \N}$ is defined by \eqref{(Km)}, 
		$$C_x(m,\delta) := C \: E_x(K_{m+2}(x), \delta) \qquad \text{for all $x \in \F$, $m \in \N$ and $\delta > 0$} $$
		(with entropy~$E_x(K_m(x), \delta)$ as introduced in \S \ref{seccvpsigma}), 
		and the constant $C$ is given by
		$$ C := 1+ \frac{2}{c} < \infty \:. $$
	\end{enumerate}
\end{Def}\noindent
As mentioned in \cite{noncompact}, we may assume that $\delta = 1$ (otherwise we suitably rescale the corresponding metric on~$X$). 
Let us point out that Definition \ref{Definition entropy new}, by contrast to Definition \ref{Definition vanishing} (see \cite[Definition~5.1]{noncompact}), does not require a Heine-Borel metric and thus allows for more general applications. 
Definition \ref{Definition vanishing} can be considered as a special case of Definition \ref{Definition entropy new}. 

Under the assumptions (a), (b), (c'), for any $x \in X$ and~$\varepsilon > 0$ there exists $K_{x, \varepsilon} \subset X$ compact such that 
\begin{align*}
\int_{X \setminus K_{x, \varepsilon}} \L(x,y) \: d\tilde{\rho}(y) < \varepsilon \:.
\end{align*}
To see this, we make use of the fact that~$K_n \subset K_{n+1}^{\circ}$ for all $n \in \N$. Given $x \in X$ and arbitrary~$\varepsilon > 0$, there exists~$\tilde{K}_{x, \varepsilon} \subset X$ compact with $f(x,y) < \varepsilon/6$ for all~$y \notin \tilde{K}_{x, \varepsilon}$. Since $X$ is hemicompact, there exists~$n \in \N$ with $\tilde{K}_{x, \varepsilon} \subset K_n$. We denote the least such integer by~$N_0 = N_0(x, \varepsilon)$. Then the compact set (cf.~\cite[eq.~(5.2)]{noncompact})
\begin{align}\label{(5.2')}
K_{x, \varepsilon} := K_{N_0} \subset X
\end{align}
has the desired property: 
\begin{align*}
&\int_{X \setminus K_{N_0}} \L(x,y) \: d\tilde{\rho}(y) = \sum_{m = N_0}^{\infty} \int_{K_{m+1} \setminus K_m} \L(x,y) \: d\tilde{\rho}(y) \\
&\qquad \le \sum_{m = N_0}^{\infty} \sup_{y \in K_{m+1}} \L(x,y) \: \underbrace{\tilde{\rho} \left(K_{m+1} \setminus K_m \right)}_{\text{$\le C_x(m,1)$}} \le \sup_{y \in X \setminus K_{N_0}} f(x,y) \sum_{m = N_0}^{\infty} 2^{-m} < \varepsilon/3 \:.
\end{align*}
By definition of $C_x(m, \delta)$, we are given
\begin{align}\label{(5.3')}
\int_{X \setminus K_{x, \varepsilon}} \L(\tilde{x}, y) \: d\tilde{\rho}(y) < \varepsilon/3
\end{align}
for all $\tilde{x}$ in a sufficiently small neighborhood of $x$. 

Assuming that the Lagrangian is continuous, we proceed similarly to \cite{noncompact} to prove that the same is true for the measures~$\tilde{\rho}^{(n)}$ as given by \cite[eq.~(4.5)]{noncompact} (where the measures~$\tilde{\rho}^{(n)}$ originate in the same manner as in \cite[\S 4.1]{noncompact}). 
More precisely, for any given~$x \in X$ and $\varepsilon > 0$, we introduce the compact sets $A_m(x) \subset X$ by
\begin{align*}
A_m(x) := \overline{K_{m+1}(x) \setminus K_m(x)} \qquad \text{for all $m \ge N_0 = N_0(x, \varepsilon)$} \:. 
\end{align*}
Next, regularity of~$\tilde{\rho}$ yields the existence of open sets $U_m(x) \supset A_m(x)$ with $U_m(x) \subset K_{m+1}(x) \setminus K_{m-1}(x)$ such that
\begin{align*}
\tilde{\rho}\left(U_m \setminus \big(K_{m+1}(x) \setminus K_m(x) \big) \right) < 2^{-m-1} \varepsilon/3 \qquad \text{for all $m \ge N_0$}
\end{align*}
In view of \cite[Lemma 2.92]{aliprantis}, for every $m \ge N_0$ there exists $\eta_m \in C_c(U_m(x); [0,1])$ such that $\eta_m|_{A_m(x)} \equiv 1$, implying that $\L(x, \cdot) \: \eta_m \in C_c(U_m(x))$ for all $m \ge N_0$. Repeating the arguments in \cite{noncompact}, we finally arrive at \cite[eq.~(5.6)]{noncompact}, i.e.
\begin{align}\label{(5.6')}
\int_{X \setminus K_{x, \varepsilon}} \L(\tilde{x},y) \: d\tilde{\rho}^{(n)}(y) < \varepsilon \qquad \text{and} \qquad \int_{X \setminus K_{x, \varepsilon}} \L(\tilde{x},y) \: d\tilde{\rho}(y) < \varepsilon
\end{align}
for all $\tilde{x}$ in a small neighborhood of $x$ and sufficiently large $n \in \N$. As a consequence, all results in \cite[Section 5]{noncompact} remain valid for Lagrangians decreasing in entropy. 

The advantage of Definition \ref{Definition entropy new} is that it applies to arbitrary second-countable, locally compact metric spaces. In particular, by contrast to Definition \ref{Definition vanishing}, it need not be endowed with an unbounded Heine-Borel metric.
Furthermore, the concept of Lagrangians vanishing in entropy carries over to possibly non-locally compact metric spaces in the following way. 
\begin{Def}\label{Definition vanish in entropy}
	Given a metric space $(X, d)$, 
	the Lagrangian $\L : X \times X \to \R_0^+$ is said to \textbf{vanish in entropy} if, for any second-countable, locally compact Hausdorff space~$Y \subset X$, its restriction~$\L|_{Y \times Y} : Y \times Y \to \R_0^+$ vanishes in entropy with respect to the 
	induced metric~$d_Y := d|_{Y \times Y}$ (see Definition~\ref{Definition entropy new}). 
\end{Def}

\subsection{Preparatory Results}\label{S Preparatory results}
After these preliminaries we return to causal variational principles in the non-locally compact setting (see Definition \ref{Definition non-locally compact}). Accordingly, let $X$ be a separable infinite-dimensional complex Banach space and assume that~$\F \subset X$ is a non-locally compact Polish subspace (with respect to the Fréchet metric~$d$ induced by the norm on~$X$). Then $(\F, d)$ is a separable, complete metric space. In what follows we assume that the Lagrangian~$\L : \F \times \F \to \R_0^+$ vanishes in entropy (see Definition~\ref{Definition vanish in entropy} and Definition~\ref{Definition entropy new}). 
Considering the finite-dimensional exhaustion~$\F^{(n)} \subset \F$ endowed with the induced metric~$d_n := d|_{\F^{(n)} \times \F^{(n)}}$ for all $n \in \N$,
the explanations in \S \ref{S Approximating} imply that each closed bounded subset of $\F^{(n)}$ (with respect to~$d_n$) is compact. Moreover, the restricted Lagrangians~$\L^{(n)} : \F^{(n)} \times \F^{(n)} \to \R_0^+$ vanish in entropy (see Definition \ref{Definition entropy new}). As outlined in \S \ref{S Entropy}, all results in \cite[Section 5]{noncompact} remain valid if we replace ``decaying in entropy'' by ``vanishing in entropy.'' 
In particular, by applying \cite[Theorem 5.8]{noncompact} we conclude that for each $n \in \N$, there is some regular Borel measure ${\rho}^{[n]}$ on $\F^{(n)}$ which is a minimizer of the corresponding action~$\Sact^{(n)} := \Sact_{\F^{(n)}}$ under variations of compact support, where 
\begin{align*}
\Sact_{E}(\rho) := \int_E d\rho(x) \int_E d\rho(y) \: \L(x,y)
\end{align*}
for any $E \in \B(\F)$ (cf. \cite[\S 3.2]{noncompact}). 
Moreover, in view of \cite[Theorem 5.5]{noncompact}, for all $n \in \N$ the following Euler-Lagrange equations hold,  
\begin{align*}
{\ell}^{[n]}|_{\supp {\rho}^{[n]}} \equiv \inf_{x \in \F} {\ell}^{[n]}(x) = 0 \:, 
\end{align*}
where the mapping ${\ell}^{[n]} : \F \to \R $ is defined by 
\begin{align*}
{\ell}^{[n]}(x) := \int_{\F} \L(x,y) \: d{\rho}^{[n]}(y) - 1 \:.
\end{align*}
We point out that Lemma \ref{Lemma upper} is applicable to the sequence $({\rho}^{[n]})_{n \in \N}$, implying that 
\begin{align*}
{\rho}^{[n]}(K) \le C_K \qquad \text{for all $n \in \N$} \:.
\end{align*}
For this reason, we may 
proceed in analogy to Section \ref{Section Construction Global} by introducing a countable set~$\mathscr{D} \subset \mathfrak{K}(\F)$ (see \S \ref{S Construction of Countable Set}). Next, in analogy to \S \ref{S Global Measure} we iteratively restrict the sequence of measures~$({\rho}^{[n]})_{n \in \N}$ to~$D_m \subset \F$ compact with $D_m \in \mathscr{D}$ for all $m \in \N$ and denote the resulting diagonal sequence by $({\rho}^{(k)})_{k \in \N}$ (cf.~\eqref{(rhorund)}). Defining the corresponding set function~$\varphi : \mathscr{D} \to [0, \infty)$ by \eqref{(infty)} and proceeding in analogy to the proof of Theorem~\ref{Theorem measure}, we obtain a (possibly trivial) measure~$\rho$ on the Borel $\sigma$-algebra $\B(\F)$. Lemma \ref{Lemma inner regular} yields that the resulting measure~$\rho : \B(\F) \to [0, + \infty]$ is regular. Moreover, the useful results Lemma \ref{Lemma vague convergence} and Lemma~\ref{Lemma weak convergence} still apply. 

In analogy to Remark~\ref{Remark non-trivial}, the following remark yields a sufficient condition for the measure~$\rho$ obtained in Theorem~\ref{Theorem measure} to be non-zero.

\begin{Remark}\label{Remark non-trivial vanishing}
	Let $(\F^{(n)})_{n \in \N}$ be a finite-dimensional approximation of $\F$ (see \S \ref{S Approximating}). 
	Assuming that the Lagrangian is bounded and vanishes in entropy (see Definition \ref{Definition vanish in entropy}), for every $x^{(n)} \in \F^{(n)}$ and $0 < \varepsilon < 1$ there exists $K_{x, \varepsilon}^{(n)} \subset \F^{(n)}$ compact such that
	\begin{align*}
	\int_{\F \setminus K_{x,\varepsilon}^{(n)}} \L(x^{(n)}, y) \: d\rho^{(n)}(y) < \varepsilon \:.
	\end{align*}
	In view of boundedness of the Lagrangian we introduce the upper bound~$\mathscr{C} < \infty$ by
	\begin{align*}
	\mathscr{C} := \sup_{x, y \in \F} \L(x,y) > 0 \:. 
	\end{align*}
	Then the EL equations \eqref{(ELn')} and \eqref{(elln')} yield
	\begin{align*}
	1 \le \int_{\F} \L(x^{(n)},y) \: d{\rho}^{(n)} = \int_{K_{x, \varepsilon}^{(n)}} \L(x^{(n)},y) \: d{\rho}^{(n)} + \int_{\F \setminus K_{x, \varepsilon}^{(n)}} \L(x^{(n)},y) \: d{\rho}^{(n)} \:,
	\end{align*}
	implying that
	\begin{align*}
	0 < \frac{1- \varepsilon}{\mathscr{C}} \le \rho^{(n)}(K_{x, \varepsilon}^{(n)}) \qquad \text{for sufficiently large~$n \in \N$} \:. 
	\end{align*}
	Without loss of generality we may assume that $K_{x, \varepsilon}^{(n)} \in \mathscr{D}$ for every $n \in \N$. Moreover, we are given~$K_{x, \varepsilon}(N) := \bigcup_{n=1}^N K_{x, \varepsilon}^{(n)} \in \mathscr{D}$ for every $N \in \N$. Therefore, whenever there exists $N \in \N$ such that $\rho^{(n)}(K_{x, \varepsilon}(N)) \ge c$ for almost all~$n \in \N$ and some~$c > 0$, the measure~$\rho$ as defined by \eqref{(rho)} is non-zero. If this holds true for an infinite number of disjoints sets $(K_{x_i, \varepsilon}(N_i))_{i \in \N}$, the measure~$\rho$ possibly has infinite total volume. 
\end{Remark}

The remainder of this section is devoted to the proof that the measure~$\rho$ defined by~\eqref{(rho)} is, under suitable assumptions, a minimizer under variations of compact support as well as under variations of finite volume. For non-trivial minimizers, we shall derive the corresponding EL equations (see~\S\ref{S EL decay}).

\subsection{Existence of Minimizers} \label{S compact support decay}
The aim of this subsection is to prove that, under suitable assumptions, the measure $\rho$ defined by~\eqref{(rho)}
is a minimizer of the causal variational principle \eqref{(cvp)}, \eqref{totvol} under variations of finite volume (see Definition \ref{Definition finite volume}). To this end, we first show that~$\rho$ is a minimizer of the causal action under variations of compact support (see Definition~\ref{Definition compact support}).

In order to show that the measure~$\rho$ obtained in Theorem \ref{Theorem measure} is a minimizer under variations of compact support, we need to assume that condition~{\rm{(iv)}} in \S \ref{seccvpsigma} holds (cf.~\cite[\S 5.4]{noncompact}), i.e.
\begin{align*}
\sup_{x \in \F} \int_{\F}\L(x,y)\: d\rho(y) < \infty \:.
\end{align*} 
Under the additional assumption that the measure~$\rho$ obtained in Theorem~\ref{Theorem measure} also satisfies condition~{\rm{(B)}} in Section~\ref{Section minimizers bounded range}, 
we obtain the following existence result.  

\begin{Lemma}\label{Lemma compact support decay}
	Assume that the Lagrangian $\L \in C_b(\F \times \F; \R_0^+)$ vanishes in entropy,
	and that condition~\eqref{(ivn)} holds. 
	Moreover, assume that the measure~$\rho$ obtained in \eqref{(rho)} satisfies condition~{\rm{(B)}} in Section~\ref{Section minimizers bounded range}, 
	and that condition~{\rm{(iv)}} in \S \ref{seccvpsigma} holds. 
	Then~$\rho$ is a minimizer under variations of compact support. 
\end{Lemma}
\begin{proof}
	Assume that~$\tilde{\rho} : \B(\F) \to [0, \infty]$ is a regular Borel measure satisfying~\eqref{totvol} such that~$K := \supp (\tilde{\rho} - \rho)$ is a compact subset of~$\F$. Then 
	the difference of actions~\eqref{integrals} as given by 
	\begin{align*}
	\left(\Sact(\tilde{\rho}) - \Sact(\rho)\right) &= 2\int_{K} d(\tilde{\rho} - \rho)(x) \int_{\F} d\rho(y) \:\L(x,y) \\
	&\qquad \quad + \int_{K} d(\tilde{\rho}- \rho)(x) \int_{K} d(\tilde{\rho}- \rho)(y) \:\L(x,y) 
	\end{align*} 
	is well-defined (see the explanations in \cite[\S 4.3]{noncompact}). Assuming that condition~{\rm{(iv)}} in~\S \ref{seccvpsigma} holds, for any~$x \in \F$ and~$\tilde{\varepsilon} > 0$ there exists an integer~$R = R(x, \tilde{\varepsilon})$ such that 
	\begin{align}
	\int_{\F \setminus B_R(x)} \L(x,y) \: d\rho(y) < \tilde{\varepsilon}/2 \:.
	\end{align}
	By continuity of the Lagrangian, there is an open neighborhood $U_x$ of $x$ such that
	\begin{align}
	\int_{\F \setminus B_R(x)} \L(z,y) \: d\rho(y) < \tilde{\varepsilon} \qquad \text{for all $z \in U_x$} \:.
	\end{align}
	Proceeding in analogy to the proof of~\cite[Theorem 5.8]{noncompact} by covering the compact set~$K \subset \F$ by a finite number of such neighborhoods $U_{x_1}, \ldots, U_{x_L}$ and introducing 
	the bounded set~$B_K := \bigcup_{j=1}^L B_R(x_j)$, we conclude that 
	\begin{align}
	\int_{\F \setminus B_K} \L(x,y) \: d\rho(y) < \tilde{\varepsilon} \qquad \text{for all $x \in K$} \:.
	\end{align}
	This implies that, by choosing $\tilde{\varepsilon}> 0$ suitably, the last summand in the expression 
	\begin{align*}
	&\left(\Sact(\tilde{\rho}) - \Sact(\rho)\right) = \left[ 2\int_{K} d(\tilde{\rho} - \rho)(x) \int_{B_K} d\rho(y) \:\L(x,y) \right. \\
	&\qquad \left. + \int_{K} d(\tilde{\rho}- \rho)(x) \int_{K} d(\tilde{\rho}- \rho)(y) \:\L(x,y) \right] + 2\int_{K} d(\tilde{\rho} - \rho)(x) \int_{\F \setminus B_K} d\rho(y) \:\L(x,y)
	\end{align*}
	is arbitrarily small. For this reason, it remains to consider the term in square brackets in more detail. 
	Combining the facts that~$\rho$ satisfies condition~{\rm{(B)}} in Section~\ref{Section minimizers bounded range} and that~$B_K \subset \F$ is bounded, Lemma \ref{Lemma B locally finite} implies that $\rho(B_K) < \infty$. Proceeding similarly to the proof of Lemma \ref{Lemma compact support}, we deduce that the term in square bracket is bigger than or equal to zero, up to an arbitrarily small error term. This gives rise to
	\begin{align*}
	\left(\Sact(\tilde{\rho}) - \Sact(\rho)\right) \ge 0 \:,
	\end{align*} 
	which proves the claim. 
\end{proof}

Proceeding similarly to the proof of Lemma~\ref{Lemma compact support decay}, we obtain the following result. 

\begin{Thm}\label{Theorem finite volume decay}
	Assume that the Lagrangian $\L \in C_b(\F \times \F; \R_0^+)$ vanishes in entropy,
	and that condition~\eqref{(ivn)} holds. 
	Moreover, assume that the measure~$\rho$ obtained in~\eqref{(rho)} satisfies condition~{\rm{(B)}} in Section~\ref{Section minimizers bounded range}, 
and that condition~{\rm{(iv)}} in \S \ref{seccvpsigma} holds. 
	Then~$\rho$ is a minimizer under variations of finite volume. 
\end{Thm}
\begin{proof}
	Assume that~$\tilde{\rho} : \B(\F) \to [0, \infty]$ is a regular measure 
	satisfying~\eqref{totvol}. 
	By virtue of Lemma~\ref{Lemma B locally finite}, the measure~$\rho$ is locally finite, implying that~$\tilde{\rho}$ is also a locally finite measure (see the explanations after Definition \ref{Definition finite volume}). 
	Introducing~$B := \supp (\tilde{\rho} - \rho)$, we are given~$\rho(B) = \tilde{\rho}(B) < \infty$. Since condition~(iv) in~\S \ref{seccvpsigma} holds, 
	the difference of actions~\eqref{integrals} is well-defined and given by 
	\begin{align*}
	\left(\Sact(\tilde{\rho}) - \Sact(\rho)\right) &= 2\int_{B} d(\tilde{\rho} - \rho)(x) \int_{\F} d\rho(y) \:\L(x,y) \\
	&\qquad \quad + \int_{B} d(\tilde{\rho}- \rho)(x) \int_{B} d(\tilde{\rho}- \rho)(y) \:\L(x,y) \:.
	\end{align*} 
	Making use of regularity of $\rho$ and $\tilde{\rho}$, for arbitrary $\tilde{\varepsilon} > 0$ there is~$U \supset B$ open such that
	\begin{align*}
	\rho(U \setminus B) < \tilde{\varepsilon} \qquad \text{and} \qquad \tilde{\rho}(U \setminus B) < \tilde{\varepsilon} \:.
	\end{align*}
	Approximating $U$ from inside by compact sets $K$ such that
	\begin{align*}
	\rho(U \setminus K) < \tilde{\varepsilon} \qquad \text{and} \qquad \tilde{\rho}(U \setminus K) < \tilde{\varepsilon}
	\end{align*}
	and proceeding in analogy to the proof of \cite[Theorem 5.9]{noncompact} and Lemma \ref{Lemma compact support decay}, we finally may deduce that 
	\begin{align*}
	\left(\Sact(\tilde{\rho}) - \Sact(\rho)\right) \ge 0 \:,
	\end{align*} 
	which proves the claim. 
\end{proof}

Theorem \ref{Theorem finite volume decay} concludes the existence theory in the non-locally compact setting. 

\subsection{Derivation of the Euler-Lagrange Equations}\label{S EL decay}
Under the assumptions of Theorem \ref{Theorem finite volume decay}, for non-trivial measures~$\rho \not= 0$
we are able to deduce the corresponding Euler-Lagrange equations. More precisely, in analogy to~\cite[Lemma 2.3]{jet} we obtain the following result.

\begin{Thm}\label{Theorem EL decay}
	Assume that the Lagrangian $\L \in C_b(\F \times \F; \R_0^+)$ vanishes in entropy (see Definition \ref{Definition vanish in entropy}), and that condition~\eqref{(ivn)} holds. Moreover, assume that the regular measure~$\rho$ obtained in \eqref{(rho)} is non-zero and satisfies condition~{\rm{(B)}} in Section~\ref{Section minimizers bounded range} 
	as well as condition~{\rm{(iv)}} in \S \ref{seccvpsigma}.  
	Then the following Euler-Lagrange equations hold,
	\begin{align}\label{(EL decay)}
	\ell|_{\supp \rho} \equiv \inf_{x \in \F} \ell(x) = 0 \:,
	\end{align}
	where $\ell \in C(\F)$ is defined by \eqref{(ell)} for a suitable parameter~$\mathfrak{s} \in \R_0^+$. 
\end{Thm}
\begin{proof}
	Under the assumptions of Theorem \ref{Theorem EL decay}, the measure $\rho$ constructed in \eqref{(rho)} is locally finite in view of Lemma \ref{Lemma iv Borel} and a minimizer under variations of finite volume. Assuming that $\rho \not= 0$ and arguing similarly to the proof of Lemma \ref{Lemma EL0}, Theorem \ref{Theorem EL} gives rise to~\eqref{(EL decay)}. 
\end{proof}

Theorem \ref{Theorem EL decay} generalizes the results of~\cite[Section 5]{noncompact} 
to the infinite-dimensional setting. It remains an open task to prove the existence of Lagrangians vanishing in entropy such that the measure~$\rho$ given by~\eqref{(rho)} is non-zero
and satisfies condition~(iv) in~\S \ref{seccvpsigma} as well as condition~{\rm{(B)}} in Section~\ref{Section minimizers bounded range}.

\section{Topological Properties of Spacetime}\label{Section topological properties of support}
The goal of this section is to derive topological properties of spacetime and to work out a connection to dimension theory. To this end, we let $\F$ be a non-locally compact Polish space 
in the non-locally compact setting (see Definition \ref{Definition non-locally compact}). Under suitable assumptions on the Lagrangian (see Theorem \ref{Theorem finite volume} and Theorem \ref{Theorem finite volume decay}), the measure~$\rho$ obtained in~\eqref{(rho)} 
is a minimizer of the corresponding variational principle~\eqref{(cvp)}. In order to obtain dimension-theoretical statements on its support, let us first recall some basic results from dimension theory (\S \ref{S dimension support}). Afterwards we specialize the setting by applying the obtained results to causal fermion systems (\S \ref{S Application}).

\subsection{Dimension-Theoretical Preliminaries}\label{S dimension support}
To begin with, let us first point out that there are several notions of ``dimension'' of a topological space, among them the \emph{small inductive dimension}~$\operatorname{ind}$, the \emph{large inductive dimension}~$\operatorname{Ind}$, the \emph{covering dimension}~$\dim$, the \emph{Hausdorff dimension}~$\dim_H$ and the \emph{metric dimension}~$\mu\dim$ (for details we refer to~\cite{fedorchuk}, \cite{engelkingdim}, \cite{hurewicz+wallman} and~\cite{munkres}). 
For a separable metric space~$X$, the relation 
\begin{align}\label{(dim le)}
\dim X \le \dim_H X
\end{align}
holds in view of \cite[Section VII.4]{hurewicz+wallman}. Moreover, for every separable metrizable space $X$ we have~$\operatorname{ind} X = \operatorname{Ind} X = \dim X$ (see \cite[Theorem~4.1.5]{engelkingdim}), and $\mu\dim Y = \dim Y$ for every compact metric space~$Y$ (see e.g.~\cite{fedorchuk}). 

For a metric space $X$, the \emph{local dimension} $\dim_{\textup{loc}} : X \to [0, \infty]$ is given by 
\[\dim_{\textup{loc}}(x) = \inf \left\{\dim_H(B_{\varepsilon}(x)) : \varepsilon > 0 \right\} \: \]
(see \cite[\S 2]{dever}), where
\[\dim_H(A) = \inf\left\{s \ge 0 : H^s(A) = 0 \right\} \]
is the Hausdorff dimension of $A \subset X$, and $H^s$ is the $s$-dimensional Hausdorff measure (the interested reader is referred to~\cite[Section~2.10]{federer}, \cite[Chapter~VII]{hurewicz+wallman} and~\cite{rogers}). Whenever $X$ is a separable metric space, then one can show that
\[\dim_H (X) = \sup_{x \in X} \dim_{\textup{loc}}(x) \:. \] 
Moreover, if $X$ is compact then the supremum is attained (cf. \cite[Proposition 2.7]{dever}). 

\begin{Def}\label{Definition locally finite dimensional}
	A normal space $X$ is \emph{locally finite-dimensional} if for every~$x \in X$ there exists a normal open subspace $U$ of $X$ such that $x \in U$ and $\dim U < \infty$. See~\cite[Section 5.5]{engelkingdim}.
\end{Def}

In order to apply the above preliminaries to minimizers of the causal variational principle, let us summarize some general topological properties of the support of a locally finite measure~$\mu$ on a Polish space~$\F$ in the next statement.  

\begin{Lemma}\label{Lemma dimension support}
	Let $X$ be a Polish space, and assume that $\mu$ is a locally finite measure on $\B(X)$. Then $\supp \mu \subset X$ is $\sigma$-compact, and there exists a locally finite-dimensional subspace $F$ being dense in~$\supp \mu$. Whenever $\supp \mu$ is hemicompact, then $\supp \mu$ is locally compact and thus locally finite-dimensional.
	Moreover, in the latter case there exists a (Heine-Borel) metric on $\supp \mu$ such that each bounded subset in $\supp \mu$ is finite-dimensional. 
\end{Lemma}
\begin{proof}
	According to Lemma \ref{Lemma support LCH},~$\supp \mu$ is a $\sigma$-compact separable metric space, and there is a dense subset $F \subset \supp \mu$ such that each $x \in F$ is contained in a compact neighborhood~$N_x$. In view of~\cite[Proposition 2.7]{dever} we conclude that~$\dim_H N_x < \infty$ for every~$x \in \supp \mu$. Thus Definition~\ref{Definition locally finite dimensional} together with \eqref{(dim le)} gives the first statement. 
	
	Whenever $\supp \mu$ is hemicompact, it is locally compact according to Lemma \ref{Lemma support LCH}. Thus each $x \in \supp \mu$ is contained in a compact neighborhood $N_x$. Making use of~\cite[Proposition 2.7]{dever} we deduce that~$\dim_H N_x < \infty$. Since $x \in \supp \mu$ is arbitrary and the interior of~$N_x$ is open, from Definition \ref{Definition locally finite dimensional} we obtain that $\supp \mu$ is locally finite-dimensional. Moreover, due to~Lemma \ref{Lemma support LCH}, the space~$\supp \mu$ can be endowed with a Heine-Borel metric. Accordingly, whenever $B \subset \supp \mu$ is bounded (with respect to the Heine-Borel metric), its closure is compact. Covering the resulting compact set by a finite number of compact neighborhoods, we conclude that~$\dim_H (B) < \infty$. 
\end{proof}

\subsection{Application to Causal Fermion Systems}\label{S Application}
In the remainder of this section, we finally apply the previous results to the case of causal fermion systems. To this end, let~$\H$ be 
an infinite-dimensional separable complex Hilbert space. 
For a given spin dimension~$n \in \N$, the set~$\F \subset \LL(\H)$ (for details see~\cite[Definition~1.1.1]{cfs}) is a non-locally compact Polish space (see Theorem \ref{Theorem Polish} and Lemma~\ref{Lemma Freg}). 
Assuming that the Lagrangian~$\L : \F \times \F \to \R_0^+$ is symmetric, lower semi-continuous and strictly positive on the diagonal~\eqref{strictpositive}, we are exactly in the non-locally compact setting as introduced in~\S \ref{S Basic Definitions} (see Definition \ref{Definition non-locally compact}). Therefore, by virtue of Theorem \ref{Theorem measure}, there exists a regular measure~$\rho : \B(\F) \to [0, \infty]$. Assuming in addition that the Lagrangian is continuous, bounded and of bounded range  
and that condition~{\rm{(B)}} in Section~\ref{Section minimizers bounded range} is satisfied, 
the measure~$\rho$ is a minimizer of the causal variational principle~\eqref{(cvp)}, \eqref{totvol} under variations of compact support by virtue of Lemma~\ref{Lemma compact support}. 
Under the additional assumption that condition~{\rm{(iv)}} in~\S \ref{seccvpsigma} is satisfied, the measure $\rho$ is a minimizer of the causal variational principle under variations of finite volume due to Theorem~\ref{Theorem finite volume}. Under these assumptions, the same is true for Lagrangians vanishing in entropy (see Theorem~\ref{Lemma compact support decay} and Theorem~\ref{Theorem finite volume decay}).
As a consequence, we are given a causal fermion system~$(\H, \F, \rho)$, and 
\emph{spacetime}~$M$ is defined as the support of the universal measure~$\rho$, 
\[M : = \supp \rho \:. \]
Combining the results of Lemma \ref{Lemma dimension support} and Lemma \ref{Lemma support LCH},
we arrive at the following main results of this section. 
\begin{Thm}\label{Theorem spacetime}
	Assume that $\L \in C_b(\F \times \F; \R_0^+)$ is of bounded range or vanishes in entropy. Moreover, assume that the measure~${\rho}$ defined by \eqref{(rho)} satisfies condition~{\rm{(B)}} in Section~\ref{Section minimizers bounded range}.
	Then spacetime $M$ is $\sigma$-compact and contains a locally finite-dimensional dense subspace. Under the additional assumption that spacetime~$M$ is hemicompact, it is a locally finite-dimensional, $\sigma$-locally compact Polish space. 
\end{Thm}
\begin{proof}
	Assuming that condition~{\rm{(B)}} in Section~\ref{Section minimizers bounded range} is satisfied, the measure~$\rho$ is locally finite in view of Lemma~\ref{Lemma B locally finite}. Henceforth the statement is a consequence of Lemma \ref{Lemma support LCH} and Lemma \ref{Lemma dimension support}.  
\end{proof}

In~\cite{bernard+finster} the question is raised whether the support of minimizing measures always is compact. Theorem~\ref{Theorem spacetime} indicates that the support should in general at least be~$\sigma$-compact. 

Under the assumption that the measure~$\rho$ is locally finite (for sufficient conditions see Lemma \ref{Lemma locally finite}, Lemma \ref{Lemma iv Borel} and Lemma~\ref{Lemma B locally finite}), we obtain the following result. 

\begin{Thm}\label{Theorem empty}
	Assume that the measure $\rho : \B(\F) \to [0, \infty]$ given by \eqref{(rho)} is locally finite. Then the interior of spacetime~$M = \supp \rho$ is empty (in the topology of $\F$).  
\end{Thm}
\begin{proof}
	Assume that $M^{\circ} \not= \varnothing$ in the topology of $\F$. 
	Then~$U^{\textup{reg}} := M^{\circ} \cap \F^{\textup{reg}}$ is open in the relative topology. Since $\F^{\textup{reg}}$ is a Banach manifold (see~\cite{finster+lottner}), it can be covered by an atlas $(U_{\alpha}, \phi_{\alpha})_{\alpha \in A}$ for some index set~$A$ (cf.~\cite[Chapter 73]{zeidlerIV}). In particular, each~$x \in U^{\textup{reg}}$ is contained in some open set $U_{\alpha}$, whose image~$V_{\alpha} := \phi_{\alpha}(U_{\alpha})$ is open in some infinite-dimensional Banach space $X_{\alpha}$. 
	From Lemma \ref{Lemma support LCH} we know that there exists a dense subset $F \subset \supp \rho$ such that each $x \in F$ has a compact neighborhood. Given $x \in F$ and choosing a compact neighborhood $N_x \subset U_{\alpha}$ for some $\alpha \in A$, from the fact that the mapping~$\phi_{\alpha}$ is a homeomorphism we conclude that $\phi_{\alpha}(N_x) \subset X_{\alpha}$ is a compact neighborhood of $\phi_{\alpha}(x) \in X_{\alpha}$ which contains a non-empty open subset in contradiction to \cite[Exercise~14.3]{koenig}. This gives the claim. 
\end{proof}

Theorem \ref{Theorem empty} generalizes \cite[Theorem 3.16]{support} to the infinite-dimensional setting.

\appendix

\section{Topological Properties of Causal Fermion Systems}\label{Appendix Polish}
The goal of this appendix is to prove the following result:
\begin{Thm}\label{Theorem Polish}
	Let $(\H, \F, \rho)$ be a causal fermion system. Then $\F$ is a Polish space.\footnote{More precisely, endowed with the Fréchet metric~$d$ induced by the operator norm on $\LL(\H)$, the space~$(\F, d)$ is a separable, complete metric space.}
\end{Thm}
Throughout this section we assume that~$(\H, \F, \rho)$ is a causal fermion system of spin dimension $s \in \N$. More precisely, we consider a (possibly infinite-dimensional) separable complex Hilbert space~$\H$ endowed with a scalar product~$\langle . \mid . \rangle_{\H}$. Denoting the set of all bounded linear operators on~$\H$ by $\LL(\H)$, we let~$\F \subset \LL(\H)$ be the subset consisting of those operators $A \in \LL(\H)$ which are self-adjoint with respect to the scalar product~$\langle . \mid . \rangle_{\H}$ on~$\H$ and have at most~$s$ positive and at most $s$ negative eigenvalues (see~\cite[\S 1.1.1]{cfs}). The proof of Theorem~\ref{Theorem Polish} is split up in two parts: We first point out that~$\F$ is separable (\S \ref{S Separability}). Afterwards, we prove that~$\F$ is completely metrizable (\S \ref{S Completeness}). The result can be immediately generalized to the case of operators which have at most~$p$ positive and at most~$q$ negative eigenvalues.

\subsection{Separability}\label{S Separability}
In order to prove separability of $\F$, we employ the following argument: Given an infinite-dimensional, separable complex Hilbert space~$H$, the set of linear operators on $H$, denoted by $\LL(H)$, is a non-separable Banach space (see \cite[\S3.A and \S 12.E]{kechris}). Since each compact linear operator on $H$ is bounded, from \cite[\S 12.E]{kechris} we infer that the class of compact operators 
$\K(H) \subset \LL(H)$ is separable (as well as a Banach space according to~\cite[Satz II.3.2]{werner}). Applying the previous results to a causal fermion system~$(\H, \F, \rho)$ shows that $\K(\H)$ is separable. Since each $A \in \F$ has finite rank (see e.g.~\cite[Chapter~15]{meise+vogt}), from \cite[Section II.3]{werner} we conclude that $A$ is compact, implying that~$\F \subset \K(\H)$. Since $\LL(\H)$ is metrizable by the Fréchet metric induced by the operator norm on $\LL(\H)$, the set $\K(\H)$ is metrizable, and hence $\F$ is separable in view of~\cite[Corollary 3.5]{aliprantis}. 

\subsection{Completeness}\label{S Completeness}
The aim of this subsection is to show that~$\F$ is completely metrizable with respect to the Fréchet metric induced by the operator norm on $\LL(\H)$. To this end, we proceed as follows. Given a sequence of operators $(A_n)_{n \in \N}$ in $\F$, our task is to prove that its limit~$A \in \K(\H)$ is self-adjoint (with respect to the scalar product~$\langle . \mid . \rangle_{\H}$ on~$\H$) and has at most $n$ positive and at most $n$ negative eigenvalues. 

\subsubsection{Self-Adjointness} 
We start by proving that $A$ is self-adjoint in the case of a general Hilbert space $H$. 

\begin{Lemma}\label{Lemma Self-Adjointness general}
	Let $(H, \langle . \mid . \rangle_H)$ be a Hilbert space, and let $(A_n)_{n \in \N}$ be a sequence of self-adjoint operators in $\LL(H)$ converging in norm to some $A \in \LL(H)$. Then $A$ is self-adjoint. 
\end{Lemma}
\begin{proof}
	For any $u, v \in H$, applying the Cauchy-Schwarz inequality and making use of the fact that $A_n$ is self-adjoint for every~$n \in \N$ yields
	\begin{align*}
	\left|\langle u \mid  A^* \: v \rangle_{H} - \langle u \mid A \: v \rangle_{H} \right| &= \left|\langle A\: u \mid v \rangle_{H} - \langle u \mid A \: v \rangle_{H} \right| \\
	&= \left|\langle A \: u \mid v \rangle_{H}-\langle A_n \: u \mid v \rangle_{H} + \langle A_n \: u \mid v \rangle_{H} - \langle u \mid A \:v \rangle_{H} \right|\\ 
	&\le 2 \: \|A-A_n\|_{\textup{L}(H)} \,\|u\| \,\|v\| \to_{n \to \infty} 0 \:.
	\end{align*}
	This completes the proof.
\end{proof}

\subsubsection{Operators, Resolvents and Spectra}
The remainder of this section is dedicated to the proof that the limit $A \in \LL(\H)$ of a sequence $(A_n)_{n \in \N}$ in $\F$ (with respect to the operator norm) has at most $s$ positive and at most $s$ negative eigenvalues. To this end, we will essentially make use of results in \cite{kato}, which we now briefly recall. 

For Banach spaces $X$ and~$Y$, by~$\mathscr{B}(X,Y)$ and~$\mathscr{C}(X,Y)$ we denote the set of all bounded and closed operators from $X$ to $Y$, respectively. Then~$\mathscr{B}(X,Y)$ is a Banach space, and we let~$\mathscr{B}(X) := \mathscr{B}(X,X)$ and~$\mathscr{C}(X) := \mathscr{C}(X,X)$. We denote the domain of an operator~$T$ from~$X$ to~$Y$ by~$\mathsf{D}(T)$, and its graph~$\mathsf{G}(T)$ is by definition the subset of~$X \times Y$ consisting of all elements of the form $(u, Tu)$ with~$u \in \mathsf{D}(T)$. Note that~$\mathsf{G}(T)$ is a closed linear subspace of $X \times Y$ if and only if~$T \in \mathscr{C}(X, Y)$ (see \cite[III-\S 5.2]{kato}). 

In what follows, let $X$, $Y$ be complex Banach spaces, and let $H$ be a complex Hilbert space.  
For $\zeta \in \C$ and~$T \in \mathscr{C}(X)$, we introduce the operator 
\begin{align*}
T_{\zeta} := T- \zeta \Id \:. 
\end{align*}
Then the \emph{resolvent set} $\rho(T)$ is defined to consist of all $\xi \in \C$ for which~$T_{\zeta}$ has an inverse, denoted by
\begin{align*}
R(\zeta) = R(\zeta, T) := (T- \zeta)^{-1} \:. 
\end{align*}
We call $R(\zeta,T)$ the \emph{resolvent} of $T$ (see~\cite[III-\S 6]{kato} and~\cite[Definition 8.38]{renardy+rogers}). 
The \emph{spectrum} $\sigma(T)$ of $T$ is given by the complementary set of the resolvent set in the complex plane, $\sigma(T) := \C \setminus \rho(T)$. 
Note that the spectrum of a compact operator~$T$ in a Banach space $X$ has a simple structure analogous to that of an operator in a finite-dimensional space. Namely, for compact operators, each non-zero eigenvalue is of finite multiplicity:
\begin{Thm}\label{Theorem III-6.26}
	Let $T \in \mathscr{B}(X)$ be compact. Then $\sigma(T)$ is a countable set with no accumulation point different from zero, and each nonzero $\lambda \in \sigma(T)$ is an eigenvalue of~$T$ with finite multiplicity. 
\end{Thm}
\begin{proof}
	See \cite[Theorem III-6.26]{kato}.
\end{proof}
Moreover, the spectrum~$\sigma(T)$ of a selfadjoint operator $T$ in~$H$ is a subset of the real axis. 

An isolated point of the spectrum is referred to as \emph{isolated eigenvalue} \cite[III-\S6.5]{kato}. 
Concerning compact operators $T$ on $X$, we may state the following remark. 
\begin{Remark}
	\label{Remark III-6.27}
	Let $T$ be a compact operator. 
	Then every complex number $\lambda \not= 0$ belongs to $\rho(T)$ or is an isolated eigenvalue with finite multiplicity. 
	See \cite[Remark III-6.27]{kato}. 
\end{Remark}

\subsubsection{Projection and Decomposition} 
Let $X, Y$ be Banach spaces, and let $M \subset X$ be a linear subspace (or ``manifold'' in the terminology of~\cite{kato}).  
As usual, an idempotent operator~$P \in \mathscr{B}(X)$ ($P^2 = P$) is called a \emph{projection}, giving rise to the decomposition
\begin{align}\label{(3.14)}
X = M \oplus N \:, 
\end{align}
where $M = PX$ and $N = (1- P) X$ 
are \emph{closed} linear subspaces of~$X$ which are referred to as \emph{complementary} (see \cite[III-\S3.4]{kato}). 
Then each $x \in X$ can be uniquely expressed in the form $u = u' + u''$ with $u' \in M$ and~$u'' \in N$. The vector~$u'$ is called the \emph{projection of $u$ on $M$ along $N$}, and~$P$ is called the \emph{projection operator} (or simply the \emph{projection}) \emph{on $M$ along $N$}. Accordingly, the operator~$1- P$ is the projection on $N$ along $M$. The range of $P$ is $M$ and the null space of $P$ is~$N$. For convenience we often write $\dim P$ for $\dim M = \dim \mathcal{R}(P)$, where $\mathcal{R}(P)$ denotes the range of~$P$. Since $Pu \in M$ for every~$u \in X$, we have $PPu = Pu$, implying that~$P$ is \emph{idempotent}: $P^2 = P$.     

Next, a linear subspace $M$ is said to be \emph{invariant} under an operator $T \in \mathscr{B}(X)$ if~$TM \subset M$. In this case, $T$ \emph{induces} a linear operator $T_M$ on $M$ to $M$, defined by~$T_Mu = Tu$ for $u \in M$. the operator~$T_M$ is called the \emph{part of $T$ in $M$}.
If there are two invariant linear subspaces $M, N$ for $T$ such that $X = M \oplus N$, the operator $T$ is said to be \emph{decomposed} (or \emph{reduced}) by the pair $M, N$. 

The notion of the \emph{decomposition} of $T$ by a pair $M, N$ of complementary subspaces (see \eqref{(3.14)}, \cite[III-(3.14)]{kato}) can be extended in the following way. An operator~$T$ is said to be decomposed according to $X = M \oplus N$ if
\begin{align}\label{(5.23)}
PD(T) \subset D(T) \:, \qquad TM \subset M \:, \qquad TN \subset N \:,
\end{align}
where $P$ is the projection on $M$ along $N$. 
When $T$ is decomposed as above, the \emph{parts}~$T_M$, $T_N$ of $T$ in $M, N$, respectively, can be defined. Then $T_M$ is an operator in the Banach space $M$ with $D(T_M) = D(T) \cap M$ such that $T_M u = Tu \in M$, and $T_N$ is defined similarly.

\subsubsection{Generalized Convergence}
Let us briefly recall the definition of convergence in the generalized sense: 

\begin{Def}
	Let $T, T_n \in \mathscr{B}(X,Y)$ for all~$n \in \N$.  
	\begin{itemize}[leftmargin=2em]
		\item[\rm{(i)}] The convergence of $(T_n)_{n \in \N}$ to $T$ in the sense of $\|T_n - T\| \to 0$ is called \emph{uniform convergence} or \emph{convergence in norm}.
		\item[\rm{(ii)}] Given closed operators~$T,S \in \mathscr{C}(X,Y)$,
		their graphs $\mathsf{G}(T)$, $\mathsf{G}(S)$ are closed linear subspaces of the product space~$X \times Y$. For two closed linear subspaces~$M, N$ of a Banach space $Z$ we let $\delta(M,N)$ be the smallest number $\delta$ such that
		\begin{align*}
		\dist (u, N) \le \delta\, \|u\| \qquad \text{for all $u \in M$.}
		\end{align*}
		We call $\hat{\delta}(T,S) := \delta(\mathsf{G}(T), \mathsf{G}(S))$ the \emph{gap} between $T$ and $S$. 
		If $\hat{\delta}(T_n, T) \to 0$, we shall also say that the operator $T_n$ converges to $T$ (or $T_n \to T$) \emph{in the generalized sense}. 
	\end{itemize}
	See \cite[Chapter III, \S 3.1]{kato} and \cite[Chapter IV, \S 2.1]{kato}.
\end{Def}

The following theorem establishes a connection between convergence in the generalized sense and uniform convergence. 

\begin{Thm}\label{Theorem IV-2.23}
	Let ${T, T_n \in \mathscr{C}(X,Y)}$ for all~$n \in \N$. If ${T \in \mathscr{B}(X,Y)}$, then~${T_n \to T}$ in the generalized sense iff ${T_n \in \mathscr{B}(X,Y)}$ for sufficiently large $n$ and ${\|T_n - T\| \to 0}$.
\end{Thm}
\begin{proof}
	See \cite[Theorem IV-2.23]{kato}. 
\end{proof}

\subsubsection{Separation of the Spectrum}\label{III-S 6.4}
Sometimes it happens that the spectrum $\sigma(T)$ of a closed operator $T$ contains a bounded part $\sigma'$ separated from the rest $\sigma''$ in such a way that a rectifiable, simple closed curve $\Gamma$ (or, more generally, a finite number of such curves) can be drawn so as to enclose an open set containing $\sigma'$ in its interior and~$\sigma''$ in its exterior. Under such a circumstance, we have the following \emph{decomposition theorem}.

\begin{Thm}\label{Theorem III-6.17}
	Let $\sigma(T)$ be separated into two parts $\sigma'$, $\sigma''$ in the way described above. Then we have a decomposition of $T$ according to a decomposition $X = M' \oplus M''$ of the space (in the sense of \eqref{(5.23)}, cf.\ \cite[III-\S 5.6]{kato}) in such a way that the spectra of the parts $T_{M'}$, $T_{M''}$ coincide with $\sigma'$, $\sigma''$ respectively and $T_{M'} \in \mathscr{B}(M')$. 
\end{Thm}
\begin{proof}
	See \cite[Theorem III-6.17]{kato}. 
\end{proof}
The proof of Theorem \ref{Theorem III-6.17} makes use of the so-called \emph{eigenprojection} 
\begin{align}\label{(6.19)}
P[T] = - \frac{1}{2\pi i} \int_{\Gamma} R(\zeta, T) \: d\zeta \in \mathscr{B}(X) \:, 
\end{align}
which is a projection on $M' = P[T]X$ along $M'' = (1-P[T])X$. 

\begin{Remark}[\textbf{Finite system of eigenvalues}]\label{Remark finite system}
	Suppose that the spectrum $\sigma(T)$ of~$T \in \mathscr{C}(X)$ has an \emph{isolated point} $\lambda$. Obviously $\sigma(T)$ is divided into two separate parts~$\sigma'$, $\sigma''$ 
	where $\sigma'$ consists of the single point $\lambda$; any closed curve enclosing $\lambda$ but no other point of $\sigma(T)$ may be chosen as $\Gamma$. Then the spectrum of the operator~$T_{M'}$ described in~\cite[Theorem III-6.17]{kato} (see Theorem \ref{Theorem III-6.17}) consists of the single point $\lambda$.
	
	If $M'$ is finite-dimensional, $\lambda$ is an \emph{eigenvalue} of $T$. In fact, since $\lambda$ belongs to the spectrum of the finite-dimensional operator $T_{M'}$, it must be an eigenvalue of $T_{M'}$ and hence of $T$. In this case, $\dim M'$ is called the \emph{(algebraic) multiplicity} of the eigenvalue~$\lambda$ of $T$.
	
	For brevity, a finite collection $\lambda_1, \ldots, \lambda_s$ of eigenvalues with finite multiplicities will be called a \emph{finite system of eigenvalues}. 
\end{Remark}

\subsubsection{Continuity of the Spectrum}
This paragraph is devoted that proof that the spectrum of a sequence of operators converging in the generalized sense behaves continuously. 

\begin{Thm}\label{Theorem IV-3.16}
	Let $T \in \mathscr{C}(X)$ and let $\sigma(T)$ be separated into two parts $\sigma'(T)$, $\sigma''(T)$ by a closed curve~$\Gamma$ as in \S \ref{III-S 6.4} (cf.\ \cite[III-\S 6.4]{kato}). Let $X = M'(T) \oplus M''(T)$ be the associated decomposition of $X$. Then there exists~$\delta >0$, depending on $T$ and~$\Gamma$, with the following properties. Any $S \in \mathscr{C}(X)$ with $\hat{\delta}(S, T) < \delta$ has spectrum $\sigma(S)$ likewise separated by~$\Gamma$ into two parts $\sigma'(S)$, $\sigma''(S)$ ($\Gamma$ itself running in $\rho(S)$). In the associated decomposition $X = M'(S) \oplus M''(S)$, $M'(S)$ and $M''(S)$ are isomorphic with $M'(T)$ and $M''(T)$, respectively. In particular,~$\dim M'(S) = \dim M'(T)$, $\dim M''(S) = \dim M''(T)$ and both $\sigma'(S)$ and $\sigma''(S)$ are nonempty if this is true for~$T$. The decomposition $X = M'(S) \oplus M''(S)$ is continuous in $S$ in the sense that the projection $P[S]$ of $X$ onto $M'(S)$ along $M''(S)$ tends to $P[T]$ in norm as $\hat{\delta}(S, T) \to 0$.
\end{Thm}
\begin{proof}
	See \cite[Theorem IV-3.16]{kato}.
\end{proof}

\begin{Lemma}
	\label{Lemma Finite System of Eigenvalues}
	Let $(T_n)_{n \in \N}$ be a sequence of compact operators in~$\K(X)$, and suppose that $T_n \to T$ in norm for some operator~$T \in \K(X)$. Moreover, let $\sigma'(T)$ be a finite system of $m$ eigenvalues, separated from the rest $\sigma''(T)$ of $\sigma(T)$ by a closed curve $\Gamma$ in the manner of \cite[III-\S 6.4]{kato}. Then~$\sigma(T_n)$ is separated by $\Gamma$ into $\sigma'(T_n)$, $\sigma''(T_n)$ such that each $\sigma'(T_n)$ also consists of $m$ eigenvalues of~$T_n$, provided that~$n$ is sufficiently large.
\end{Lemma}
\begin{proof}
	Since $T$ is a compact operator, Remark~\ref{Remark III-6.27} states that each non-zero eigenvalue of $T$ is isolated (see \cite[Remark III-6.27]{kato}). 
	For this reason, each non-zero eigenvalue of $T$ as well as all positive or all non-zero eigenvalues of $T$ can be enclosed by a closed curve $\Gamma$ running in $\rho(T)$ as described in \S \ref{III-S 6.4} (cf.\ \cite[III-\S 6.4]{kato}). We denote the set of eigenvalues lying within $\Gamma$ by $\sigma'(T)$, and the set of eigenvalues without~$\Gamma$ by $\sigma''(T)$; thus $\Gamma$ encloses a finite system of eigenvalues. By virtue of Theorem~\ref{Theorem III-6.17}, we have decomposition of $T$ according to $X = M' \oplus M''$ such that the spectra of the parts~$T_{M'}$, $T_{M''}$ coincide with $\sigma'(T)$, $\sigma''(T)$. Since $T_n, T \in \mathscr{B}(X)$ for all~$n \in \N$, Theorem \ref{Theorem IV-2.23} ensures that~$T_n \to T$ in the generalized sense, i.e.~$\hat{\delta}(T_n, T) \to 0$ in the limit~$n \to \infty$. Hence we can apply Theorem~\ref{Theorem IV-3.16}, giving rise to some~$N \in \N$ such that~$\Gamma$ separates the spectra $\sigma(T_n)$ into two parts $\sigma'(T_n)$, $\sigma''(T_n)$ for all $n \ge N$. Considering the corresponding decomposition $X = M'(T_n) \oplus M''(T_n)$ for all~$n \ge N$, by virtue of Theorem~\ref{Theorem IV-3.16} there exist isomorphisms~$M'(T_n) \simeq M'(T)$ for all~$n \ge N$. In particular, we obtain~$\dim M'(T_n) = \dim M'(T)$ for all~$n \ge N$, and the decomposition~$X = M'(T_n) \oplus M''(T_n)$ is continuous in $n$ in the sense that the projection $P[T_n]$ (see~\eqref{(6.19)}) of $X$ onto $M'(T_n)$ along $M''(T_n)$ tends to $P[T]$ in norm as $\hat{\delta}(T_n, T) \to 0$ which holds true in view of Theorem~\ref{Theorem IV-2.23}. A fortiori, the algebraic multiplicity of eigenvalues of~$T$ within $\Gamma$ coincides with the algebraic multiplicity of eigenvalues of $T_n$ within $\Gamma$ for all~$n \ge N$. 
\end{proof}

Making use of continuity of the spectrum according to Lemma \ref{Lemma Finite System of Eigenvalues}, we may derive the following result. 

\begin{Lemma}\label{Lemma Convergence of Eigenvalues}
	Let $H$ be a Hilbert space, and let $(A_n)_{n \in \N}$ be a sequence of self-adjoint compact operators in~$\mathscr{K}(H)$ such that each operator $A_n$ has at most $s$ positive and at most $s$ negative eigenvalues. If $(A_n)_{n \in \N}$ converges in norm to some $A \in \LL(H)$, then~$A$ is also a self-adjoint compact operator which has at most $s$ positive and at most $s$ negative eigenvalues. 
\end{Lemma}
\begin{proof}
	Given a sequence $(A_n)_{n \in \N}$ in $\K(H)$ with $A_n \to A$ in $\LL(H)$, the fact that~$\K(H)$ is a closed subspace of $\LL(H)$ implies that the limit~$A$ is a compact operator. Thus in view of Theorem~\ref{Theorem III-6.26} and Remark \ref{Remark III-6.27}, each non-zero eigenvalue of~$A$ is isolated and has finite multiplicity, and according to Remark \ref{Remark finite system} the non-zero eigenvalues of~$A$ form a finite system of isolated eigenvalues.
	In particular, there is a closed curve $\Gamma$ as described in Paragraph~\ref{III-S 6.4} (cf.\ \cite[III-\S 6.4]{kato}) which encloses all positive (negative) eigenvalues of $A$. 
	Now assume that $A$ has $m > s$ positive (negative) eigenvalues.
	Then Lemma~\ref{Lemma Finite System of Eigenvalues} yields the existence of some~$N \in \N$ such that the spectrum $\sigma(A_n)$ is separated by~$\Gamma$ into two parts $\sigma'(A_n)$ within $\Gamma$ and $\sigma''(A_n)$ without $\Gamma$ for all~$n \ge N$. As a consequence, $\sigma'(A_n)$ consists of $m > s$ positive (negative) eigenvalues 
	for all $n \ge N$ in contradiction to the fact that $A_n$ has at most $s$ positive and at most $s$ negative eigenvalues for all~$n \in \N$. Hence~$A \in \K(H)$ is a selfadjoint operator which has at most~$s$ positive and at most $s$ negative eigenvalues. This concludes the proof.
\end{proof}

\subsubsection{Application to Causal Fermion Systems}
After these preparations, we finally are in the position to prove Theorem \ref{Theorem Polish}.

\begin{proof}[Proof of Theorem \ref{Theorem Polish}]
	Let $(\H, \F, \rho)$ be a causal fermion system. 
	Separability of $\F$ follows from \S \ref{S Separability}.
	By virtue of Lemma~\ref{Lemma Convergence of Eigenvalues}, we conclude that~$\F \subset \LL(\H)$ is closed. Since~$\LL(\H)$ is a complete metric space with respect to the Fréchet metric induced by the operator norm, we conclude that~$\F$ is completely metrizable. Taken together, $\F$ is a separable, completely metrizable space, and thus Polish~\cite[Definition~(3.1)]{kechris}.   
\end{proof} 

More precisely, the space $(\F, d)$ is a complete metric space, where~$d$ is the Fréchet metric induced by the operator norm on~$\LL(\H)$ (cf.~\cite[\S 0.7]{alt}).

\section{Support of Locally Finite Measures on Polish Spaces}\label{Appendix Support}
In this section we derive useful 
topological properties concerning the support of locally finite measures (or Borel measures in the sense of \cite{gardner+pfeffer}) on Polish spaces (see Lemma \ref{Lemma support LCH} below). 
To begin with, let us recall the following preparatory result. 

\begin{Prp}\label{Proposition Kechris}
	Let $X$ be Polish and $\mu$ a finite measure on $\B(X)$. Then~$A \subset X$ is $\mu$-measurable if and only if there exists a $\sigma$-compact set~$F \subset A$ with $\mu(A \setminus F) = 0$. 
\end{Prp}
\begin{proof}
	See \cite[Theorem (17.11)]{kechris}. 
\end{proof}

Moreover, based on \cite[Chapter IX]{bourbaki}, a Borel measure (in the sense of \cite{gardner+pfeffer}) on a topological Hausdorff space $X$ is said to be \emph{moderated} if $X$ is the union of countably many open subsets of finite $\mu$-measure (see \cite[Chapter VIII]{elstrodt}). (Since open sets are measurable, every moderated measure is $\sigma$-finite.) We point out that, due to Ulam's theorem~\cite[Satz~VIII.1.16]{elstrodt}, every Borel measure on a Polish space is regular and moderated. (Due to Meyer's theorem, the same is true for Borel measures on Souslin spaces, see~\cite[Satz~VIII.1.17]{elstrodt}.) 
As a consequence, we may derive useful properties of the support of Borel measures on Polish spaces, as the following lemma shows. 
\begin{Lemma}\label{Lemma support LCH}
	Let $X$ be a Polish space, and assume that $\mu$ is a Borel measure on~$\B(X)$. Then $\supp \mu \subset X$ is a $\sigma$-compact topological space. Moreover, there exists a dense subset~$F \subset \supp \mu$ such that each $x \in F$ has a compact neighborhood in $\supp \mu$. 
	Whenever~$\supp \mu$ is hemicompact, one can arrange that $\supp \mu$ is a Polish space which has the Heine-Borel property. The corresponding Heine-Borel metric can be chosen locally identical to a complete metric on $X$ (in the relative topology of $\supp \mu$). 
\end{Lemma}
\begin{proof}
	We make essentially use of the fact that the measure~$\mu$ is moderated in view of Ulam's theorem. As a consequence, there is a sequence of sets $(U_n)_{n \in \N}$ with~$U_n \subset X$ open and~$\mu(U_n) < \infty$ for all $n \in \N$ such that~$\bigcup_{n \in \N} U_n = X$. From the fact that open subsets of Polish spaces are Polish (in the relative topology, see~\cite[\S 26]{bauer} or~\cite[Theorem~(3.11)]{kechris}), we conclude that each set~$U_n \subset X$ is Polish. Thus each~$\mu|_{U_n}$ is a finite Borel measure on a Polish space. Due to~\cite[Proposition~7.2.9]{bogachev}, every (finite) Borel measure on a separable metric space has support, implying that
	$$\mu|_{U_n}(U_n \setminus \supp \mu|_{U_n}) = 0 \:.$$ 
	In view of Proposition~\ref{Proposition Kechris}, we conclude that~$\supp \mu|_{U_n}$ is contained in a $\sigma$-compact set~$F_n \subset U_n$, and thus~$\supp \mu|_{U_n}$ is $\sigma$-compact for every $n \in \N$. From this we deduce that~$\bigcup_{n \in \N} \supp \mu|_{U_n}$ is $\sigma$-compact. Making use of the fact that subsets of $\sigma$-compact spaces are $\sigma$-compact, we conclude that 
	\begin{align*}
	\supp \mu \subset \bigcup_{n \in \N} \supp \mu|_{U_n} \qquad \text{is $\sigma$-compact} 
	\end{align*}
	(where $\mu$ has support in view of \cite[Lemma~VIII.2.15]{elstrodt}). Since~$\supp \mu \subset X$ is closed (see~\cite[\S VIII.2.5]{elstrodt}), the support~$\supp \mu$ is Polish and thus Baire (cf.~\S \ref{S Basic Definitions}). 
	From~\cite[25B]{willard} we conclude that there exists a dense subset~$F \subset \supp \mu$ such that each $x \in F$ has a compact neighborhood in $\supp \mu$. 
	
	Assuming that $\supp \mu$ is hemicompact (see for instance \cite[17I]{willard}), then~$\supp \mu$ is locally compact in view of \cite[Exercise~3.4.E]{engelking} and thus a complete $\sigma$-locally compact space (in the sense of~\cite{steen+seebach}). In this case, due to~\cite[Theorem~2']{williamson} and the explanations in~\cite[Section 3]{noncompact}, the space~$\supp \mu$ is metrizable by a Heine-Borel metric which is (Cauchy) locally identical to a complete metric on $X$; endowed with such a metric, the space~$\supp \mu$ has the Heine-Borel property, i.e.\ each closed bounded subset (with respect to the Heine-Borel metric) is compact. 
\end{proof}

\Thanks {{\em{Acknowledgments:}}
	C.~L.\ gratefully acknowledges support by the ``Studienstiftung des deutschen Volkes.''
	
\providecommand{\bysame}{\leavevmode\hbox to3em{\hrulefill}\thinspace}
\providecommand{\MR}{\relax\ifhmode\unskip\space\fi MR }
\providecommand{\MRhref}[2]{%
	\href{http://www.ams.org/mathscinet-getitem?mr=#1}{#2}
}

\providecommand{\href}[2]{#2}

\end{document}